\documentclass[12pt,oneside,reqno,titlepage, hyperfootnotes=false]{amsart}
\usepackage{lmodern}
\usepackage[T1]{fontenc}
\usepackage[latin9]{inputenc}
\usepackage[a4paper]{geometry}
\usepackage{color}
\usepackage{array}
\usepackage{bm}
\usepackage{amstext}
\usepackage{amsthm}
\usepackage{amssymb}
\usepackage{graphicx}
\usepackage{setspace}
\usepackage[authoryear]{natbib}
\usepackage[unicode=true,pdfusetitle,
 bookmarks=true,bookmarksnumbered=false,bookmarksopen=false,
 breaklinks=false,pdfborder={0 0 0},pdfborderstyle={},backref=false,colorlinks=true]
 {hyperref}

\makeatletter

\providecommand{\tabularnewline}{\\}

\numberwithin{equation}{section}
\numberwithin{figure}{section}

\usepackage{amsfonts}
\usepackage{amstext}
\usepackage{amsthm}

\usepackage{hyperref}
\hypersetup{
    colorlinks = true,
    citecolor = black
}

\usepackage{threeparttable}
\usepackage{graphicx}
\usepackage{enumerate}
\usepackage{listings}
\usepackage{subcaption}

\usepackage{float}

\newtheorem{prop}{Proposition}
\newtheorem{thm}{Theorem}

\newtheorem{lem}{Lemma}
\newtheorem*{lem*}{Lemma}
\newtheorem*{thm*}{Theorem}

\newtheorem*{asm1}{Assumption 1}
\newtheorem*{asm2}{Assumption 2}
\newtheorem*{asm3}{Assumption 3}
\newtheorem*{asm4}{Assumption 4}
\newtheorem*{asm5}{Assumption 5}
\newtheorem*{asm6}{Assumption 6}

\theoremstyle{definition}

\makeatother

\begin{document}
\title{Optimal tests following sequential experiments}
\author{Karun Adusumilli$^\dagger$}
\begin{abstract}
Recent years have seen tremendous advances in the theory and application
of sequential experiments. While these experiments are not always
designed with hypothesis testing in mind, researchers may still be
interested in performing tests after the experiment is completed.
The purpose of this paper is to aid in the development of optimal
tests for sequential experiments by analyzing their asymptotic properties.
Our key finding is that the asymptotic power function of any test
can be matched by a test in a limit experiment where a Gaussian process
is observed for each treatment, and inference is made for the drifts
of these processes. This result has important implications, including
a powerful sufficiency result: any candidate test only needs to rely
on a fixed set of statistics, regardless of the type of sequential
experiment. These statistics are the number of times each treatment
has been sampled by the end of the experiment, along with final value
of the score (for parametric models) or efficient influence function
(for non-parametric models) process for each treatment. We then characterize
asymptotically optimal tests under various restrictions such as unbiasedness,
$\alpha$-spending constraints etc. Finally, we apply our our results
to three key classes of sequential experiments: costly sampling, group
sequential trials, and bandit experiments, and show how optimal inference
can be conducted in these scenarios.
\end{abstract}

\thanks{\textit{This version}: \today{}\\
\thispagestyle{empty}I would like to thank Kei Hirano and Jack Porter
for insightful discussions that stimulated this research. Thanks also
to seminar participants at LSE and UCL for helpful comments.\\
$^\dagger$Department of Economics, University of Pennsylvania}
\maketitle

\section{Introduction \label{sec:Introduction}}

Recent years have seen tremendous advances in the theory and application
of sequential/adaptive experiments. Such experiments are now used
being in a wide variety of fields, ranging from online advertising
\citep{russo2017tutorial}, to dynamic pricing \citep{ferreira2018online},
drug discovery \citep{wassmer2016group}, public health \citep{athey2021shared},
and economic interventions \citep{kasy2019adaptive}. Compared to
traditional randomized trials, these experiments allow one to target
and achieve a more efficient balance of welfare, ethical, and economic
considerations. In fact, starting from the Critical Path Initiative
in 2006, the FDA has actively promoted the use of sequential designs
in clinical trials for reducing trial costs and risks for participants
\citep{guidance2018adaptive}. For instance, group-sequential designs,
wherein researchers conduct interim analyses at predetermined stages
of the experiment, are now routinely used in clinical trials. If the
analysis suggests a significant positive or negative effect from the
treatment, the trial may be stopped early. Other examples of sequential
experiments include bandit experiments \citep{lattimore2020bandit},
best-arm identification \citep{russo2016information} and costly sampling
\citep{adusumilli2022sample}, among many others.

Although hypothesis testing is not always the primary goal of sequential
experiments, one may still desire to conduct a hypothesis test after
the experiment is completed. For example, a pharmaceutical company
may conduct an adaptive trial for drug testing with the explicit goal
of maximizing welfare or minimizing costs, but may nevertheless be
required to test the null hypothesis of a zero average treatment effect
for the drug after the trial. Despite the practical importance of
such inferential methods, there are currently few results characterizing
optimal tests, or even identifying which sample statistics to use
when conducting tests after sequential experiments. This paper aims
to fill this gap. 

To this end, we follow the standard approach in econometrics and statistics
(see, e.g., \citealp[Chapter 14]{van2000asymptotic}) of studying
the properties of various candidate tests by characterizing their
power against local alternatives, also known as Pitman alternatives.
These are alternatives that converge to the null at the parametric,
i.e., $1/\sqrt{n}$ rate, leading to non-trivial asymptotic power.
Here, $n$ is typically the sample size, although it can have other
interpretations in experiments which are open-ended, see Section \ref{sec:Diffusion-asymptotics-and}
for a discussion. The main finding of this paper is that the asymptotic
power function of any test can be matched by that of a test in a limit
experiment where one observes a Gaussian process for each treatment,
and the aim is to conduct inference on the drifts of the Gaussian
processes. 

As a by-product of this equivalence, we show that the power function
of any candidate test (which may employ additional information beyond
the sufficient statistics) can be matched asymptotically by one that
only depends on a finite set of sufficient statistics. In the most
general scenario, the sufficient statistics are the number of times
each treatment has been sampled by the end of the experiment, along
with final value of the score (for parametric models) or efficient
influence function (for non-parametric models) process for each treatment.
However, even these statistics can be further reduced under additional
assumptions on the sampling and stopping rules. Our results thus show
that a substantial dimension reduction is possible, and only a few
statistics are relevant for conducting tests.

Furthermore, we characterize the optimal tests in the limit experiment.
We then show that finite sample analogues of these are asymptotically
optimal under the original sequential experiment. Our results can
also be used to compute the power envelope, i.e., an upper bound on
the asymptotic power function of any test. Although a uniformly most
powerful test in the limit experiment may not always exist, some positive
results are obtained for testing linear combinations under unbiasedness
or $\alpha$-spending restrictions. Alternatively, one may impose
less stringent criteria for optimality, like weighted average power,
and we show how to compute optimal tests under such criteria as well. 

We provide two new asymptotic representation theorems (ARTs) for formalizing
the equivalence of tests between the original and limit experiments.
The first applies to `stopping-time experiments', where the sampling
rule is fixed beforehand but the stopping rule (which describes when
the experiment is to be terminated) is fully adaptive (i.e., it can
be updated after every new observation). Our second ART allows for
the sampling rule to be adaptive as well, but we require the sampling
and stopping decision to be updated only a finite number of times,
after observing the data in batches. While constraining attention
to batched experiments is undoubtedly a limitation, practical considerations
often necessitate conducting sequential experiments in batches anyway.
Also, as shown in \citet{adusumilli2021risk}, any fully adaptive
experiment can be approximated by a batched experiment with a sufficiently
large number of batches. Our second ART builds on, and extends, the
recent work of \citet{hirano2023asymptotic} on asymptotic representations.
We refer to Sections \ref{subsec:Related-literature} and \ref{subsec:Asymptotic-representation-theorem-batched}
for a detailed comparison. 

Importantly, our framework covers both parametric and non-parametric
settings. Finally, we apply our results to three important examples
of sequential experiments: costly sampling, group sequential trials
and bandit experiments, and suggest new inferential procedures for
these experiments that are asymptotically optimal under different
scenarios.

\subsection{Related literature\label{subsec:Related-literature}}

Despite the vast amount of work on the development of sequential learning
algorithms, the literature on inference following the use of such
algorithms is relatively sparse. One approach gaining some popularity
in computer-science is called `any-time inference'. Here, one seeks
to construct tests and confidence intervals that are correctly sized
no matter how, or when, the experiment is stopped. We refer to \citet{ramdas2022game}
for a survey and to \citet{grunwald2020safe}, \citet{howard2021time},
\citet{johari2022always} for some recent contributions. The uniform-in-time
size constraint is a stringent requirement, and this comes at the
expense of lower power than could be achieved otherwise. By contrast,
our focus in this paper is on classical notions of testing, where
size control is only achieved when the experimental protocol, i.e.,
the specific sampling rule and stopping time, is followed exactly.
In essence, this requires the decision maker to pre-register the experiment
and fully commit to the protocol. We believe this is valid assumption
in most applications; adaptive experiments are usually constructed
with the explicit goal of welfare maximization, so there is little
incentive to deviate from the protocol as long as the preferences
of the experimenter and the end-user of the experiment are aligned
(e.g., in the case of online marketplaces they would be the same entity).
In other situations, pre-registration of the experimental design is
usually mandatory, see, e.g., the FDA guidance on sequential designs
\citep{guidance2018adaptive}. 

There are other recent papers which propose inferential methods under
the `classical' hypothesis-testing framework. \citet{zhang2020inference}
and \citet{hadad2021confidence} suggest asymptotically normal tests
for some specific classes of sequential experiments. These tests are
based on re-weighing the observations. There are also a number of
methods for group sequential and linear boundary designs commonly
used in clinical trials, see \citet{hall2013analysis} for a review.
However, it is not clear if any of them are optimal even within their
specific use cases. 

Finally, in prior and closely related work to our own, \citet{hirano2023asymptotic}
obtain an Asymptotic Representation Theorem (ART) for batched sequential
experiments that is different from ours and apply this to testing.
The ART of \citet{hirano2023asymptotic} is a lot more general than
our own, e.g., it can be used to determine optimal conditional tests
given outcomes from previous stages. However, this generality comes
at a price as the state variables increase linearly with the number
of batches. Here, we build on and extend these results to show that
only a fixed number of sufficient statistics are needed to match the
\textit{unconditional} asymptotic power of any test, irrespective
of the number of batches (our results also apply to asymptotic power
\textit{conditional} on stopping times). We also derive a number of
additional results that are new to this literature: First, our ART
for stopping-time experiments applies to fully adaptive experiments
(this result is not based on \citealp{hirano2023asymptotic}; rather,
it makes use of a representation theorem for stopping times due to
\citealp{LeCam1979}). Second, our analysis covers non-parametric
models, which is important for applications. Third, we characterize
the properties of optimal tests in a number of different scenarios,
e.g., for testing linear combinations of parameters, or under unbiased
and $\alpha$-spending requirements. This is useful as UMP tests do
not generally exist otherwise.

As noted earlier, this paper employs the local asymptotic power criterion
to rank tests. This criterion naturally leads to `diffusion asymptotics',
where the limit experiment consists of Gaussian diffusions. Diffusion
asymptotics were first introduced by \citet{wager2021diffusion} and
\citet{fan2021diffusion} to study the properties of a class of sequential
algorithms. In previous work \citep{adusumilli2021risk}, this author
demonstrated some asymptotic equivalence results for comparing the
Bayes and minimax risk of bandit experiments. Here, we apply the techniques
devised in those papers to study inference.

\subsection{Examples\label{subsec:Examples}}

Before describing our procedures, it can be instructive to consider
some examples of sequential experiments. 

\subsubsection{Costly sampling}

Consider a sequential experiment in which sampling is costly, and
the aim is to select the best of two possible treatments. Previous
work by this author (\citealp{adusumilli2022sample}) showed that
the minimax optimal strategy in this setting involves a fixed sampling
rule (the Neyman allocation) and stopping when the average difference
in treatment outcomes multiplied by the number of observations exceeds
a specific threshold. In fact, the stopping rule here has the same
form as the SPRT procedure of \citet{wald1947sequential}, even though
the latter is motivated by very different considerations. SPRT is
itself a special case of `fully sequential linear boundary designs',
as discussed, e.g., in \citet{whitehead1997design}. Typically these
procedures recommend sampling the two treatments in equal proportions
instead of the Neyman allocation. In Section \ref{sec:Applications},
we show that for `horizontal fully sequential boundary designs' with
any fixed sampling rule (including, but not restricted to, the Neyman
allocation), the most powerful unbiased test for treatment effects
depends only on the stopping time and rejects when it is below a specific
threshold. 

\subsubsection{Group sequential trials}

In many applications, it is not feasible to employ continuous-time
monitoring designs that update the decision rule after each observation.
Instead, one may wish to stop the experiment only at a limited number
of pre-specified times. Such designs are known as group-sequential
trials, see \citet{wassmer2016group} for a textbook treatment. Recently,
these experiments have become very popular for conducting clinical
trials; they have been used, e.g., to test the efficacy of Coronavirus
vaccines \citep{zaks2020phase}. While a number of methods have been
proposed for inference following these experiments, as reviewed, e.g.,
in \citet{hall2013analysis}, it is not clear which, if any, are optimal.
In Section \ref{sec:Applications}, we derive optimal non-parametric
tests and confidence intervals for such designs under an $\alpha$-spending
size criterion (see, Section \ref{subsec:Conditional-inference}). 

\subsubsection{Bandit experiments}

In the previous two examples, the decision maker could choose when
to end the experiment, but the sampling strategy was fixed beforehand.
In many experiments however, the sampling rule can also be modified
based on the information revealed from past data. Bandit experiments
are a canonical example of these. Previously, \citet{hirano2023asymptotic}
derived asymptotic power envelopes for any test following batched
parametric bandit experiments. In this paper, we refine the results
of \citet{hirano2023asymptotic} further by showing that only a finite
number of sufficient statistics are needed for testing, irrespective
of the number of batches. Our results apply to non-parametric models
as well.

\section{Optimal tests in experiments involving stopping times\label{sec:Diffusion-asymptotics-and}}

In this section we study the asymptotic properties of tests for parametric
stopping-time experiments, i.e., sequential experiments that involve
a pre-determined stopping time. 

\subsection{Setup and assumptions}

Consider a decision-maker (DM) who wishes to conduct an experiment
involving some outcome variable $Y$. Before starting the experiment,
the DM registers a stopping time, $\hat{\tau}$, that describes the
eventual sample size in multiples of $n$ observations (see below
for the interpretation of $n$). The choice of $\hat{\tau}$ may involve
a balancing a number of considerations such as costs, ethics, welfare
etc. Here, we abstract away from these issues and take $\hat{\tau}$
as given. In the course of the experiment, the DM observes a sequence
of outcomes $Y_{1},Y_{2},\dots$ . The experiment ends in accordance
with $\hat{\tau}$, which we assume to be adapted to the filtration
generated by the outcome observations. Let $P_{\theta}$ denote a
parametric model for the outcomes. Our interest in this section is
in testing $H_{0}:\theta=\Theta_{0}$ vs $H_{1}:\theta\in\Theta_{1}$
where $\Theta_{0}\cap\Theta_{1}=\emptyset$. Let $\theta_{0}\in\Theta_{0}$
denote some reference parameter in the null set.

There are two notions of asymptotics one could employ in this setting,
and consequently, two different interpretations of $n$. In many settings,
e.g., group sequential trials, there is a limit on the maximum number
of observations that can be collected; this limit is pre-specified
and we take it to be $n$. Consequently, in these experiments, $\hat{\tau}\in[0,1]$.
Alternatively, we may have open-ended experiments where the stopping
time is determined by balancing the benefit of experimentation with
the cost for sampling each additional unit of observation. In this
case, we employ small-cost asymptotics and $n$ then indexes the rate
at which the sampling costs go to $0$ (alternatively, we can relate
$n$ to the population size in the implementation phase following
the experiment, see \citealp{adusumilli2022sample}). The results
in this section apply to both asymptotic regimes.

Let $\varphi_{n}\in[0,1]$ denote a candidate test. It is required
to be measurable with respect to $\sigma\{Y_{1},\dots,Y_{\left\lfloor n\tau\right\rfloor }\}$.
Now, it is fairly straightforward to construct tests that have power
1 against any fixed alternative as $n\to\infty$. Consequently, to
obtain a more fine-grained characterization of tests, we consider
their performance against local perturbations of the form $\{\theta_{0}+h/\sqrt{n};h\in\mathbb{R}^{d}\}$.
Denote $P_{h}:=P_{\theta_{0}^ {}+h/\sqrt{n}}$ and let $\mathbb{E}_{h}^{(a)}[\cdot]$
denote its corresponding expectation. Also, let $\nu$ denote a dominating
measure for $\{P_{\theta}:\theta\in\mathbb{R}\}$, and set $p_{\theta}:=dP_{\theta}/d\nu$.
We impose the following regularity conditions on the family $P_{\theta}$,
and the stopping time $\hat{\tau}$: 

\begin{asm1} The class $\{P_{\theta}:\theta\in\mathbb{R}^{d}\}$
is differentiable in quadratic mean around $\theta_{0}$, i.e., there
exists a score function $\psi(\cdot)$ such that for each $h\in\mathbb{R}^{d},$
\begin{equation}
\int\left[\sqrt{p_{\theta_{0}+h}}-\sqrt{p_{\theta_{0}}}-\frac{1}{2}h^{\intercal}\psi\sqrt{p_{\theta_{0}}}\right]^{2}d\nu=o(\vert h\vert^{2}).\label{eq:QMD}
\end{equation}
\end{asm1}

\begin{asm2} There exists $T<\infty$ independent of $n$ such that
$\hat{\tau}\le T$.\end{asm2}

Both assumptions are fairly innocuous. As noted previously, in many
examples we already have $\tau\le1$.

Let $P_{nt,h}$ denote the joint probability measure over the iid
sequence of outcomes $Y_{1},\dots,Y_{nt}$ and take $\mathbb{E}_{nt,\bm{h}}[\cdot]$
to be its corresponding expectation. Define the (standardized) score
process $x_{n}(t)$ as
\[
x_{n}(t)=\frac{I^{-1/2}}{\sqrt{n}}\sum_{i=1}^{\left\lfloor nt\right\rfloor }\psi(Y_{i}),
\]
where $I:=\mathbb{E}_{0}[\psi(Y_{i})\psi(Y_{i})^{\intercal}]$ is
the information matrix. It is well known, see e.g., \citet[Chapter 7]{van2000asymptotic},
that quadratic mean differentiability implies $\mathbb{E}_{nT,0}[\psi(Y_{i})]=0$
and that $I$ exists. Then, by a functional central limit theorem,
\begin{equation}
x_{n}(\cdot)\xrightarrow[P_{nT,0}]{d}x(\cdot);\ x(\cdot)\sim W(\cdot).\label{eq:Convergence of score process}
\end{equation}
Here, and in what follows, $W(\cdot)$ denotes the standard $d$-dimensional
Brownian motion. Assumption 1 also implies the important property
of Sequential Local Asymptotic Normality (SLAN; \citealp{adusumilli2021risk}):
for any given $h\in\mathbb{R}^{d}$, 
\begin{equation}
\sum_{i=1}^{\left\lfloor nt\right\rfloor }\ln\frac{dp_{\theta_{0}+h/\sqrt{n}}}{dp_{\theta_{0}}}=h^{\intercal}I^{1/2}x_{n}(t)-\frac{t}{2}h^{\intercal}Ih+o_{P_{nT,0}}(1),\ \textrm{uniformly over }t\le T.\label{eq:SLAN property}
\end{equation}
The above states that the likelihood ratio admits a quadratic approximation
uniformly over all $t$. 

\subsection{Asymptotic representation theorem\label{subsec:Asymptotic-representation-theore}}

In what follows, take $U$ to be a $\textrm{Uniform}[0,1]$ random
variable that is independent of the process $x(\cdot)$, and define
$\mathcal{F}_{t}:=\sigma\{x(s),U;s\le t\}$ to be the filtration generated
by $U$ and the stochastic process $x(\cdot)$ until time $t$. 

Consider a limit experiment where one observes $U$ and a Gaussian
diffusion $x(t):=I^{1/2}ht+W(t)$ with some unknown $h$, and constructs
a test statistic $\varphi$ based on knowledge only of (i) an $\mathcal{F}_{t}$-adapted
stopping time $\tau$ that is the limiting version of $\hat{\tau}$
(in a sense made precise below); and (ii) the stopped process $x(\tau)$.
Let $\mathbb{P}_{h}$ denote the induced probability over the sample
paths of $x(\cdot)$ given $h$, and $\mathbb{E}_{h}[\cdot]$ its
corresponding expectation. The following theorem relates the original
testing problem to the one in such a limit experiment:

\begin{thm} \label{Thm: ART}Suppose Assumptions 1 and 2 hold. Let
$\varphi_{n}$ be some test function defined on the sample space $Y_{1},\dots,Y_{n\hat{\tau}}$,
and $\beta_{n}(h)$, its power against $P_{nT,h}$. Then, for every
sequence $\{n_{j}\}$, there is a further sub-sequence $\{n_{j_{m}}\}$
such that: \\
(i) \citep{LeCam1979} There exists an $\mathcal{F}_{t}$-adapted
stopping time $\tau$ for which $(\hat{\tau},x_{n}(\hat{\tau}))\xrightarrow[P_{nT,0}]{d}(\tau,x(\tau))$
on this sub-sequence.\\
(ii) There exists a test $\varphi$ in the limit experiment depending
only on $\tau,x(\tau)$ such that $\beta_{n_{j_{m}}}(h)\to\beta(h)$
for every $h\in\mathbb{R}^{d}$, where $\beta(h):=\mathbb{E}_{h}[\varphi(\tau,x(\tau))]$
is the power of $\varphi$ in the limit experiment.\end{thm} 

The first part of Theorem \ref{Thm: ART} is essentially due to \citet{LeCam1979}. 

To the best of our knowledge, the second part of Theorem \ref{Thm: ART}
is new. Previously, \citet{LeCam1979} showed that for $\{P_{\theta}\}$
in the exponential family of distributions, 
\[
\ln\frac{dP_{n\hat{\tau}.,h}}{dP_{n\hat{\tau},0}}\left({\bf y}_{n\hat{\tau}}\right)\xrightarrow[P_{nT,0}]{d}h^{\intercal}I^{1/2}x(\tau)-\frac{\tau}{2}h^{\intercal}Ih.
\]
Here, we extend the above to general families of distributions satisfying
Assumption 1. We then derive an asymptotic representation theorem
for $\varphi_{n}$ as a consequence of this result.

Note that in the second part of Theorem \ref{Thm: ART}, $\tau$ is
taken as given (this mirrors how $\hat{\tau}$ is taken as given in
the context of the original experiment). It is chosen so that the
first part of the theorem is satisfied. In order to derive optimal
tests, one would need to know the joint distribution of $\tau,x(\tau)$.
Unfortunately, the first part of Theorem \ref{Thm: ART} does not
provide a characterization of $\tau$; it only asserts that such a
stopping time must exist. Fortunately, in practice, most stopping
times are functions, $\hat{\tau}=\tau(x_{n}(\cdot))$, of the score
process, e.g., the optimal stopping time under costly sampling is
given by $\hat{\tau}=\inf\{t:\vert x_{n}(t)\vert\ge\gamma\}$. Indeed,
previous work by this author (\citealp{adusumilli2022sample}) and
others has shown that if the stopping time is to be chosen according
some notion of Bayes or minimax risk, then it is sufficient to restrict
attention to stopping times that depend only on $x_{n}(\cdot)$. In
such cases, the continuous mapping theorem allows us to determine
$\tau$ as $\tau=\tau(x(\cdot))$. 

\subsection{Characterization of optimal tests in the limit experiment\label{subsec:Characterization-of-optimal}}

\subsubsection{Testing a parameter vector}

The simplest hypothesis testing problem in the limit experiment concerns
testing $H_{0}:h=0$ vs $H_{1}:h=h_{1}$. By the Neyman-Pearson lemma,
the uniformly most powerful (UMP) test is 
\[
\varphi_{h_{1}}^{*}=\mathbb{I}\left\{ h_{1}^{\intercal}I^{1/2}x(\tau)-\frac{\tau}{2}h_{1}^{\intercal}Ih_{1}\ge\gamma_{h_{1}}\right\} ,
\]
where $\gamma_{h_{1}}\in\mathbb{R}$ is chosen by the size requirement.
Let $\beta^{*}(h_{1})$ denote the power function of $\varphi_{h_{1}}^{*}$.
Then, by Theorem \ref{Thm: ART}, $\beta^{*}(\cdot)$ is an upper
bound on the limiting power function of any test of $H_{0}:\theta=\theta_{0}$. 

\subsubsection{Testing linear combinations\label{subsec:Testing-linear-combinations}}

We now consider tests of linear combinations of $h$, i.e., $H_{0}:a^{\intercal}h=0$,
in the limit experiment. In this case, a further dimension reduction
is possible if the stopping time is also dependent on a reduced set
of statistics. 

Define $\sigma^{2}:=a^{\intercal}I^{-1}a$, $\tilde{x}(t):=\sigma^{-1}a^{\intercal}I^{-1/2}x(t)$,
let $U_{1}$ denote a $\textrm{Uniform}[0,1]$ random variable independent
of $\tilde{x}(\cdot)$, and take $\tilde{\mathcal{F}}_{t}$ to be
the filtration generated by $\sigma\{U_{1},\tilde{x}(s):s\le t\}$.
Note that $\tilde{x}(\cdot)\sim W(\cdot)$ under the null; hence,
it is pivotal. 

\begin{prop} \label{Prop: Linear combinations} Suppose that the
stopping time $\tau$ in Theorem \ref{Thm: ART} is $\tilde{\mathcal{F}}_{t}$-adapted.
Then, the UMP test of $H_{0}:a^{\intercal}h=0$ vs $H_{1}:a^{\intercal}h=c$
in the limit experiment is
\[
\varphi_{c}^{*}(\tau,\tilde{x}(\tau))=\mathbb{I}\left\{ c\tilde{x}(\tau)-\frac{c^{2}}{2\sigma}\tau\ge\gamma_{c}\right\} .
\]
In addition, suppose Assumptions 1 and 2 hold, let $\beta^{*}(c)$
denote the power of $\varphi_{c}^{*}$ for a given $c$, and $\beta_{n}(h)$
the power of some test, $\varphi_{n}$, of $H_{0}:a^{\intercal}\theta=0$
in the original experiment against local alternatives $\theta\equiv\theta_{0}+h/\sqrt{n}$
. Then, for each $h\in\mathbb{R}^{d}$ , $\lim_{n\to\infty}\beta_{n}(h)\le\beta^{*}(a^{\intercal}h)$.\end{prop} 

The above result suggests that $\tilde{x}(\tau)$ and $\tau$ are
sufficient statistics for the optimal test. An important caveat, however,
is that the class of stopping times are further constrained to only
depend on $\tilde{x}(t)$ in the limit. In practice, this would happen
if the stopping time $\hat{\tau}$ in the original experiment is a
function only of $\hat{\tilde{x}}_{n}(\cdot):=\sigma^{-1}a^{\intercal}I^{-1/2}x_{n}(\cdot)$.
Fortunately, this is the case in a number of examples. 

It is straightforward to show that the same power envelope, $\beta^{*}(\cdot)$,
also applies to tests of the composite hypothesis $H_{0}:a^{\intercal}\theta\le0$. 

\subsubsection{Unbiased tests}

A test is said to be unbiased if its power is greater than size under
all alternatives. The following result describes a useful property
of unbiased tests in the limit experiment: 

\begin{prop} \label{Prop: Unbiased} Any unbiased test of $H_{0}:h=0$
vs $H_{1}:h\neq0$ in the limit experiment must satisfy $\mathbb{E}_{0}[x(\tau)\varphi(\tau,x(\tau))]=0$.
\end{prop} 

See Section \ref{subsec:Sequential-linear-boundary} for an application
of the above result.

\subsubsection{Weighted average power}

Suppose we specify a weight function, $w(\cdot)$, over alternatives
$h\neq0$. Then, the test of $H_{0}:h=0$ in the limit experiment
that maximizes weighted average power is given by
\[
\varphi_{w}^{*}(\tau,x(\tau))=\mathbb{I}\left\{ \int e^{h^{\intercal}I^{1/2}x(\tau)-\frac{\tau}{2}h^{\intercal}Ih}dw(h)\ge\gamma\right\} .
\]
The value of $\gamma$ is determined by the size requirement. 

\subsection{Alpha-spending criterion\label{subsec:Conditional-inference}}

In this section, we study inference under a stronger version of the
size constraint, inspired by the $\bm{\alpha}$-spending approach
in group sequential trials (\citealp{gordon1983discrete}). Suppose
that the stopping time is discrete, taking only the values $t=1,2,\dots,T$.
Then, instead of an overall size constraint of the form $\mathbb{E}_{nT,\bm{0}}[\varphi_{n}]\le\alpha$,
we may specify a `spending-vector' $\bm{\alpha}:=(\alpha_{1},\dots,\alpha_{T})$
satisfying $\sum_{t=1}^{T}\alpha_{t}=\alpha$, and require
\begin{equation}
\mathbb{E}_{nT,\bm{0}}[\mathbb{I}\{\hat{\tau}=t\}\varphi_{n}]\le\alpha_{t}\ \forall\ t.\label{eq:alpha spending}
\end{equation}
In what follows, we call a test, $\varphi_{n}$, satisfying (\ref{eq:alpha spending})
a level-$\bm{\alpha}$ test (with a boldface $\bm{\alpha}$). Intuitively,
if each $t$ corresponds to a different stage of the experiment, the
$\bm{\alpha}$-spending constraint prescribes the maximum amount of
Type-I error that may be expended at stage $t$. As a practical matter,
it enables us to characterize a UMP or UMP unbiased test in settings
where such tests do not otherwise exist. We also envision the criterion
as a useful conceptual device: even if we are ultimately interested
in a standard level-$\alpha$ test, we can obtain this by optimizing
a chosen power criterion (average power, etc.) over the spending vectors
$\bm{\alpha}:=(\alpha_{1},\dots,\alpha_{K})$ satisfying $\sum_{k}\alpha_{k}\le\alpha$. 

A particularly interesting example of an $\bm{\alpha}$-spending vector
is $(\alpha P_{nT,0}(\hat{\tau}=1),\dots,\alpha P_{nT,0}(\hat{\tau}=k))$;
this corresponds to the requirement that $\mathbb{E}_{nT,\bm{0}}\left[\left.\varphi_{n}\right|\hat{\tau}=t\right]\le\alpha$
for all $t$, i.e., the test be conditionally level-$\alpha$ given
any realization of the stopping time. This may have some intuitive
appeal, though it does disregard any information provided by the stopping
time for discriminating between the hypotheses. 

Under the $\bm{\alpha}$-spending constraint, a test that maximizes
expected power also maximizes expected power conditional on each realization
of stopping time. This is a simple consequence of the law of iterated
expectations. Consequently, we focus on conditional power in this
section. Our main result here is a generalization of Theorem \ref{Thm: ART}
to $\bm{\alpha}$-spending restrictions. The limit experiment is the
same as in Section \ref{subsec:Asymptotic-representation-theore}.

\begin{thm} \label{Thm: ART-conditional inference}Suppose Assumptions
1, 2 hold, and the stopping times are discrete, taking only the values
$1,2,\dots,T$. Let $\varphi_{n}$ be some level-$\bm{\alpha}$ test
defined on the sample space $Y_{1},\dots,Y_{n\hat{\tau}}$, and $\beta_{n}(h\vert t)$,
its conditional power against $P_{nT,h}$ given $\hat{\tau}=t$. Then,
there exists a level-$\bm{\alpha}$ test, $\varphi(\cdot)$, in the
limit experiment depending only on $\tau,x(\tau)$ such that, for
every $h\in\mathbb{R}^{d}$ and $t\in\{1,2,\dots,T\}$ for which $\mathbb{P}_{0}(\tau=t)\neq0$,
$\beta_{n}(h\vert t)$ converges to $\beta(h\vert t)$ on subsequences,
where $\beta(h\vert t):=\mathbb{E}_{h}[\varphi(\tau,x(\tau))\vert\tau=t]$
is the conditional power of $\varphi(\cdot)$ in the limit experiment.\end{thm} 

It may be possible to extend the above result to continuous stopping
times using Le Cam's discretization device, though we do not take
this up here. 

\subsubsection{Power envelope}

By the Neyman-Pearson lemma, the uniformly most powerful level-$\bm{\alpha}$
(UMP-$\bm{\alpha}$) test of $H_{0}:h=0$ vs $H_{1}:h=h_{1}$ in the
limit experiment is given by
\[
\varphi_{h_{1}}^{*}(t,x(t))=\begin{cases}
1 & \textrm{if }\mathbb{P}_{0}(\tau=t)\le\alpha_{t}\\
\mathbb{I}\left\{ h_{1}^{\intercal}I^{1/2}x(t)\ge\gamma(t)\right\}  & \textrm{if }\mathbb{P}_{0}(\tau=t)>\alpha_{t}
\end{cases}.
\]
Here, $\gamma(t)\in\mathbb{R}$ is chosen by the $\bm{\alpha}$-spending
requirement that $\mathbb{E}_{0}[\varphi_{h_{1}}^{*}(\tau,x(\tau))\vert\tau=t]\le\alpha_{t}/\mathbb{P}_{0}(\tau=t)$
for each $t$. If we take $\beta^{*}(h_{1}\vert t)$ to be the power
function of $\varphi_{h_{1}}^{*}(\cdot)$, Theorem \ref{Thm: ART-conditional inference}
implies $\beta^{*}(\cdot\vert t)$ is an upper bound on the limiting
conditional power function of any level-$\bm{\alpha}$ test of $H_{0}:\theta=\theta_{0}$. 

\subsubsection{Testing linear combinations}

A stronger result is possible for tests of linear combinations of
$\theta$. Recall the definitions of $\tilde{x}(t)$ and $\tilde{\mathcal{F}_{t}}$
from Section \ref{subsec:Testing-linear-combinations}. If the limiting
stopping time is $\tilde{\mathcal{F}_{t}}$ -adapted, we have, as
in Proposition \ref{Prop: Linear combinations}, that the sufficient
statistics are only $\tilde{x}(\tau),\tau$, and the UMP-$\bm{\alpha}$
test of $H_{0}:a^{\intercal}h=0$ vs $H_{1}:a^{\intercal}h=c\ (>0)$
in the limit experiment is
\[
\breve{\varphi}^{*}(t,\tilde{x}(t))=\begin{cases}
1 & \textrm{if }\mathbb{P}_{0}(\tau=t)\le\alpha_{t}\\
\mathbb{I}\left\{ c\tilde{x}(t)\ge\gamma_{c}(t)\right\} \equiv\mathbb{I}\left\{ \tilde{x}(t)\ge\tilde{\gamma}(t)\right\}  & \textrm{if }\mathbb{P}_{0}(\tau=t)>\alpha_{t}
\end{cases}.
\]
Here, $\tilde{\gamma}(t)$ is chosen such that $\mathbb{E}_{0}[\breve{\varphi}^{*}(\tau,\tilde{x}(\tau))\vert\tau=t]=\alpha_{t}/\mathbb{P}_{0}(\tau=t)$.
Clearly, $\tilde{\gamma}(t)$ it is independent of $c$ for $c>0$.
Since $\breve{\varphi}^{*}(\cdot)$ is thereby also independent of
$c$ for $c>0$, we conclude that it is UMP-$\bm{\alpha}$ for testing
the composite one-sided alternative $H_{0}:a^{\intercal}h=0$ vs $H_{1}:a^{\intercal}h>0$.
Thus, a UMP-$\bm{\alpha}$ test exists in this scenario even as a
UMP test doesn't. What is more, by Theorem \ref{Thm: ART-conditional inference},
the conditional power function, $\breve{\beta}^{*}(c\vert t)$, of
$\breve{\varphi}^{*}(\cdot)$ is an asymptotic upper bound on the
conditional power of any level-$\bm{\alpha}$ test, $\varphi_{n}$,
of $H_{0}:a^{\intercal}\theta=0$ vs $H_{1}:a^{\intercal}\theta>0$
in the original experiment against local alternatives $\theta\equiv\theta_{0}+h/\sqrt{n}$
satisfying $a^{\intercal}\theta=c/\sqrt{n}$. 

\subsubsection{Conditionally unbiased tests\label{subsec:Conditionally-unbiased-tests}}

We call a test conditionally unbiased if it is unbiased conditional
on any possible realization of the stopping time. In analogy with
Proposition \ref{Prop: Unbiased}, a necessary condition for $\varphi(\cdot)$
being conditionally unbiased in the limit experiment is that
\begin{equation}
\mathbb{E}_{0}\left[x(\tau)\left(\varphi(\tau,x(\tau))-\alpha\right)\vert\tau=t\right]=0\ \forall\ t.\label{eq:Conditional unbiasedness definition}
\end{equation}
Then, by a similar argument as in \citet[Section 4.2]{lehmann2005testing},
the UMP conditionally unbiased (level-$\bm{\alpha}$) test of $H_{0}:a^{\intercal}h=0$
vs $H_{1}:a^{\intercal}h\neq0$ in the limit experiment can be shown
to be
\[
\bar{\varphi}^{*}(t,\tilde{x}(t))=\begin{cases}
1 & \textrm{if }\mathbb{P}_{0}(\tau=t)\le\alpha_{t}\\
\mathbb{I}\left\{ \tilde{x}(t)\notin\left[\gamma_{L}(t),\gamma_{U}(t)\right]\right\}  & \textrm{if }\mathbb{P}_{0}(\tau=t)>\alpha_{t}
\end{cases}.
\]
The quantities $\gamma_{L}(t),\gamma_{U}(t)$ are chosen to satisfy
both (\ref{eq:alpha spending}) and (\ref{eq:Conditional unbiasedness definition}).
In practice, this requires simulating the distribution of $\tilde{x}(\tau)$
given $\tau=t$. Also, $\gamma_{L}(\cdot)=-\gamma_{U}(\cdot)$ if
the distribution of $\tilde{x}(\tau)$ given $\tau=t$ is symmetric
around 0 under the null. 

\subsection{On the choice of $\theta_{0}$ and employing a drifting null}

Earlier in this section, we took $\theta_{0}\in\Theta_{0}$ to be
some reference parameter in the null set. However, such a choice may
result in the limiting stopping time, $\tau$, collapsing to $0$.
Consider, for example, the case of costly sampling (Example 1 in Section
\ref{subsec:Examples}). In this experiment, the stopping time, $\hat{\tau}$,
is itself chosen around a reference parameter $\theta_{0}$ (typically
chosen so that the effect of interest is $0$ at $\theta_{0}$). But
suppose we are interested in testing $H_{0}:\theta=\bar{\theta}_{0}$,
for some $\bar{\theta}_{0}\neq\theta_{0}$. Under this null, $\hat{\tau}$
converges to $0$ in probability as $\bar{\theta}_{0}$ is a fixed
distance away from $\theta_{0}$. This issue with the stopping time
arises because the null hypothesis and the stopping time are not centered
around the same reference parameter. 

One way to still provide inference in such settings is to set the
reference parameter to $\theta_{0}$, but employ a drifting null $H_{0}:h=h_{0}/\sqrt{n}$,
where $h_{0}$ is taken to be fixed over $n$, and is calibrated as
$\sqrt{n}(\bar{\theta}-\theta_{0})$. The null, $H_{0}$, thus changes
with $n$, but for the observed sample size we are still testing $\theta=\bar{\theta}_{0}$.
It is then straightforward to show that Theorems \ref{Thm: ART} and
\ref{Thm: ART-conditional inference} continue to apply in this setting;
asymptotically, the inference problem is equivalent to testing that
the drift of $x(\cdot)$ is $I^{1/2}h_{0}$ in the limit experiment.
The asymptotic approximation is expected to be more accurate the closer
$\bar{\theta}_{0}$ is to $\theta_{0}$; but for distant values of
$\bar{\theta}_{0}$, we caution that local asymptotics may not provide
a good approximation. 

\subsection{Attaining the bound\label{subsec:Attaining-the-bound}}

So far we have described upper bounds on the asymptotic power functions
of tests. Now, given a UMP test, $\varphi^{*}(\tau,x(\tau))$, in
the limit experiment, we can construct a finite sample version of
this, $\varphi_{n}^{*}:=\varphi^{*}(\hat{\tau},x_{n}(\hat{\tau}))$,
by replacing $\tau,x(\tau)$ with $\hat{\tau},x_{n}(\hat{\tau})$.
Since $x_{n}(\hat{\tau})$ depends on the information matrix, $I$,
one would need to either calibrate it to $I(\theta_{0})$ (if $\theta_{0}$
is known), or replace it with a consistent estimate. We discuss variance
estimators in Appendix \ref{subsec:Variance-estimators}. 

The test, $\varphi_{n}^{*}$, would then be asymptotically optimal,
in the sense of attaining the power envelope, under mild assumptions.
In particular, we only require that $\varphi^{*}(\cdot,\cdot)$ satisfy
the conditions for an extended continuous mapping theorem. Together
with (\ref{eq:SLAN property}) and the first part of Theorem \ref{Thm: ART},
this implies
\[
\left(\begin{array}{c}
\varphi^{*}(\hat{\tau},x_{n}(\hat{\tau}))\\
\sum_{i=1}^{\left\lfloor n\hat{\tau}\right\rfloor }\ln\frac{dp_{\theta_{0}+h/\sqrt{n}}}{dp_{\theta_{0}}}(Y_{i})
\end{array}\right)\xrightarrow[P_{nT,0}]{d}\left(\begin{array}{c}
\varphi^{*}(\tau,x(\tau))\\
h^{\intercal}I^{1/2}x(\tau)-\frac{\tau}{2}h^{\intercal}Ih
\end{array}\right),
\]
for any $h\in\mathbb{R}^{d}$. Then, a similar argument as in the
proof of Theorem \ref{Thm: ART} shows that the local power of $\varphi_{n}^{*}$
converges to that of $\varphi^{*}$ in the limit experiment. 

\section{Testing in non-parametric settings\label{sec:Testing-in-non-parametric}}

We now turn to the setting where the distribution of outcomes is non-parametric.
Let $\mathcal{P}$ denote a candidate class of probability measures
for the outcome $Y$, with bounded variance, and dominated by some
measure $\nu$. We are interested in conducting inference on some
regular functional, $\mu:=\mu(P)$, of the unknown data distribution
$P\in\mathcal{P}$. We assume for simplicity that $\mu$ is scalar.
Let $P_{0}\in\mathcal{P}$ denote some reference probability distribution
on the boundary of the null hypothesis so that $\mu(P_{0})=0$. Following
\citet[Section 25.6]{van2000asymptotic}, we consider the power of
tests against smooth one-dimensional sub-models of the form $\{P_{s,h}:s\le\eta\}$
for some $\eta>0$, where $h(\cdot)$ is a measurable function satisfying
\begin{equation}
\int\left[\frac{dP_{s,h}^{1/2}-dP_{0}^{1/2}}{s}-\frac{1}{2}hdP_{0}^{1/2}\right]^{2}d\nu\to0\ \textrm{as}\ s\to0.\label{eq:qmd non-parametrics}
\end{equation}

By \citet{van2000asymptotic}, (\ref{eq:qmd non-parametrics}) implies
$\int hdP_{0}=0$ and $\int h^{2}dP_{0}<\infty$. The set of all such
candidate $h$ is termed the tangent space $T(P_{0})$. This is a
subset of the Hilbert space $L^{2}(P_{0})$, endowed with the inner
product $\left\langle f,g\right\rangle =\mathbb{E}_{P_{0}}[fg]$ and
norm $\left\Vert f\right\Vert =\mathbb{E}_{P_{0}}[f^{2}]^{1/2}$.
For any $h\in T(P_{0})$, let $P_{nT,h}$ denote the joint probability
measure over $Y_{1},\dots,Y_{nT}$, when each $Y_{i}$ is an iid draw
from $P_{1/\sqrt{n},h}$. Also, take $\mathbb{E}_{nT,h}[\cdot]$ to
be its corresponding expectation. An important implication of (\ref{eq:qmd non-parametrics})
is the SLAN property that for all $h\in T(P_{0})$,
\begin{align}
\sum_{i=1}^{\left\lfloor nt\right\rfloor }\ln\frac{dP_{1/\sqrt{n},h}}{dP_{0}}(Y_{i}) & =\frac{1}{\sqrt{n}}\sum_{i=1}^{\left\lfloor nt\right\rfloor }h(Y_{i})-\frac{t}{2}\left\Vert h\right\Vert ^{2}+o_{P_{nT,0}}(1),\ \textrm{ uniformly over }t.\label{eq:SLAN nonparametric setting}
\end{align}
See \citet[Lemma 2]{adusumilli2021risk} for the proof. 

Let $\psi\in T(P_{0})$ denote the efficient influence function corresponding
to estimation of $\mu$, in the sense that for any $h\in T(P_{0})$,
\begin{equation}
\frac{\mu(P_{s,h})-\mu(P_{0})}{s}-\left\langle \psi,h\right\rangle =o(s).\label{eq:influence function}
\end{equation}
Denote $\sigma^{2}=\mathbb{E}_{P_{0}}[\psi^{2}]$. The analogue of
the score process in the non-parametric setting is the efficient influence
function process 
\[
x_{n}(t):=\frac{\sigma^{-1}}{\sqrt{n}}\sum_{i=1}^{\left\lfloor nt\right\rfloor }\psi(Y_{i}).
\]

At a high level, the theory for inference in non-parametric settings
is closely related to that for testing linear combinations in parametric
models (see, Section \ref{subsec:Characterization-of-optimal}). It
is not entirely surprising, then, that the assumptions described below
are similar to those used in Proposition \ref{Prop: Linear combinations}:

\begin{asm3} (i) The sub-models $\{P_{s,h};h\in T(P_{0})\}$ satisfy
(\ref{eq:qmd non-parametrics}). Furthermore, they admit an efficient
influence function, $\psi$, such that (\ref{eq:influence function})
holds. 

(ii) The stopping time $\hat{\tau}$ is a continuous function of $x_{n}(\cdot)$
in the sense that $\hat{\tau}=\tau(x_{n}(\cdot))$, where $\tau(\cdot)$
satisfies the conditions for an extended continuous mapping theorem
\citep[Theorem 1.11.1]{van1996weak}.\end{asm3}

Assumption 3(i) is a mild regularity condition that is common in non-parametric
analysis. Assumption 3(ii), which is substantive, states that the
stopping times depend only on the efficient influence function process.
This is indeed the case for the examples considered in Section \ref{sec:Applications}.
More generally, however, it may be that $\hat{\tau}$ depends on other
statistics beyond $x_{n}(\cdot)$. In such situations, the set of
asymptotically sufficient statistics should be expanded to include
these additional ones. We remark that an extension of our results
to these situations is straightforward, see Section \ref{subsec:Non-parametric-tests-batched}
for an illustration.

We call a test, $\varphi_{n}$, of $H_{0}:\mu=0$ asymptotically level-$\alpha$
if 
\[
\sup_{\left\{ h\in T(P_{0}):\left\langle \psi,h\right\rangle =0\right\} }\limsup_{n}\int\varphi_{n}dP_{nT,h}\le\alpha.
\]
Our first result in this section is a power envelope for asymptotically
level-$\alpha$ tests. Consider a limit experiment where one observes
a stopping time $\tau$, which is the weak limit of $\hat{\tau}$,
and a Gaussian process $x(\cdot)\sim\sigma^{-1}\mu\cdot+W(\cdot)$,
where $W(\cdot)$ denotes 1-dimensional Brownian motion. By Assumption
3(ii), $\tau$ is adapted to the filtration generated by the sample
paths of $x(\cdot)$. For any $\mu\in\mathbb{R}$, let $\mathbb{E}_{\mu}[\cdot]$
denote the induced distribution over the sample paths of $x(\cdot)$
between $[0,T]$. Also, define
\begin{equation}
\varphi_{\mu}^{*}(\tau,x(\tau)):=\mathbb{I}\left\{ \mu x(\tau)-\frac{\mu^{2}}{2\sigma}\tau\ge\gamma\right\} ,\label{eq:UMP test: non-parametrics}
\end{equation}
with $\gamma$ being determined by the requirement $\mathbb{E}_{0}[\varphi_{\mu}^{*}]=\alpha$,
and set $\beta^{*}(\mu):=\mathbb{E}_{\mu}[\varphi_{\mu}^{*}]$. 

\begin{prop} \label{Prop: Non-parametric}Suppose Assumption 3 holds.
Let $\beta_{n}(h)$ the power of some asymptotically level-$\alpha$
test, $\varphi_{n}$, of $H_{0}:\mu=0$ against local alternatives
$P_{\delta/\sqrt{n},h}$. Then, for every $h\in T(P_{0})$ and $\mu:=\delta\left\langle \psi,h\right\rangle $,
$\limsup_{n\to\infty}\beta_{n}(h)\le\beta^{*}\left(\mu\right)$.\end{prop} 

A similar result holds for unbiased tests. Following \citet{choi1996asymptotically},
we say that a test $\varphi_{n}$ of $H_{0}:\mu=0$ vs $H_{1}:\mu\neq0$
is asymptotically unbiased if 
\begin{align*}
\sup_{\left\{ h\in T(P_{0}):\left\langle \psi,h\right\rangle =0\right\} }\limsup_{n}\int\varphi_{n}dP_{nT,h} & \le\alpha,\ \textrm{and}\\
\inf_{\left\{ h\in T(P_{0}):\left\langle \psi,h\right\rangle \ne0\right\} }\liminf_{n}\int\varphi_{n}dP_{nT,h} & \ge\alpha.
\end{align*}
The next result states that the local power of such a test is bounded
by that of a best unbiased in the limit experiment, assuming one exists.

\begin{prop} \label{Prop: Non-parametric-unbiased}Suppose Assumption
3 holds and there exists a best unbiased test, $\tilde{\varphi}^{*}$,
in the limit experiment with power function $\bar{\beta}^{*}(\mu)$.
Let $\beta_{n}(h)$ denote the power of some asymptotically unbiased
test, $\varphi_{n}$, of $H_{0}:\mu=0$ vs $H_{1}:\mu\neq0$ over
local alternatives $P_{\delta/\sqrt{n},h}$. Then, for every $h\in T(P_{0})$
and $\mu:=\delta\left\langle \psi,h\right\rangle $, $\limsup_{n\to\infty}\beta_{n}(h)\le\tilde{\beta}^{*}\left(\mu\right)$.\end{prop} 

The proof is analogous to that of Proposition \ref{Prop: Non-parametric},
and is therefore omitted. Also, both propositions can be extended
to $\bm{\alpha}$-spending constraints but we omit formal statements
for brevity.

By similar reasoning as in Section \ref{subsec:Attaining-the-bound}
(using parametric sub-models), it follows that we can attain the power
bounds $\beta^{*}(\cdot),\tilde{\beta}^{*}(\cdot)$ by employing plug-in
versions of the corresponding UMP tests. This process simply involves
replacing $\tau,x(\tau)$ with $\hat{\tau},x_{n}(\hat{\tau})$. The
statistic $x_{n}(\hat{\tau})$ depends on the variance, $\sigma$,
so we must substitute it with a consistent estimate. We discuss various
estimators for $\sigma$ in Appendix \ref{subsec:Variance-estimators}. 

\section{Non-parametric two-sample tests\label{sec:Two-sample-tests}}

In many sequential experiments it is common to test two treatments
simultaneously. We may then be interested in conducting inference
on the difference between some regular functionals of the two treatments.
A salient example of this is inference on the expected treatment effect. 

To make matters precise, let $a\in\{0,1\}$ denote the two treatments,
with $P^{(a)}$ being the corresponding outcome distribution. Suppose
that at each period, the experimenter samples treatment 1 at some
fixed proportion $\pi$. It is without loss of generality to suppose
that the outcomes from the two treatments are independent as we can
only ever observe the effect of a single treatment. We are interested
in conducting inference on the difference, $\mu(P^{(1)})-\mu(P^{(0)})$,
where $\mu(\cdot)$ is some regular functional of the data distribution.
As before, we take $\mu$ to be scalar.

Let $P_{0}^{(1)},P_{0}^{(0)}$ denote some reference probability distributions
on the boundary of the null hypothesis so that $\mu(P_{0}^{(1)})-\mu(P_{0}^{(0)})=0$.
Following \citet[Section 25.6]{van2000asymptotic}, we consider the
power of tests against smooth one-dimensional sub-models of the form
$\left\{ \left(P_{s,h_{1}}^{(1)},P_{s,h_{0}}^{(0)}\right):s\le\eta\right\} $
for some $\eta>0$, where $h_{a}(\cdot)$ is a measurable function
satisfying 
\begin{equation}
\int\left[\frac{\sqrt{dP_{s,h_{a}}^{(a)}}-\sqrt{dP_{0}^{(a)}}}{s}-\frac{1}{2}h_{a}\sqrt{dP_{0}^{(a)}}\right]^{2}d\nu\to0\ \textrm{as}\ s\to0.\label{eq:qmd non-parametrics-1}
\end{equation}

As before, the set of all possible $h_{a}$ satisfying $\int h_{a}dP_{0}^{(a)}=0$
and $\int h_{a}^{2}dP_{0}^{(a)}<\infty$ forms a tangent space $T(P_{0}^{(a)})$.
This is a subset of the Hilbert space $L^{2}(P_{0}^{(a)})$, endowed
with the inner product $\left\langle f,g\right\rangle _{a}=\mathbb{E}_{P_{0}^{(a)}}[fg]$
and norm $\left\Vert f\right\Vert _{a}=\mathbb{E}_{P_{0}^{(a)}}[f^{2}]^{1/2}$.
Let $\psi_{a}\in T(P_{0}^{(a)})$ denote the efficient influence function
satisfying
\begin{equation}
\frac{\mu(P_{s,h_{a}}^{(a)})-\mu(P_{0}^{(a)})}{s}-\left\langle \psi_{a},h_{a}\right\rangle _{a}=o(s)\label{eq:influence function-1}
\end{equation}
for any $h_{a}\in T(P_{0}^{(a)})$. Denote $\sigma_{a}^{2}=\mathbb{E}_{P_{0}^{(a)}}[\psi_{a}^{2}]$.
The sufficient statistic here is the differenced efficient influence
function process 
\begin{equation}
x_{n}(t):=\frac{1}{\sigma}\left(\frac{1}{\pi\sqrt{n}}\sum_{i=1}^{\left\lfloor n\pi t\right\rfloor }\psi_{1}(Y_{i}^{(1)})-\frac{1}{(1-\pi)\sqrt{n}}\sum_{i=1}^{\left\lfloor n(1-\pi)t\right\rfloor }\psi_{0}(Y_{i}^{(0)})\right),\label{eq:two-sample tests sufficient statistic}
\end{equation}
where $\sigma^{2}:=\left(\frac{\sigma_{1}^{2}}{\pi}+\frac{\sigma_{0}^{2}}{1-\pi}\right)$.
Note that the number of observations from each treatment at time $t$
is $\left\lfloor n\pi t\right\rfloor ,\left\lfloor n(1-\pi)t\right\rfloor $.
The assumptions below are analogous to Assumption 3:

\begin{asm4} (i) The sub-models $\{P_{s,h_{a}}^{(a)};h_{a}\in T(P_{0}^{(a)})\}$
satisfy (\ref{eq:qmd non-parametrics-1}). Furthermore, they admit
an efficient influence function, $\psi_{a}$, such that (\ref{eq:influence function-1})
holds. 

(ii) The stopping time $\hat{\tau}$ is a continuous function of $x_{n}(\cdot)$
in the sense that $\hat{\tau}=\tau(x_{n}(\cdot))$, where $\tau(\cdot)$
satisfies the conditions for an extended continuous mapping theorem
\citep[Theorem 1.11.1]{van1996weak}.\end{asm4}

Set $\mu_{a}:=\mu(P^{(a)})$. A test, $\varphi_{n}$, of $H_{0}:\mu_{1}-\mu_{0}=0$
is asymptotically level-$\alpha$ if 
\begin{equation}
\sup_{\left\{ \bm{h}:\left\langle \psi_{1},h_{1}\right\rangle _{1}-\left\langle \psi_{0},h_{0}\right\rangle _{0}=0\right\} }\limsup_{n}\int\varphi_{n}dP_{nT,\bm{h}}\le\alpha.\label{eq:Level -alpha - two sample definition}
\end{equation}
Similarly, a test, $\varphi_{n}$, of $H_{0}:\mu_{1}-\mu_{0}=0$ vs
$H_{1}:\mu_{1}-\mu_{0}\neq0$ is asymptotically unbiased if 
\begin{align}
\sup_{\left\{ \bm{h}:\left\langle \psi_{1},h_{1}\right\rangle _{1}-\left\langle \psi_{0},h_{0}\right\rangle _{0}=0\right\} }\limsup_{n}\int\varphi_{n}dP_{nT,\bm{h}} & \le\alpha,\ \textrm{and}\nonumber \\
\inf_{\left\{ \bm{h}:\left\langle \psi_{1},h_{1}\right\rangle _{1}-\left\langle \psi_{0},h_{0}\right\rangle _{0}\neq0\right\} }\liminf_{n}\int\varphi_{n}dP_{nT,\bm{h}} & \ge\alpha.\label{eq:Unbiased test definition - tow sample}
\end{align}
Consider the limit experiment where one observes $x(\cdot)\sim\sigma^{-1}(\mu_{1}-\mu_{0})\cdot+W(\cdot)$
and a $\mathcal{F}_{t}\equiv\sigma\{x(s);s\le t\}$ adapted stopping
time $\tau$ that is the weak limit of $\hat{\tau}$. Then, setting
$\mu:=\mu_{1}-\mu_{0}$, define the power functions $\beta^{*}(\cdot),\tilde{\beta}^{*}(\cdot)$
as in the previous section. The following results provide upper bounds
on asymptotically level-$\alpha$ and asymptotically unbiased tests. 

\begin{prop} \label{Prop: Two-sample non-parametric}Suppose Assumption
4 holds. Let $\beta_{n}(\bm{h})$ the power of some asymptotically
level-$\alpha$ test, $\varphi_{n}$, of $H_{0}:\mu_{1}-\mu_{0}=0$
against local alternatives $P_{\delta_{1}/\sqrt{n},h_{1}}^{(1)}\times P_{\delta_{0}/\sqrt{n},h_{0}}^{(0)}$.
Then, for every $\bm{h}\in T(P_{0}^{(1)})\times T(P_{0}^{(0)})$ and
$\mu:=\delta_{1}\left\langle \psi_{1},h_{1}\right\rangle _{1}-\delta_{0}\left\langle \psi_{0},h_{0}\right\rangle _{0}$,
$\limsup_{n\to\infty}\beta_{n}(\bm{h})\le\beta^{*}\left(\mu\right)$.\end{prop} 

\begin{prop} \label{Prop: Non-parametric-unbiased-two sample}Suppose
Assumption 4 holds and there exists a best unbiased test $\tilde{\varphi}^{*}$
in the limit experiment. Let $\beta_{n}(\bm{h})$ the power of some
asymptotically unbiased test, $\varphi_{n}$, of $H_{0}:\mu_{1}-\mu_{0}=0$
vs $H_{1}:\mu_{1}-\mu_{0}\neq0$ against local alternatives $P_{\delta_{1}/\sqrt{n},h_{1}}^{(1)}\times P_{\delta_{0}/\sqrt{n},h_{0}}^{(0)}$.
Then, for every $\bm{h}\in T(P_{0}^{(1)})\times T(P_{0}^{(0)})$ and
$\mu:=\delta_{1}\left\langle \psi_{1},h_{1}\right\rangle _{1}-\delta_{0}\left\langle \psi_{0},h_{0}\right\rangle _{0}$,
$\limsup_{n\to\infty}\beta_{n}(\bm{h})\le\tilde{\beta}^{*}\left(\mu\right)$.\end{prop} 

We prove Proposition \ref{Prop: Two-sample non-parametric} in Appendix
A. The proof of Proposition \ref{Prop: Non-parametric-unbiased-two sample}
is similar and therefore omitted. Both Propositions \ref{Prop: Two-sample non-parametric}
and \ref{Prop: Non-parametric-unbiased-two sample} can be extended
to $\bm{\alpha}$-spending constraints. We omit the formal statements
for brevity.

\section{Optimal tests in batched experiments\label{sec:ART:2}}

We now analyze sequential experiments with multiple treatments and
where the sampling rule, i.e., the number of units allocated to each
treatment, also changes over the course of the experiment. Since our
results here draw on \citet{hirano2023asymptotic}, we restrict attention
to batched experiments, where the sampling strategy is only allowed
to be changed at some fixed, discrete set of times. 

Suppose there are $K$ treatments under consideration. We take $K=2$
to simplify the notation, but all our results extend to any fixed
$K$. The outcomes, $Y^{(a)}$, under treatment $a\in\{0,1\}$ are
distributed according to some parametric model $\{P_{\theta^{(a)}}^{(a)}\}$.
Here $\theta^{(a)}\in\mathbb{R}^{d}$ is some unknown parameter vector;
we assume for simplicity that the dimension of $\theta^{(1)},\theta^{(0)}$
is the same, but none of our results actually require this. It is
without loss of generality to suppose that the outcomes from each
treatment are independent conditional on $\theta^{(1)},\theta^{(0)}$,
as we only ever observe one of the two potential outcomes for any
given observation. In the batch setting, the DM divides the observations
into batches of size $n$, and registers a sampling rule $\{\hat{\pi}_{j}^{(a)}\}_{j}$
that prescribes the fraction of observations allocated to treatment
$a$ in batch $j$ based on information from the previous batches
$1,\dots,j-1$. The experiment ends after $J$ batches. It is possible
to set $\pi_{j}^{(a)}=0$ for some or all treatments (e.g., the experiment
may be stopped early); we only require $\sum_{a}\hat{\pi}_{j}^{(a)}\le1$
for each $j$. We develop asymptotic representation theorems for tests
of $H_{0}:\theta=\Theta_{0}$ vs $H_{1}:\theta\in\Theta_{1}$, where
$\theta:=(\theta^{(1)},\theta^{(0)})$. Let $(\theta_{0}^{(1)},\theta_{0}^{(0)})\in\Theta_{0}$
denote some reference parameter in the null set.

Take $\hat{q}_{j}^{(a)}$ to be the proportion of observations allocated
to treatment $a$ up-to batch $j$, as a fraction of $n$. Let $Y_{j}^{(a)}$
denote the $j$-th observation of treatment $a$ in the experiment.
Any candidate test, $\delta(\cdot)$, is required to be 
\[
\sigma\left\{ \left(Y_{1}^{(0)},\dots,Y_{nq_{J}^{(0)}}^{(0)}\right),\left(Y_{1}^{(1)},\dots,Y_{nq_{J}^{(1)}}^{(1)}\right)\right\} 
\]
measurable. As in the previous sections, we measure the performance
of tests against local perturbations of the form $\{\theta_{0}^{(a)}+h_{a}/\sqrt{n};h_{a}\in\mathbb{R}^{d}\}$.
Let $\nu$ denote a dominating measure for $\{P_{\theta}^{(a)}:\theta\in\mathbb{R}^{d},a\in\{0,1\}\}$,
and set $p_{\theta}^{(a)}:=dP_{\theta}^{(a)}/d\nu$. We require $\{P_{\theta}^{(a)}\}$
to be quadratically mean differentiable (qmd): 

\begin{asm5} The class $\{P_{\theta}^{(a)}:\theta\in\mathbb{R}^{d}\}$
is qmd around $\theta_{0}^{(a)}$ for each $a\in\{0,1\}$, i.e., there
exists a score function $\psi_{a}(\cdot)$ such that for each $h_{a}\in\mathbb{R}^{d},$
\[
\int\left[\sqrt{p_{\theta_{0}^{(a)}+h_{a}}^{(a)}}-\sqrt{p_{\theta_{0}^{(a)}}^{(a)}}-\frac{1}{2}h_{a}^{\intercal}\psi_{a}\sqrt{p_{\theta_{0}^{(a)}}}\right]^{2}d\nu=o(\vert h_{a}\vert^{2}).
\]
Furthermore, the information matrix $I_{a}:=\mathbb{E}_{0}[\psi_{a}\psi_{a}^{\intercal}]$
is invertible for $a\in\{0,1\}$. \end{asm5}

Define $z_{j,n}^{(a)}(\hat{\pi}_{j})$ as the standardized score process
from each batch, where
\[
z_{j,n}^{(a)}(t):=\frac{I_{a}^{-1/2}}{\sqrt{n}}\sum_{i=1}^{\left\lfloor nt\right\rfloor }\psi_{a}(Y_{i,j}^{(a)})
\]
for each $t\in[0,1]$. Let $Y_{i,j}^{(a)}$ denote the $i$-th outcome
observation from arm $a$ in batch $j$. At each batch $j$, one can
imagine that there is a potential set of outcomes, $\{{\bf y}_{j}^{(1)},{\bf y}_{j}^{(0)}\}$
with ${\bf y}_{j}^{(a)}:=\{Y_{i,j}^{(a)}\}_{i=1}^{n}$, that could
be sampled from both arms, but only a sub-collection, $\{Y_{i,j}^{(a)};i=1,\dots,n\hat{\pi}_{j}^{(a)}\}$,
of these are actually sampled. Let $\bm{h}:=(h_{1},h_{0})$, take
$P_{n,\bm{h}}$ to be the joint probability measure over 
\[
\{{\bf y}_{1}^{(1)},{\bf y}_{1}^{(0)},\dots,{\bf y}_{J}^{(1)},{\bf y}_{J}^{(0)}\}
\]
when each $Y_{i,j}^{(a)}\sim P_{\theta_{0}^{(a)}+h_{a}/\sqrt{n}}$,
and take $\mathbb{E}_{n,\bm{h}}[\cdot]$ to be its corresponding expectation.
Then, by a standard functional central limit theorem, 
\begin{equation}
z_{j,n}^{(a)}(t)\xrightarrow[P_{n,0}]{d}z(t);\ z(\cdot)\sim W_{j}^{(a)}(\cdot),\label{eq:Convergence of score process-1}
\end{equation}
where $\{W_{j}^{(a)}\}_{j,a}$ are independent $d$-dimensional Brownian
motions. 

\subsection{Asymptotic representation theorem\label{subsec:Asymptotic-representation-theorem-batched}}

Consider a limit experiment where $\bm{h}:=(h_{1},h_{0})$ is unknown,
and for each batch $j$, one observes the stopped process $z_{j}^{(a)}(\pi_{j}^{(a)})$,
where 
\begin{equation}
z_{j}^{(a)}(t):=I_{a}^{1/2}h_{a}t+W_{j}^{(a)}(t),\label{eq:distribution z_j in limit experiment}
\end{equation}
and $\{W_{j}^{(a)};j=1,\dots,J;a=0,1\}$ are independent Brownian
motions. Each $\pi_{j}^{(a)}$ is required to satisfy $\sum_{a}\pi_{j}^{(a)}\le1$
and also to be
\[
\sigma\left\{ (z_{1}^{(1)},z_{1}^{(0)},U_{1}),\dots,(z_{j-1}^{(1)},z_{j-1}^{(0)},U_{j-1})\right\} 
\]
measurable, where $U_{j}\sim\textrm{Uniform}[0,1]$ is exogenous to
all the past values $\left\{ z_{j^{\prime}}^{(a)},U_{j^{\prime}}:j^{\prime}<j\right\} $.
Let $\varphi$ denote a test statistic for $H_{0}:h=0$ that depends
only on: (i) $q_{a}=\sum_{j}\pi_{j}^{(a)}$, i.e., the number of times
each arm was pulled; and (ii) $x_{a}=\sum_{j}z_{j}^{(a)}(\pi_{j}^{(a)})$,
i.e., the sum of outcomes from each arm. Let $\mathbb{P}_{\bm{h}}$
denote the joint probability measure over $\{z_{j}^{(a)}(\cdot);a\in\{0,1\},j\in\{1,\dots,J\}\}$
when each $z_{j}^{(a)}(\cdot)$ is distributed as in (\ref{eq:distribution z_j in limit experiment}),
and take $\mathbb{E}_{\bm{h}}[\cdot]$ to be its corresponding expectation. 

The following theorem shows that the power function of any test $\varphi_{n}$
in the original testing problem can be matched by one such test, $\varphi$,
in the limit experiment.

\begin{thm} \label{Thm: ART-Batched}Suppose Assumption 5 holds.
Let $\varphi_{n}$ be some test function in the original batched experiment,
and $\beta_{n}(\bm{h})$, its power against $P_{n,\bm{h}}$. Then,
for every sequence $\{n_{j}\}$, there is a further sub-sequence $\{n_{j_{m}}\}$
such that: \\
(i) \citep{hirano2023asymptotic} There exists a batched policy function
$\pi=\{\pi_{j}^{(a)}\}_{j}$ and processes $\{z_{j}^{(a)}(\cdot)\}_{j,a}$
defined on the limit experiment for which
\begin{align*}
 & \left(\left(\hat{\pi}_{1}^{(1)},\hat{\pi}_{1}^{(0)},z_{1,n}^{(1)}(\hat{\pi}_{1}^{(1)}),z_{1,n}^{(0)}(\hat{\pi}_{1}^{(0)})\right),\dots,\left(\hat{\pi}_{J}^{(1)},\hat{\pi}_{J}^{(0)},z_{J,n}^{(1)}(\hat{\pi}_{J}^{(1)}),z_{J,n}^{(0)}(\hat{\pi}_{J}^{(0)})\right)\right)\\
 & \xrightarrow[P_{n,0}]{d}\left(\left(\pi_{1}^{(1)},\pi_{1}^{(0)},z_{1}^{(1)}(\pi_{1}^{(1)}),z_{1}^{(0)}(\pi_{1}^{(0)})\right),\dots,\left(\pi_{J}^{(1)},\pi_{J}^{(0)},z_{J}^{(1)}(\pi_{J}^{(1)}),z_{J}^{(0)}(\pi_{J}^{(0)})\right)\right).
\end{align*}
(ii) There exists a test $\varphi$ in the limit experiment depending
only on $q_{1},q_{0},x_{1},x_{0}$ such that $\beta_{n_{j_{m}}}(\bm{h})\to\beta(\bm{h})$
for every $\bm{h}\in\mathbb{R}^{d}\times\mathbb{R}^{d}$, where $\beta(\bm{h}):=\mathbb{E}_{\bm{h}}[\varphi]$
is the power of $\varphi$ in the limit experiment.\end{thm} 

The first part of Theorem \ref{Thm: ART-Batched} is due to \citet{hirano2023asymptotic};
we only modify the terminology slightly. Note that the results of
\citet{hirano2023asymptotic} already imply that any $\varphi_{n}$
can be asymptotically matched by a test $\varphi$ in the limit experiment
that is $\sigma\left\{ (z_{1}^{(1)},z_{1}^{(0)},U_{1}),\dots,(z_{J}^{(1)},z_{J}^{(0)},U_{J})\right\} $
measurable. The novel result here is the second part of Theorem \ref{Thm: ART-Batched},
which shows that a further dimension reduction is possible. A naive
application of \citet{hirano2023asymptotic} would require sufficient
statistics that grow linearly with the number of batches, leading
to a vector of dimension $2dJ+1$ (the uniform random variables $U_{1},\dots,U_{J}$
can be subsumed into a single $U\sim\textrm{Uniform}[0,1]$). Here,
we show that one only need condition on $q_{1},q_{0},x_{1},x_{0}$,
which are of a fixed dimension $2d+2$ (or $2d+1$ if we impose $q^{(1)}+q^{(0)}=J$).
This is a substantial reduction in dimension.

\subsubsection{An alternative representation of the limit experiment\label{subsec:An-alternative-representation}}

From the distribution of $z_{j}^{(a)}(\cdot)$ given in (\ref{eq:distribution z_j in limit experiment}),
it is easy to verify that 
\[
z_{j}^{(a)}(\pi_{j}^{(a)})\sim I_{a}^{1/2}h_{a}\pi_{j}^{(a)}+W_{j}^{(a)}(\pi_{j}^{(a)}).
\]
Combined with the definition $q_{a}=\sum_{j}\pi_{j}^{(a)}$ and the
fact $\{W_{j}^{(a)};j=1,\dots,J;a=0,1\}$ are independent Brownian
motions, we obtain 
\begin{equation}
x_{a}=\sum_{j}z_{j}^{(a)}(\pi_{j}^{(a)})\sim I_{a}^{1/2}h_{a}q_{a}+W_{a}(q_{a}),\label{eq:distribution of x_a  in limit experiment}
\end{equation}
where $W_{1}(\cdot).W_{0}(\cdot)$ are standard $d$-dimensional Brownian
motions that are again independent of each other. In view of the above,
we can alternatively think of the limit experiment as observing $\{q_{a}\}_{a}$
along with $\{x_{a}\}_{a}$, with the latter distributed as in (\ref{eq:distribution of x_a  in limit experiment}).
The advantage of this formulation is that it is independent of the
number of batches. It therefore provides suggestive evidence that
the asymptotic representation in Theorem \ref{Thm: ART-Batched} would
remain valid under continuous experimentation (however, our proof
only applies to a finite number of batches). 

\subsection{Characterization of optimal tests in the limit experiment\label{subsec:Characterization-of-optimal-batched}}

It is generally unrealistic in batched sequential experiments for
the sampling rule to depend on fewer statistics than $q_{1},q_{0},x_{1},x_{0}$.
Consequently, we do not have sharp results for testing linear combinations
as in Proposition \ref{Prop: Linear combinations}. We do, however,
have analogues to the other results in Section \ref{subsec:Characterization-of-optimal}.

\subsubsection{Power envelope}

Consider testing $H_{0}:\bm{h}=0$ vs $H_{1}:\bm{h}=\bm{h}_{1}$ in
the limit experiment. By the Neyman-Pearson lemma, and the Girsanov
theorem applied on (\ref{eq:distribution of x_a  in limit experiment}),
the optimal test is given by 
\begin{equation}
\varphi_{\bm{h}_{1}}^{*}=\mathbb{I}\left\{ \sum_{a\in\{0,1\}}\left(h_{a}^{\intercal}I_{a}^{1/2}x_{a}-\frac{q_{a}}{2}h_{a}^{\intercal}I_{a}h_{a}\right)\ge\gamma_{h_{1}}\right\} ,\label{eq:UMP test batched}
\end{equation}
where $\gamma_{\bm{h}_{1}}$ is chosen such that $\mathbb{E}_{0}[\varphi_{h_{1}}^{*}]=\alpha$.
Take $\beta^{*}(\bm{h}_{1})$ to be the power function of $\varphi_{\bm{h}_{1}}^{*}$
against $H_{1}:\bm{h}=\bm{h}_{1}$. Theorem \ref{Thm: ART-Batched}
shows that $\beta^{*}(\cdot)$ is an asymptotic power envelope for
any test of $H_{0}:\theta=\theta_{0}$ in the original experiment. 

\subsubsection{Unbiased tests}

Suppose $\varphi(q_{1},q_{0},x_{1},x_{0})$ is an unbiased test of
$H_{0}:\bm{h}=0$ vs $H_{1}:\bm{h}\neq0$ in the limit experiment.
Then, in analogy with Proposition \ref{Prop: Unbiased}, it needs
to satisfy the following property:

\begin{prop}\label{Prop: Unbiased-Batched} Any unbiased test of
$H_{0}:\bm{h}=0$ vs $H_{1}:\bm{h}\neq0$ in the limit experiment
must satisfy $\mathbb{E}_{0}[x_{a}\varphi(q_{1},q_{0},x_{1},x_{0})]=0$
where $x_{a}\sim W_{a}(q_{a})$ under $\mathbb{P}_{0}$. \end{prop} 

\subsubsection{Weighted average power}

Let $w(\cdot)$ denote a weight function over alternatives $\bm{h}\neq0$.
Then, the uniquely optimal test of $H_{0}:\bm{h}=0$ that maximizes
weighted average power over $w(\cdot)$ is given by
\[
\varphi_{w}^{*}=\mathbb{I}\left\{ \int\exp\left\{ \sum_{a\in\{0,1\}}\left(h_{a}^{\intercal}I_{a}^{1/2}x_{a}-\frac{q_{a}}{2}h_{a}^{\intercal}I_{a}h_{a}\right)\right\} dw(\bm{h})\ge\gamma\right\} .
\]
The value of $\gamma$ is chosen to satisfy $\mathbb{E}_{0}[\varphi_{w}^{*}]=\alpha$.
In practice, it can be computed by simulation. 

\subsection{Non-parametric tests\label{subsec:Non-parametric-tests-batched}}

For the non-parametric setting, we make use of the same notation as
in Section \ref{sec:Two-sample-tests}. We are interested in conducting
inference on some regular vector of functionals, $\left(\mu(P^{(1)}),\mu(P^{(0)})\right)$,
of the outcome distributions $P^{(1)},P^{(0)}$ for the two treatments.
To simplify matters, we take $\mu_{a}:=\mu(P^{(a)})$ to be scalar.
The definition of asymptotically level-$\alpha$ and unbiased tests
is unchanged from (\ref{eq:Level -alpha - two sample definition})
and (\ref{eq:Unbiased test definition - tow sample}). 

Let $\psi_{a},\sigma_{a}$ be defined as in Section \ref{sec:Two-sample-tests}.
Set
\[
z_{j,n}^{(a)}:=\frac{1}{\sigma_{a}\sqrt{n}}\sum_{i=1}^{\left\lfloor nt\right\rfloor }\psi_{a}(Y_{i,j}^{(a)}),
\]
and take $s_{n}(\cdot)=\left\{ x_{n,1}(\cdot),x_{n,0}(\cdot),q_{n,1}(\cdot),q_{n,0}(\cdot)\right\} $
to be the vector of state variables, where
\begin{align*}
x_{n,a}(k) & :=\sum_{j=1}^{k}z_{n,j}^{(a)}(\hat{\pi}_{j}^{(a)}),\ \textrm{and }\ q_{n,a}(k):=\sum_{j=1}^{k}\hat{\pi}_{j}^{(a)}.
\end{align*}
\begin{asm6}(i) The sub-models $\{P_{s,h_{a}}^{(a)};h_{a}\in T(P_{0}^{(a)})\}$
satisfy (\ref{eq:qmd non-parametrics-1}). Furthermore, they admit
an efficient influence function, $\psi_{a}$, such that (\ref{eq:influence function-1})
holds. 

(ii) The sampling rule $\hat{\pi}_{j+1}$ in batch $j$ is a continuous
function of $s_{n}(j)$ in the sense that $\hat{\pi}_{j+1}=\pi_{j+1}(s_{n}(j))$,
where $\pi_{j+1}(\cdot)$ satisfies the conditions for an extended
continuous mapping theorem \citep[Theorem 1.11.1]{van1996weak} for
each $j=0,\dots,K-1$.\end{asm6}

Assumption 6(i) is standard. Assumption 6(ii) implies that the sampling
rule depends on a vector of four state variables. This is in contrast
to the single sufficient statistic used in Section \ref{sec:Two-sample-tests}.
We impose Assumption 6(ii) as it is more realistic; many commonly
used algorithms, e.g., Thompson sampling, depend on all four statistics.
The assumption still imposes a dimension reduction as it requires
the sampling rule to be independent of the data conditional on knowing
$s_{n}(\cdot)$. In practice, any Bayes or minimax optimal algorithm
would only depend on $s_{n}(\cdot)$ anyway, as noted in \citet{adusumilli2021risk}.
In fact, we are not aware of any commonly used algorithm that requires
more statistics beyond these four. 

The reliance of the sampling rule on the vector $s_{n}(\cdot)$ implies
that the optimal test should also depend on the full vector, and cannot
be reduced further. The relevant limit experiment is the one described
in Section \ref{subsec:An-alternative-representation}, with $\mu_{a}$
replacing $h_{a}$. Also, let 
\[
\varphi_{\bar{\mu_{1}},\bar{\mu_{0}}}=\mathbb{I}\left\{ \sum_{a\in\{0,1\}}\left(\frac{\bar{\mu_{a}}}{\sigma_{a}}x_{a}-\frac{q_{a}}{2\sigma_{a}^{2}}\bar{\mu}_{a}^{2}\right)\ge\gamma_{\bar{\mu}_{1},\bar{\mu}_{0}}\right\} 
\]
denote the Neyman-Pearson test of $H_{0}:(\mu_{1},\mu_{0})=(0,0)$
vs $H_{1}:(\mu_{1},\mu_{0})=(\bar{\mu}_{1},\bar{\mu}_{0})$ in the
limit experiment, with $\gamma_{\bar{\mu}_{1},\bar{\mu}_{0}}$ determined
by the size requirement. Take $\beta(\bar{\mu}_{1},\bar{\mu}_{0}$)
to be its corresponding power. 

\begin{prop} \label{Prop: Non-parametric batched experiments}Suppose
Assumption 6 holds. Let $\beta_{n}(\bm{h})$ the power of some asymptotically
level-$\alpha$ test, $\varphi_{n}$, of $H_{0}:(\mu_{1},\mu_{0})=(0,0)$
against local alternatives $P_{\delta_{1}/\sqrt{n},h_{1}}^{(1)}\times P_{\delta_{0}/\sqrt{n},h_{0}}^{(0)}$.
Then, for every $\bm{h}\in T(P_{0}^{(1)})\times T(P_{0}^{(0)})$ and
$\mu_{a}:=\delta_{a}\left\langle \psi_{a},h_{a}\right\rangle _{a}$
for $a\in\{0,1\}$, $\limsup_{n\to\infty}\beta_{n}(\bm{h})\le\beta^{*}\left(\mu_{1},\mu_{0}\right)$.\end{prop} 

Proposition \ref{Prop: Non-parametric batched experiments} describes
the power envelope for testing that the parameter vector takes on
a given value. Suppose, however, that one is only interested in providing
inference for single component of that vector, say $\mu_{1}$. Then
$\mu_{0}$ is a nuisance parameter under the null, and one would need
to employ the usual strategies for getting rid of the dependence on
$\mu_{0}$, e.g., through conditional inference or minimax tests.
We leave the discussion of these possibilities for future research. 

\section{Applications\label{sec:Applications}}

\subsection{Horizontal boundary designs\label{subsec:Sequential-linear-boundary}}

As a first illustration of our methods, consider the class of horizontal
boundary designs with a fixed sampling rule, $\pi$, and the stopping
time $\hat{\tau}=\inf\left\{ t:\vert x_{n}(t)\vert\ge\gamma\right\} $,
where $x_{n}(t)$ is defined as in (\ref{eq:two-sample tests sufficient statistic}).
As a concrete example, suppose $\mu_{1},\mu_{0}$ denote the mean
values of outcomes from each treatment, with $\sigma_{1},\sigma_{0}$
their corresponding standard deviations. If the goal of the experiment
is to determine the treatment with the largest mean while minimizing
the number of samples, which are costly, then, as shown in \citet{adusumilli2022sample},
the minimax optimal sampling strategy is the Neyman allocation $\pi_{1}^{*}=\sigma_{1}/(\sigma_{1}+\sigma_{0})$,
and optimal stopping rule is $\hat{\tau}=\inf\left\{ t:\vert x_{n}(t)\vert\ge\gamma\right\} $
with the efficient influence functions $\psi_{1}(Y)=\psi_{0}(Y)=Y$. 

We are interested in testing the null of no treatment effect, $H_{0}:\mu_{1}-\mu_{0}=0$
vs $H_{1}:\mu_{1}-\mu_{0}\neq0$. Let $F_{\mu}(\cdot)$ denote the
distribution of $\tau$ in the limit experiment where $x(t)\sim\sigma^{-1}\mu t+W(t)$
and $\tau=\inf\{t:\vert x(t)\vert\ge\gamma\}$. In \citet{adusumilli2022sample},
this author suggested employing the test function $\hat{\varphi}=\mathbb{I}\{\hat{\tau}\le F_{0}^{-1}(\alpha)\}$.
This corresponds to the test $\varphi^{*}=\mathbb{I}\{\tau\le F_{0}^{-1}(\alpha)\}$
in the limit experiment. However, no argument was given as to its
optimality. The following result, proved in Appendix \ref{subsec:Supporting-information-for-sequential-linear-boundary},
shows that $\hat{\varphi}$ is in fact the UMP asymptotically unbiased
test.

\begin{lem} \label{Lem: costly sampling} Consider the sequential
experiment described above with a fixed sampling rule $\pi$ and stopping
time $\hat{\tau}=\inf\left\{ t:\vert x_{n}(t)\vert\ge\gamma\right\} $.
The test, $\hat{\varphi}=\mathbb{I}\{\hat{\tau}\le F_{0}^{-1}(\alpha)\}$,
is the UMP asymptotically unbiased test (in the sense that it attains
the upper bound in Proposition \ref{Prop: Non-parametric}) of $H_{0}:\mu_{1}=\mu_{0}$
vs $H_{1}:\mu_{1}\neq\mu_{0}$ in this experiment.\end{lem} 

\subsubsection{Numerical Illustration}

To illustrate the finite sample performance of this test, we ran Monte-Carlo
simulations with $Y_{i}^{(1)}=\delta+\epsilon_{i}^{(1)}$ and $Y_{i}^{(0)}=\epsilon_{i}^{(0)}$
where $\epsilon_{i}^{(1)},\epsilon_{i}^{(0)}\sim\sqrt{3}\times\textrm{Uniform}[-1,1]$.
The threshold, $\gamma$, was taken to be $0.536$ (this corresponds
to a sampling cost of $c=1$ for each observation in the costly sampling
framework), and the treatments were sampled in equal proportions $(\pi=1/2$).
Figure \ref{fig:Test}, Panel A plots the size of the test for different
values of $n$ under the nominal $5\%$ significance level. Even for
relatively small values of $n$, the size is close to nominal. We
also plot the size of the standard two-sample test for comparison;
due to the adaptive stopping rule, this test is not valid and its
actual size is close to 9\%. Panel B of the same figure plots the
finite sample power functions for $\hat{\varphi}$ under different
$n$. The power is computed against local alternatives; the reward
gap in the figure is the scaled one, $\mu=\sqrt{n}\vert\delta\vert$.
But for any given $n$, the actual difference in mean outcomes is
$\mu/\sqrt{n}$. The same plot also displays the asymptotic power
envelope for unbiased tests, obtained as the power function of the
best unbiased test, $\varphi{}^{*}=\mathbb{I}\{\tau\le F_{0}^{-1}(\alpha)\}$,
in the limit experiment. Even for small samples, the power function
of $\hat{\varphi}$ is close to the asymptotic upper bound. 
\begin{figure}
\includegraphics[height=5cm]{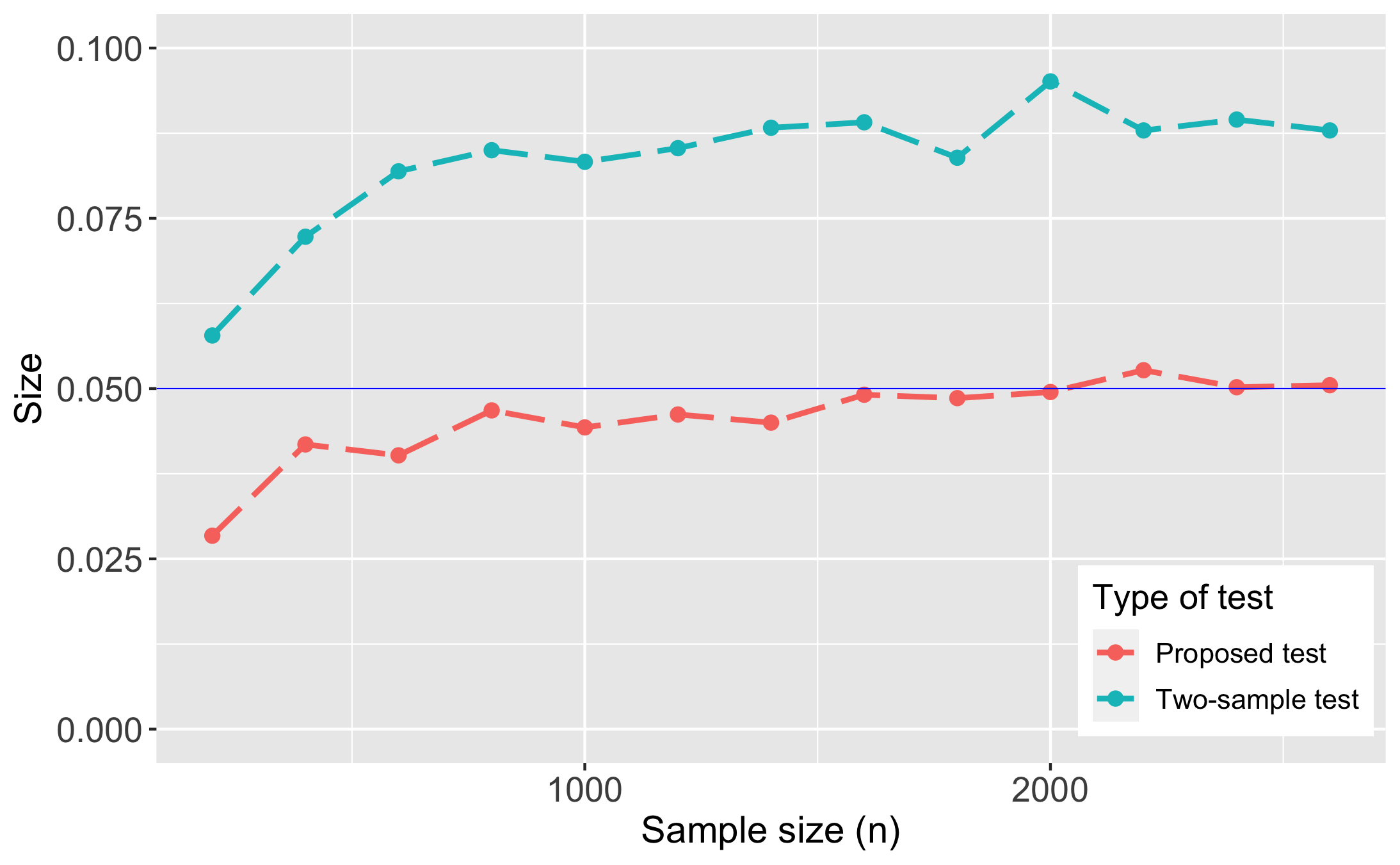}~\includegraphics[height=5cm]{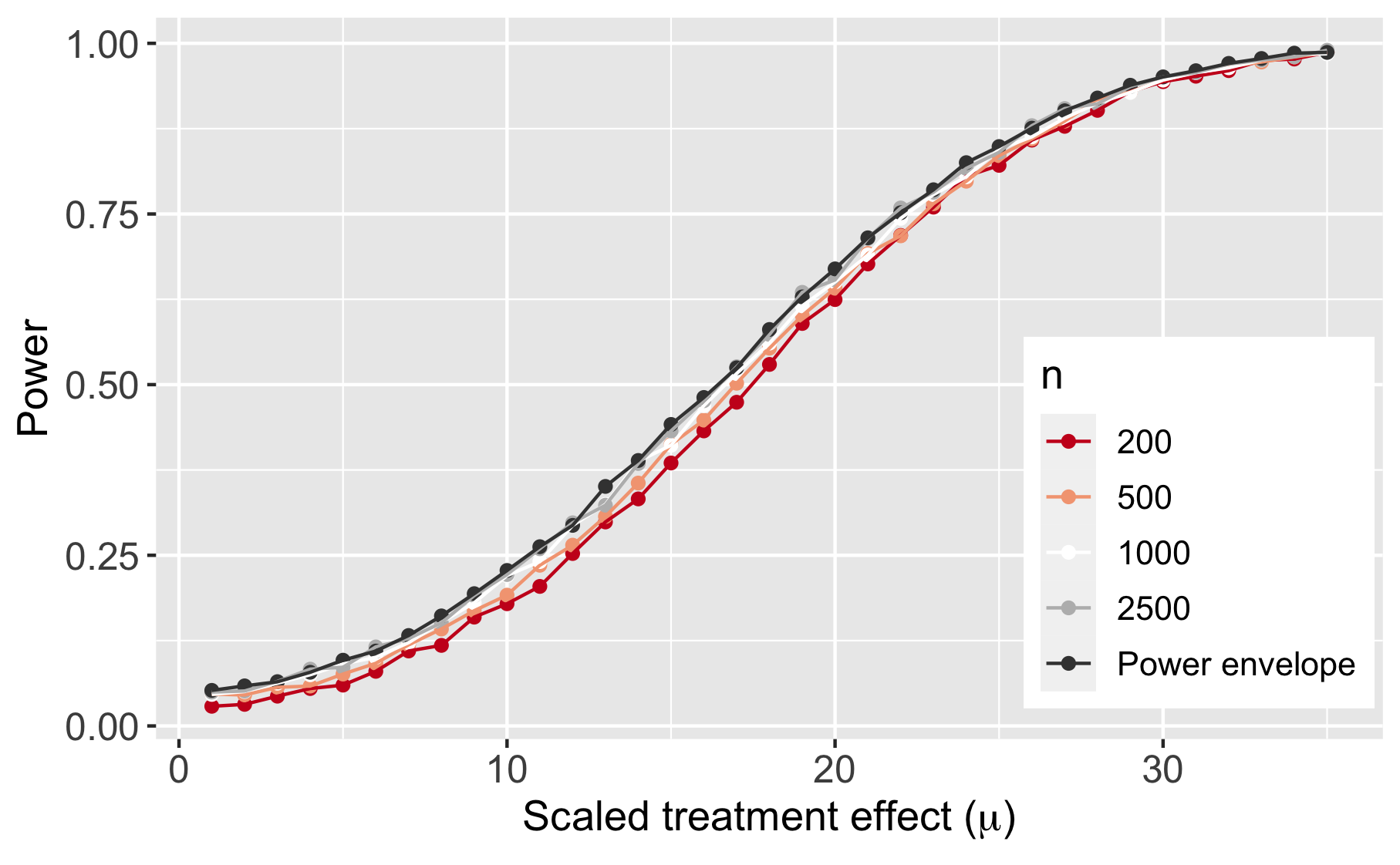}~~

\begin{tabular}{>{\centering}p{7.5cm}>{\centering}p{7.5cm}}
{\scriptsize{}A: Size} & {\scriptsize{}$\ $B: Power function}\tabularnewline
\end{tabular}

\vspace{0.5em}
\begin{raggedright}
{\scriptsize{}Note: Panel A plots the size of $\hat{\varphi}$ along
with that of the standard two-sample test at the nominal $5$\% level
(solid blue line) when the errors are drawn from a $\sqrt{3}\times\textrm{Uniform}[-1,1]$
distribution for each treatment. Panel B plots the finite sample power
envelopes of $\hat{\varphi}$ under different $n$, along with asymptotic
power envelope for unbiased tests. The scaled treatment effect  is
defined as $\mu=\sqrt{n}\vert\delta\vert$. }{\scriptsize\par}
\par\end{raggedright}
\caption{Finite sample performance of $\hat{\varphi}$ under horizontal boundary
designs\label{fig:Test}}
\end{figure}

\subsection{Group sequential experiments\label{subsec:Group-Sequential-Experiments}}

In this application, we suggest methods for inference on treatment
effects following group sequential experiments. To simplify matters,
suppose that the researchers assign the two treatments with equal
probability in each stage. Let $\mu_{1},\mu_{0}$ denote the expectation
of outcomes from the two treatments. Also, take $x_{n}(\cdot)$ to
be the scaled difference in sample means, i.e., it is the quantity
defined in (\ref{eq:two-sample tests sufficient statistic}) with
$\psi_{1}(Y)=\psi_{0}(Y)=Y$. While there are a number of different
group sequential designs, see, e.g., \citet{wassmer2016group} for
a textbook overview, the general construction is that the experiment
is terminated at the end of stage $t$ if $x_{n}(t)$ is outside some
interval $\mathcal{I}_{t}$. The stopping time $\hat{\tau}$ thus
satisfies $\{\hat{\tau}>t-1\}\equiv\cap_{l=1}^{t-1}\left\{ x_{n}(l)\in\mathcal{I}_{l}\right\} $.
The intervals $\{\mathcal{I}_{t}\}_{t=1}^{T}$ are pre-determined
and chosen by balancing various ethical, cost and power criteria.
We take them as given. 

We are interested in testing the drifting hypotheses $H_{0}:\mu_{1}-\mu_{0}=\bar{\mu}/\sqrt{n}$
vs $H_{1}:\mu_{1}-\mu_{0}>\bar{\mu}/\sqrt{n}$ at some spending level
$\bm{\alpha}$ that is chosen by experimenter.\footnote{In most examples of group sequential designs, the intervals $\mathcal{I}_{t}$
are themselves chosen to maximize power under some $\bar{\bm{\alpha}}$-spending
criterion, given the null of $\mu_{1}=\mu_{0}$. In general, our $\bm{\alpha}$
here may be different from $\bar{\bm{\alpha}}$. Furthermore, we are
interested in conducting inference on general null hypotheses of the
form $H_{0}:\mu_{1}-\mu_{0}=\bar{\mu}/\sqrt{n}$; these are different
from the null hypothesis of no average treatment effect used to motivate
the group sequential design.} We can then invert these tests to obtain one-sided confidence intervals
for the treatment effect $\mu_{1}-\mu_{0}$. The limit experiment
in this setting consists of observing $x(t)\sim\sigma^{-1}\mu t+W(t)$,
where $\mu:=\mu_{1}-\mu_{0}$, along with a discrete stopping time
$\tau\in\{1,\dots,T\}$ such that $\{\tau>t-1\}$ if and only if $x(l)\in\mathcal{I}_{l}$
for all $l=1,\dots,t-1$. Let $\mathbb{P}_{\mu}(\cdot)$ denote the
induced probability measure over the sample paths of $x(\cdot)$ between
$0$ and $T$, and $\mathbb{E}_{\mu}[\cdot]$ its corresponding expectation.
In view of the results in Section \ref{subsec:Conditional-inference},
the optimal level-$\bm{\alpha}$ test $\varphi^{*}(\cdot)$ of $H_{0}:\mu=\bar{\mu}$
vs $H_{1}:\mu>\bar{\mu}$ in the limit experiment is given by 
\begin{equation}
\varphi^{*}(\tau,x(\tau))=\begin{cases}
1 & \textrm{if }\mathbb{P}_{\bar{\mu}}(\tau=t)\le\alpha_{t}\\
\mathbb{I}\left\{ x(t)\ge\gamma(t)\right\}  & \textrm{if }\mathbb{P}_{\bar{\mu}}(\tau=t)>\alpha_{t},
\end{cases}\label{eq:alpha-spending optimal tests}
\end{equation}
where $\gamma(t)$ is chosen such that $\mathbb{E}_{\bar{\mu}}[\varphi^{*}(\tau,x(\tau))\vert\tau=t]=\alpha_{t}/\mathbb{P}_{\bar{\mu}}(\tau=t)$.

A finite sample version, $\hat{\varphi}$, of this test can be constructed
by replacing $\tau,x(\tau)$ in $\varphi^{*}$ with $\hat{\tau},x_{n}(\hat{\tau})$.
The resulting test would be asymptotically optimal under a suitable
non-parametric version of the $\bm{\alpha}$-spending requirement.
We refer to Appendix \ref{subsec:Supporting-information-for-GST}
for the details and for the proof that $\hat{\varphi}$ is asymptotically
optimal, in the sense that it attains the power of $\varphi^{*}$
in the limit experiment. A two-sided test for $H_{0}:\mu_{1}-\mu_{0}=\bar{\mu}/\sqrt{n}$
vs $H_{1}:\mu_{1}-\mu_{0}\neq\bar{\mu}/\sqrt{n}$ can be similarly
constructed by imposing a conditional unbiasedness restriction as
in Section \ref{subsec:Conditionally-unbiased-tests}. 

\subsubsection{Numerical Illustration}

To illustrate the methodology, consider a group sequential trial based
on the widely-used design of \citet{o1979multiple}, with $T=2$ stages.
This corresponds to setting $\mathcal{I}_{1}=[-2.797,2.797]$. We
would like to test $H_{0}:\mu_{1}-\mu_{0}=\bar{\mu}/\sqrt{n}$ vs
$H_{1}:\mu_{1}-\mu_{0}>\bar{\mu}/\sqrt{n}$ at the spending level
$(\alpha/\mathbb{P}_{\bar{\mu}}(\tau=1),\alpha/\mathbb{P}_{\bar{\mu}}(\tau=2))$,
equivalent to a conditional size constraint, $\mathbb{P}_{\bar{\mu}}(\varphi=1\vert\tau=t)=\alpha\ \forall\ t$.
Figure \ref{fig:GST} Panel A plots the thresholds, $(\gamma(1),\gamma(2))$,
for this test under $\alpha=0.05$ and $\sigma_{1}=\sigma_{0}=1$.
Unsurprisingly, the thresholds are increasing in $\bar{\mu}$ , but
it is interesting to observe that they cross at some $\bar{\mu}$. 

To describe the finite sample performance of this test, we ran Monte-Carlo
simulations with $Y_{i}^{(1)}=\bar{\mu}/\sqrt{n}+\epsilon_{i}^{(1)}$
and $Y_{i}^{(0)}=\epsilon_{i}^{(0)}$ where $\epsilon_{i}^{(1)},\epsilon_{i}^{(0)}\sim\sqrt{3}\times\textrm{Uniform}[-1,1]$.
The treatments were sampled in equal proportions $(\pi=1/2$). Since
$\sigma_{1},\sigma_{0}$ are unknown in practice, we estimate them
using data from the first stage. Figure \ref{fig:GST}, Panel B plots
the overall size of the test (which is the sum of the $\alpha$-spending
values at each stage) for different values of $n$ and $\bar{\mu}$
under the nominal $\alpha$-spending level of $(0.05/\mathbb{P}_{\bar{\mu}}(\tau=1),0.05/\mathbb{P}_{\bar{\mu}}(\tau=2))$.
We see that the asymptotic approximation worsens for larger values
of $\bar{\mu}$, but overall, the size is close to nominal even for
relatively small values of $n$. 

\begin{figure}
\includegraphics[height=5cm]{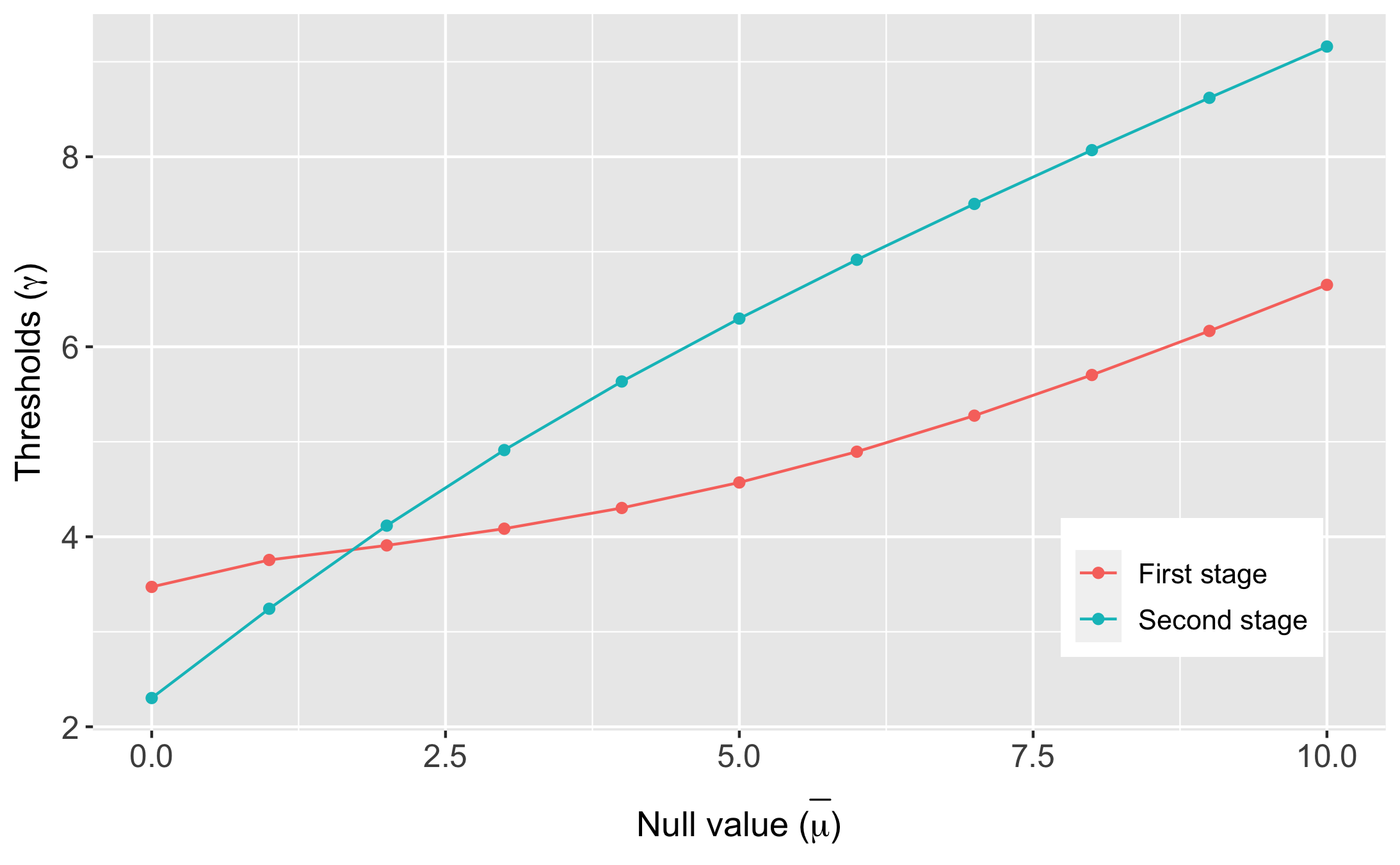}~\includegraphics[height=5cm]{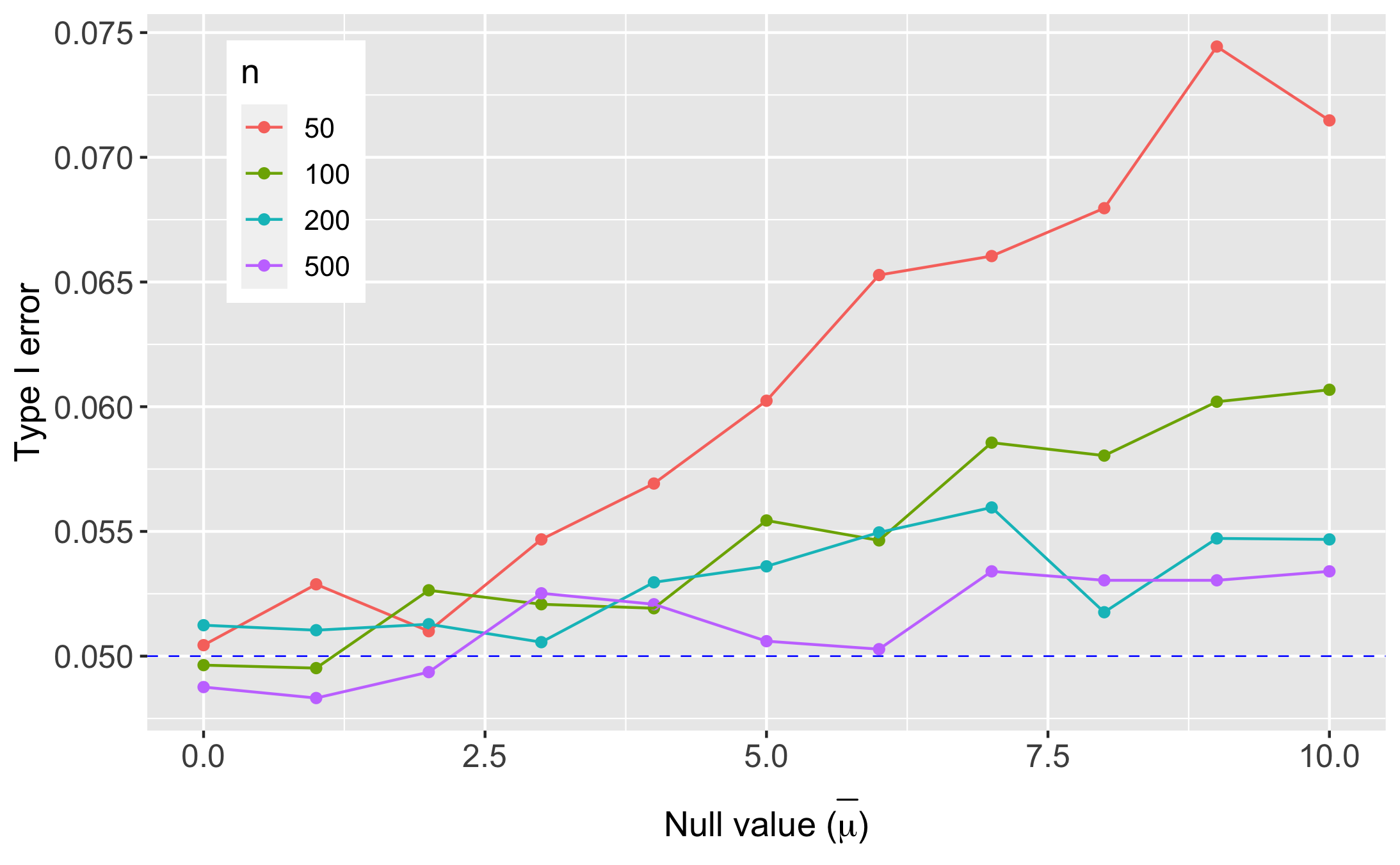}~~

\begin{tabular}{>{\centering}p{7.5cm}>{\centering}p{7.5cm}}
{\scriptsize{}A: Critical values} & {\scriptsize{}$\ $B: Finite sample size}\tabularnewline
\end{tabular}

\vspace{0.5em}
\begin{raggedright}
{\scriptsize{}Note: Panel A plots the threshold values in each stage
for the optimal, one-sided, level-$\bm{\alpha}$ test, (\ref{eq:alpha-spending optimal tests}),
at the $(0.05/\mathbb{P}_{\bar{\mu}}(\tau=1),0.05/\mathbb{P}_{\bar{\mu}}(\tau=2))$
spending level. Panel B plots the overall type-I error in finite samples
for different values of $n$ and null values, $\bar{\mu}$, when the
errors are drawn from a $\sqrt{3}\times\textrm{Uniform}[-1,1]$ distribution
for each treatment. }{\scriptsize\par}
\par\end{raggedright}
\caption{Testing in group sequential experiments\label{fig:GST}}
\end{figure}

\subsection{Bandit experiments}

Here, we describe inferential procedures for the batched Thompson-sampling
algorithm. For illustration, we employ $K=2$ treatments and $J=10$
batches. Let $(\bar{\mu}_{1},\bar{\mu}_{0})$ and $(\sigma_{1}^{2},\sigma_{0}^{2})$
denote the population means and variances for each treatment. For
simplicity, we take $\sigma_{1}^{2}=\sigma_{0}^{2}=1$. The limit
experiment can be described as follows: Suppose the decision maker
(DM) employs the sampling rule $\pi_{j}^{(a)}$ in batch $j$. The
DM then observes $Z_{j}^{(a)}\sim\mathcal{N}\left(\bar{\mu}_{a}\pi_{a},\pi_{a}\right)$
for $a\in\{0,1\}$ and updates the state variables $x_{a},q_{a}$
(which are initially set to $0$) as 
\[
x_{a}\leftarrow x_{a}+Z_{j}^{(a)},\quad q_{a}\leftarrow q_{a}+\pi_{a}.
\]
Under an under-smoothed prior, suggested by \citet{wager2021diffusion},
the Thompson sampling rule in batch $j+1$ is
\[
\pi_{j+1}^{(1)}=\Phi\left(\frac{q_{1}^{-1}x_{1}-q_{0}^{-1}x_{0}}{\sqrt{j/q_{1}q_{0}}}\right).
\]
We set $\pi_{1}^{(a)}=1/2$ for first batch. In what follows, we let
$\mu_{a}:=J\bar{\mu}_{a}$. We are interested in testing $H_{0}:(\mu_{1},\mu_{0})=(0,0)$.

Figure \ref{fig:Power-envelope-for}, Panel A plots the asymptotic
power envelope for testing $H_{0}:(\mu_{1},\mu_{2})=(0,0)$. Clearly,
the envelope is not symmetric; distinguishing $(a,0)$ from $(0,0)$
is easier than distinguishing $(-a,0)$ from $(0,0)$ for any $a>0$.
This is because of the asymmetry in treatment allocation under Thompson
sampling; under $(-a,0)$, treatment 1 is sampled more often than
treatment $0$ but the data from treatment $1$ is uninformative for
distinguishing $(-a,0)$ from $(0,0)$. 

\begin{figure}
\includegraphics[height=7.5cm]{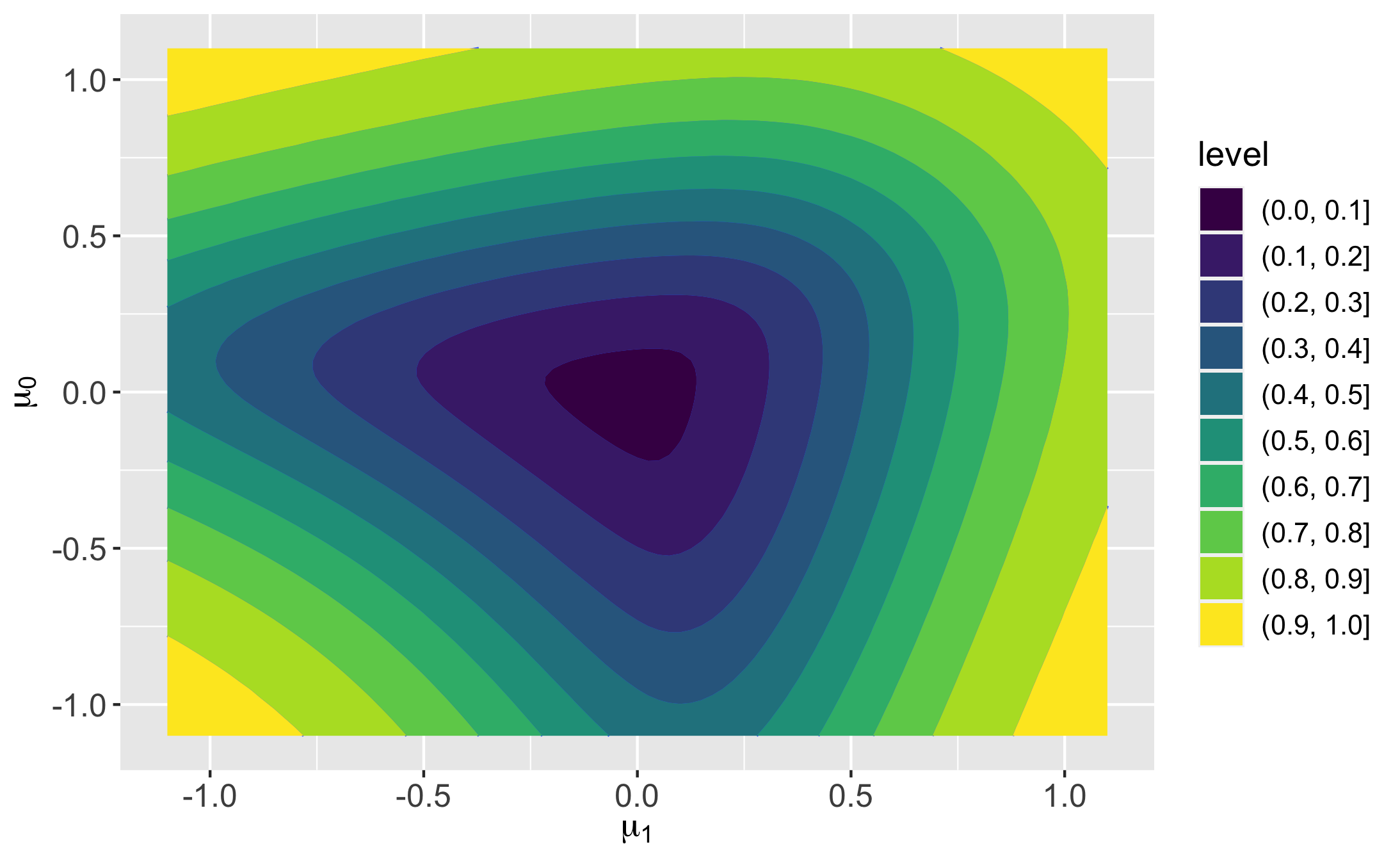}

\vspace{0.5em}
\begin{raggedright}
{\scriptsize{}Note: The figure plots the asymptotic power envelope
for any test of $H_{0}:(\mu,\mu)=(0,0)$ against different values
$(\mu_{1},\mu_{0})$ under the alternative. }{\scriptsize\par}
\par\end{raggedright}
\caption{Power envelope for Thompson-sampling with $10$ batches\label{fig:Power-envelope-for}}
\end{figure}

\subsubsection{Numerical illustration}

To determine the accuracy of our asymptotic approximations, we ran
Monte-Carlo simulations with $Y_{i}^{(a)}=\mu_{a}+\epsilon_{i}^{(a)}$
where $\epsilon_{i}^{(1)},\epsilon_{i}^{(0)}\sim\sqrt{3}\times\textrm{Uniform}[-1,1]$.
Figure \ref{fig:TS:Monte-Carlo}, Panel A plots the finite sample
performance of the Neyman-Pearson tests in the limit experiment for
testing $H_{0}:(\mu_{1},\mu_{0})=(0,0)$ vs $H_{1}:(\mu_{1},\mu_{0})=(\mu,\mu)$
under various values of $\mu$ (due to symmetry, we only report the
results for positive $\mu$). Panel B repeats the same calculation,
but against alternatives of the form $H_{1}:(\mu,0)$. As noted earlier,
power is higher here for $\mu>0$ as opposed to $\mu<0$. Both plots
show that the asymptotic approximation is quite accurate even for
$n$ as small as $20$ (note that the number of batches is $10$,
so this corresponds to $200$ observations overall). The approximation
is somewhat worse for testing $\mu<0$; this is because Thompson-sampling
allocates much fewer units to treatment 0 in this instance, even though
it is only data from this treatment that is informative for distinguishing
the two hypotheses. 

\begin{figure}
\includegraphics[height=5cm]{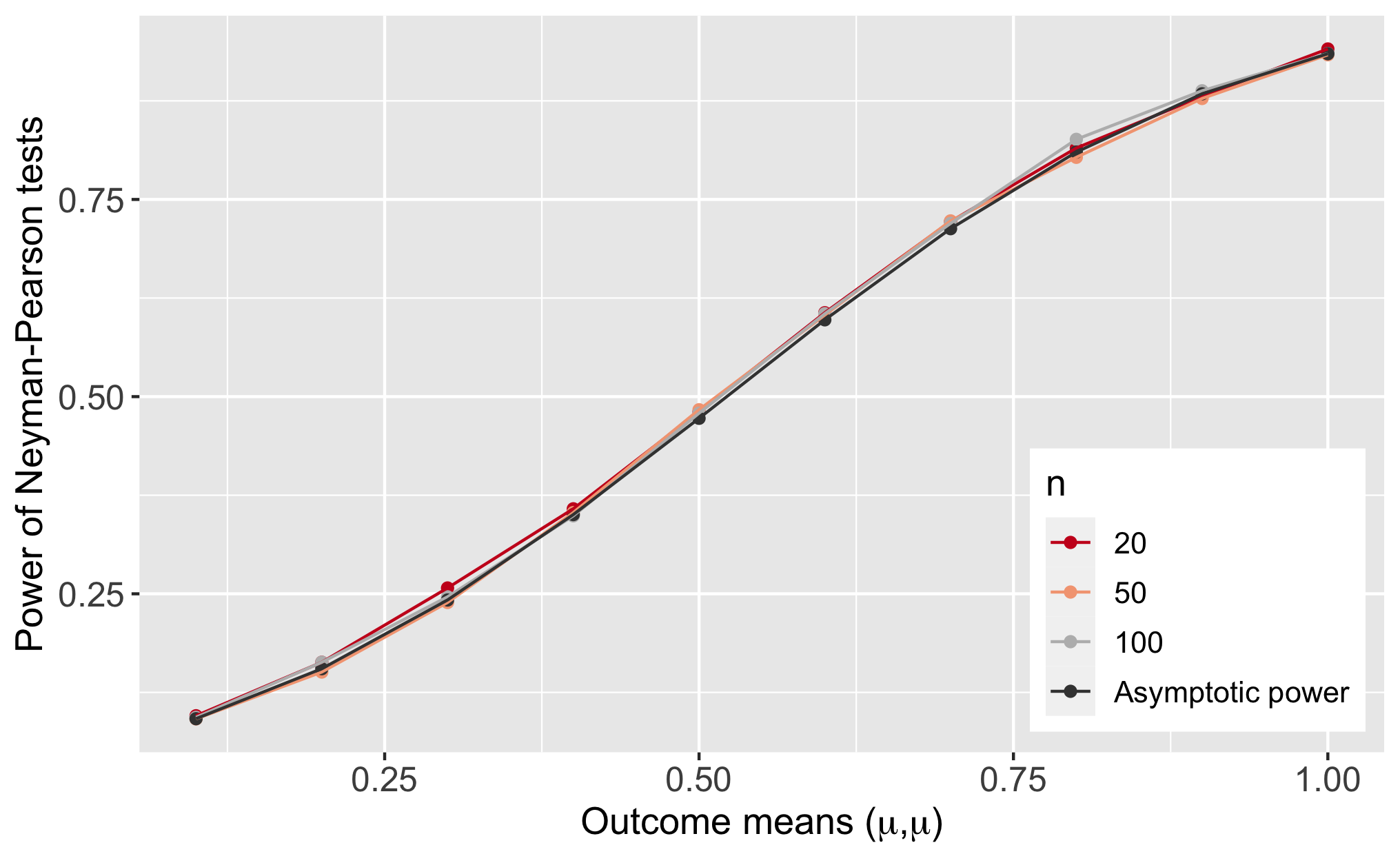}~\includegraphics[height=5cm]{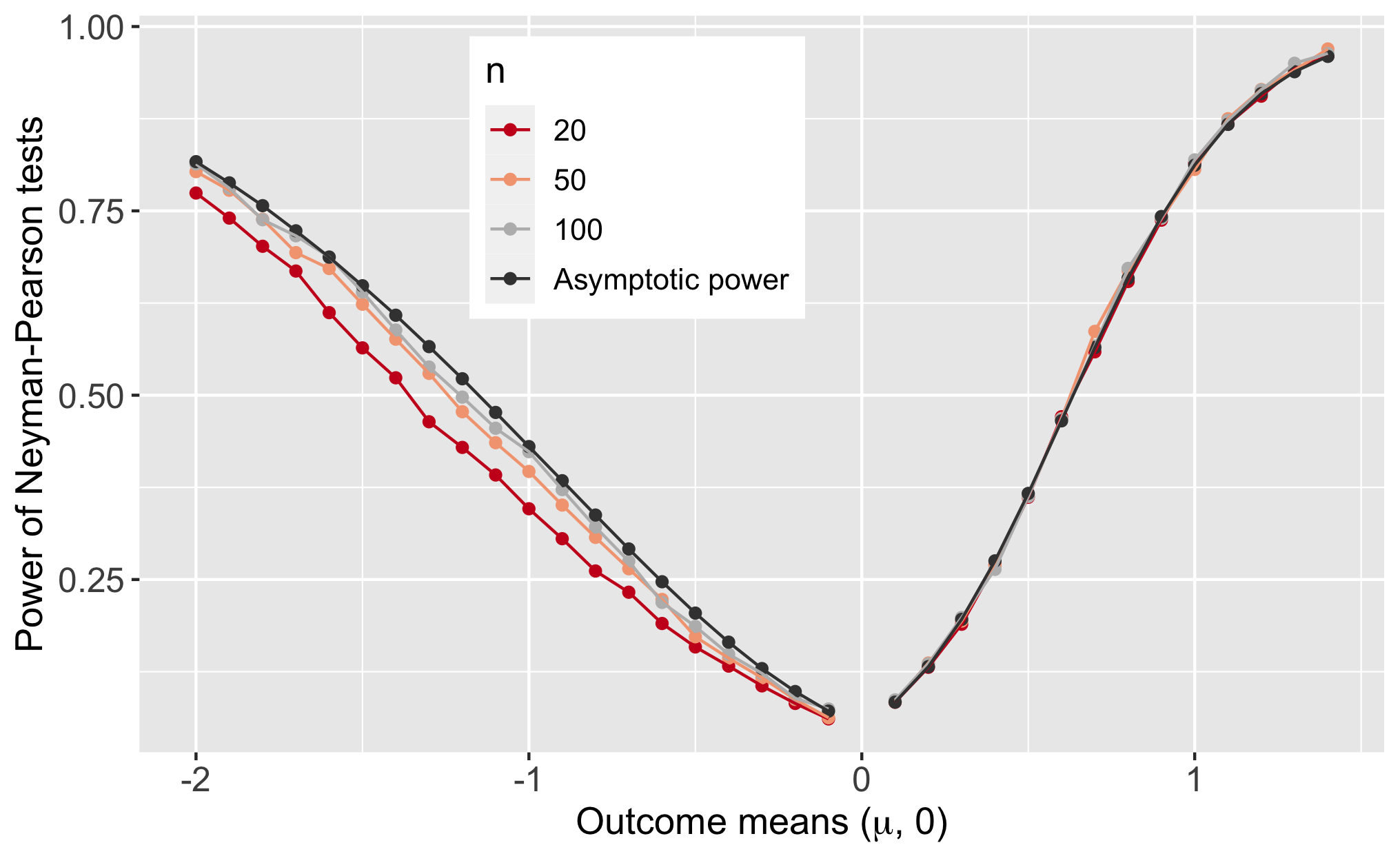}~~

\begin{tabular}{>{\centering}p{7.5cm}>{\centering}p{7.5cm}}
{\scriptsize{}A: Power against $H_{1}:(\mu,\mu)$} & {\scriptsize{}$\ $B: Power against $H_{1}:(\mu,0)$}\tabularnewline
\end{tabular}

\vspace{0.5em}
\begin{raggedright}
{\scriptsize{}Note: Panel A plots the finite sample power of Neyman-Pearson
tests at the nominal $5$\% level (solid blue line) for testing $H_{0}:(\mu_{1},\mu_{0})=(0,0)$
against $H_{1}:(\mu_{1},\mu_{0})=(\mu,\mu)$ when the errors are drawn
from a $\sqrt{3}\times\textrm{Uniform}[-1,1]$ distribution for each
treatment. Panel B repeats the same calculation for alternatives of
the form $H_{1}:(\mu_{1},\mu_{0})=(\mu,0)$. Both panels also display
the asymptotic power envelope.}{\scriptsize\par}
\par\end{raggedright}
\caption{Finite sample performance of Neyman-Pearson tests in bandit experiments\label{fig:TS:Monte-Carlo}}
\end{figure}

\section{Conclusion}

Conducting inference after sequential experiments is a challenging
task. However, significant progress can be made by analyzing the optimal
inference problem under an appropriate limit experiment. We showed
that the data from any sequential experiment can be condensed into
a finite number of sufficient statistics, while still maintaining
the power of tests. Furthermore, we were able to establish uniquely
optimal tests under reasonable constraints such as unbiasedness and
$\bm{\alpha}$-spending, in both parametric and non-parametric regimes.
Taken together, these findings offer a comprehensive framework for
conducting optimal inference following sequential experiments.

Despite these results, there are still several avenues for future
research. While we believe that our results for experiments with adaptive
sampling rules apply without batching, this needs be formally verified.
Our characterization of uniquely optimal tests is also limited in
this context, as $\bm{\alpha}$-spending restrictions are not feasible.
Therefore, exploring other types of testing considerations such as
invariance or conditional inference may be worthwhile. We believe
that the techniques developed in this paper will prove useful for
analyzing these other types of tests.

\bibliographystyle{IEEEtranSN}
\bibliography{Optimal_sequential_tests}

\newpage{}

\appendix

\section{Proofs\label{sec:Appendix:A}}

\subsection{Proof of Theorem \ref{Thm: ART}}

To prove the first claim, observe that both $\hat{\tau}$ and $x_{n}(\hat{\tau})$
are tight under $P_{nT,0}$: the former by Assumption 2, and the latter
by the fact $\max_{t\le T}x_{n}(t)$ is tight (by the continuous mapping
theorem it converges to the tight limit $\max_{t}x(t)$ under $P_{nT,0}$).
Hence, the joint $(\hat{\tau},x_{n}(\hat{\tau}))$ is a also tight,
and by Prohorov's theorem, converges in distribution under sub-sequences.
The first part of the theorem then follows from \citet[Theorem 1]{LeCam1979}. 

To prove the second claim, denote ${\bf y}_{nt}=(Y_{1},\dots,Y_{nt})$.
Defining
\[
\ln\frac{dP_{nt,h}}{dP_{nt,0}}({\bf y}_{nt})=\sum_{i=1}^{\left\lfloor nt\right\rfloor }\ln\frac{dp_{\theta_{0}+h/\sqrt{n}}}{dp_{\theta_{0}}}(Y_{i}),
\]
we have by the SLAN property, (\ref{eq:SLAN property}), and Assumption
1(i) that
\[
\ln\frac{dP_{n\hat{\tau},h}}{dP_{n\hat{\tau},0}}({\bf y}_{n\hat{\tau}})=h^{\intercal}I^{1/2}x_{n}(\hat{\tau})-\frac{\hat{\tau}}{2}h^{\intercal}Ih+o_{P_{nT,0}}(1).
\]
Combining the above with the first part of the theorem gives 
\begin{equation}
\ln\frac{dP_{n\hat{\tau},h}}{dP_{n\hat{\tau},0}}({\bf y}_{n\hat{\tau}})\xrightarrow[P_{nT,0}]{d}h^{\intercal}I^{1/2}x(\tau)-\frac{\tau}{2}h^{\intercal}Ih,\label{pf:Thm:ART:eq:1}
\end{equation}
where $x(\cdot)$ has the same distribution as $d$-dimensional Brownian
motion. 

Now, $\varphi_{n}$ is tight since $\varphi_{n}\in[0,1]$. Together
with (\ref{pf:Thm:ART:eq:1}), this implies the joint $\left(\varphi_{n},\ln\frac{dP_{n\hat{\tau},h}}{dP_{n\hat{\tau},0}}({\bf y}_{n\hat{\tau}})\right)$
is also tight. Hence, by Prohorov's theorem, given any sequence $\{n_{j}\}$,
there exists a further sub-sequence $\{n_{j_{m}}\}$ - represented
as $\{n\}$ without loss of generality - such that 
\begin{equation}
\left(\begin{array}{c}
\varphi_{n}\\
\frac{dP_{n\hat{\tau}.,h}}{dP_{n\hat{\tau},0}}\left({\bf y}_{n\hat{\tau}}\right)
\end{array}\right)\xrightarrow[P_{nT,0}]{d}\left(\begin{array}{c}
\bar{\varphi}\\
V
\end{array}\right);\quad V\sim\exp\left\{ h^{\intercal}I^{1/2}x(\tau)-\frac{\tau}{2}h^{\intercal}Ih\right\} ,\label{eq:pf:Thm1:weak convergence}
\end{equation}
where $\bar{\varphi}\in[0,1]$. It is a well known property of Brownian
motion that $M(t):=\exp\left\{ h^{\intercal}I^{1/2}x(t)-\frac{t}{2}h^{\intercal}Ih\right\} $
is a martingale with respect to the filtration $\mathcal{F}_{t}$.
Since $\tau$ is an $\mathcal{F}_{t}$-adapted stopping time, the
optional stopping theorem then implies $E[V]\equiv E[M(\tau)]=E[M(0)]=1$. 

We now claim that 
\begin{equation}
\varphi_{n}\xrightarrow[P_{nT,h}]{d}L;\ \textrm{where }L(B):=E[\mathbb{I}\{\bar{\varphi}\in B\}V]\ \forall\ B\in\mathcal{B}(\mathbb{R}).\label{eq:pf:Thm1:weak convergence 2}
\end{equation}
It is clear from $V\ge0$ and $E[V]=1$ that $L(\cdot)$ is a probability
measure, and that for every measurable function $f:\mathbb{R}\to\mathbb{R}$,
$\int fdL=E[f(\bar{\varphi})V]$. Furthermore, for any $f(\cdot)$
lower-semicontinuous and non-negative, 
\begin{align*}
\lim\inf\mathbb{E}_{nT,h}[f(\varphi_{n})] & \ge\lim\inf\mathbb{E}_{nT,0}\left[f(\varphi_{n})\frac{dP_{nT,h}}{dP_{nT,0}}\right]\\
 & =\lim\inf\mathbb{E}_{nT,0}\left[f(\varphi_{n})\frac{dP_{n\hat{\tau},h}}{dP_{n\hat{\tau},0}}\right]\ge E[f(\bar{\varphi})V],
\end{align*}
where the equality follows from the law of iterated expectations since
$\varphi_{n}$ is a function only of ${\bf y}_{n\hat{\tau}}$ and
$dP_{n\hat{\tau},h}/dP_{n\hat{\tau},0}$ is a martingale under $P_{nT,0}$;
and the last inequality follows from applying the portmanteau lemma
on (\ref{eq:pf:Thm1:weak convergence}). Finally, applying the portmanteau
lemma again, in the converse direction, gives (\ref{eq:pf:Thm1:weak convergence 2}). 

Since $\varphi_{n}$ is bounded, (\ref{eq:pf:Thm1:weak convergence 2})
implies
\begin{equation}
\lim_{n\to\infty}\beta_{n}(h):=\lim_{n\to\infty}\mathbb{E}_{nT,h}\left[\varphi_{n}\right]=E\left[\bar{\varphi}e^{h^{\intercal}I^{1/2}x(\tau)-\frac{\tau}{2}h^{\intercal}Ih}\right].\label{eq:pf:Thm:ART:2}
\end{equation}
Define $\varphi(\tau,x(\tau)):=E[\bar{\varphi}\vert\tau,x(\tau)]$;
this is a test statistic since $\varphi\in[0,1]$. The right hand
side of (\ref{eq:pf:Thm:ART:2}) then becomes 
\[
E\left[\varphi(\tau,x(\tau))e^{h^{\intercal}I^{1/2}x(\tau)-\frac{\tau}{2}h^{\intercal}Ih}\right].
\]
But by the Girsanov theorem, this is just the expectation, $\mathbb{E}_{h}[\varphi(\tau,x(\tau))]$,
of $\varphi(\tau,x(\tau))$ when $x(t)$ is distributed as a Gaussian
process with drift $I^{1/2}h$, i.e., when $x(t)\sim I^{1/2}ht+W(t)$. 

\subsection{Proof of Proposition \ref{Prop: Linear combinations}}

We start by proving the first claim. Denote $H_{0}\equiv\{h:a^{\intercal}h=0\}$
and $H_{1}\equiv\{h:a^{\intercal}h=c\}$. Let $\mathbb{P}_{h}$ denote
the induced probability measure over the sample paths generated by
$x(t)\sim I^{1/2}ht+W(t)$ between $t\in[0,T]$. As before, $\mathcal{F}_{t}$
denotes the filtration generated by $\{U,x(s):s\le t\}$. Given any
$h_{1}\in H_{1}$, define $h_{0}=h_{1}-(a^{\intercal}h_{1}/a^{\intercal}I^{-1}a)I^{-1}a$.
Note that $a^{\intercal}h_{1}=c$ and $h_{0}\in H_{0}$. Let $\ln\frac{d\mathbb{P}_{h_{1}}}{d\mathbb{P}_{h_{0}}}(\mathcal{F}_{t})$
denote the likelihood ratio between the probabilities induced by the
parameters $h_{1},h_{0}$ over the filtration $\mathcal{F}_{t}$.
By the Girsanov theorem,
\begin{align*}
\ln\frac{d\mathbb{P}_{h_{1}}}{d\mathbb{P}_{h_{0}}}(\mathcal{F}_{\tau}) & =\left(h_{1}^{\intercal}I^{1/2}x(\tau)-\frac{\tau}{2}h_{1}^{\intercal}Ih_{1}\right)-\left(h_{0}^{\intercal}I^{1/2}x(\tau)-\frac{\tau}{2}h_{0}^{\intercal}Ih_{0}\right)\\
 & =\frac{1}{\sigma}c\tilde{x}(\tau)-\frac{c^{2}}{2\sigma^{2}}\tau,
\end{align*}
where $\tilde{x}(t):=\sigma^{-1}a^{\intercal}I^{-1/2}x(t)$. Hence,
an application of the Neyman-Pearson lemma shows that the UMP test
of $H_{0}^{\prime}:h=h_{0}$ vs $H_{1}^{\prime}:h=h_{1}$ is given
by 
\[
\varphi_{c}^{*}=\mathbb{I}\left\{ c\tilde{x}(\tau)-\frac{c^{2}}{2\sigma}\tau\ge\gamma\right\} ,
\]
where $\gamma$ is chosen by the size requirement. Now, for any $h_{0}\in H_{0}$,
\[
\tilde{x}(t)\equiv\sigma^{-1}a^{\intercal}I^{-1/2}x(t)\sim W(t).
\]
Hence, the distribution of the sample paths of $\tilde{x}(\cdot)$
is independent of $h_{0}$ under the null. Combined with the assumption
that $\tau$ is $\tilde{\mathcal{F}}_{t}$-adapted, this implies $\varphi_{c}^{*}$
does not depend on $h_{0}$ and, by extension, $h_{1}$, except through
$c$. Since $h_{1}\in H_{1}$ was arbitrary, we are led to conclude
$\varphi_{c}^{*}$ is UMP more generally for testing $H_{0}:a^{\intercal}h=0$
vs $H_{1}:a^{\intercal}h=c$. 

The second claim is an easy consequence of the first claim and Theorem
\ref{Thm: ART}.

\subsection{Proof of Proposition \ref{Prop: Unbiased}}

By the Girsanov theorem,
\[
\beta(h):=\mathbb{E}_{h}[\varphi]=\mathbb{E}_{0}\left[\varphi(\tau,x(\tau))e^{h^{\intercal}I^{1/2}x(\tau)-\frac{\tau}{2}h^{\intercal}Ih}\right].
\]
It can be verified from the above that $\beta(h)$ is differentiable
around $h=0$. But unbiasedness requires $\mathbb{E}_{h}[\varphi]\ge\alpha$
for all $h$ and $\mathbb{E}_{0}[\varphi]=\alpha$. This is only possible
if $\beta^{\prime}(0)=0$, i.e., $\mathbb{E}_{0}[x(\tau)\varphi(\tau,x(\tau))]=0$. 

\subsection{Proof of Theorem \ref{Thm: ART-conditional inference}}

Since $\hat{\tau}$ is bounded, it follows by similar arguments as
in the proof of Theorem \ref{Thm: ART} that $\left(\varphi_{n},\hat{\tau},\ln\frac{dP_{n\hat{\tau},h}}{dP_{n\hat{\tau},0}}({\bf y}_{n\hat{\tau}})\right)$
is tight. Consequently, by Prohorov's theorem, given any sequence
$\{n_{j}\}$, there exists a further sub-sequence $\{n_{j_{m}}\}$
- represented as $\{n\}$ without loss of generality - such that 
\begin{equation}
\left(\begin{array}{c}
\varphi_{n}\\
\hat{\tau}\\
\frac{dP_{n\hat{\tau}.,h}}{dP_{n\hat{\tau},0}}\left({\bf y}_{n\hat{\tau}}\right)
\end{array}\right)\xrightarrow[P_{nT,0}]{d}\left(\begin{array}{c}
\bar{\varphi}\\
\tau\\
V
\end{array}\right);\quad V\sim\exp\left\{ h^{\intercal}I^{1/2}x(\tau)-\frac{\tau}{2}h^{\intercal}Ih\right\} .\label{eq:pf:Thm1:weak convergence-2}
\end{equation}
It then follows as in the proof of Theorem \ref{Thm: ART} that 
\begin{equation}
\left(\begin{array}{c}
\varphi_{n}\\
\hat{\tau}
\end{array}\right)\xrightarrow[P_{nT,h}]{d}L;\ \textrm{where }L(B):=E[\mathbb{I}\{(\bar{\varphi},\tau)\in B\}V]\ \forall\ B\in\mathcal{B}(\mathbb{R}^{2}).\label{eq:pf:Thm1:weak convergence 2-1}
\end{equation}
The above in turn implies 
\begin{align}
\lim_{n\to\infty}\mathbb{E}_{nT,h}\left[\varphi_{n}\mathbb{I}\{\hat{\tau}=t\}\right] & =E\left[\bar{\varphi}\mathbb{I}\{\tau=t\}e^{h^{\intercal}I^{1/2}x(\tau)-\frac{\tau}{2}h^{\intercal}Ih}\right],\ \textrm{and }\label{eq:pf:ART-conditional-inference-2}\\
\lim_{n\to\infty}\mathbb{E}_{nT,h}\left[\mathbb{I}\{\hat{\tau}=t\}\right] & =E\left[\mathbb{I}\{\tau=t\}e^{h^{\intercal}I^{1/2}x(\tau)-\frac{\tau}{2}h^{\intercal}Ih}\right].\label{eq:pf:ART-conditional-inference-3}
\end{align}
for every $t\in\{1,2,\dots,T\}$.

Denote $\varphi(\tau,x(\tau))=E[\bar{\varphi}\vert\tau,x(\tau)]$;
this is a level-$\bm{\alpha}$ test, as can be verified by setting
$h=0$ in (\ref{eq:pf:ART-conditional-inference-2}). The right hand
side of (\ref{eq:pf:ART-conditional-inference-2}) then becomes 
\[
E\left[\varphi(\tau,x(\tau))\mathbb{I}\{\tau=t\}e^{h^{\intercal}I^{1/2}x(\tau)-\frac{\tau}{2}h^{\intercal}Ih}\right].
\]
An application of the Girsanov theorem then shows that the right hand
sides of (\ref{eq:pf:ART-conditional-inference-2}) and (\ref{eq:pf:ART-conditional-inference-3})
are just the expectations, $\mathbb{E}_{h}[\varphi(\tau,x(\tau))\mathbb{I}\{\tau=t\}]$
and $\mathbb{E}_{h}[\mathbb{I}\{\tau=t\}]$ when $x(t)\sim I^{1/2}ht+W(t)$.
What is more, the measures $\mathbb{P}_{0}(\cdot),\mathbb{P}_{h}(\cdot)$
are absolutely continuous, so $\mathbb{P}_{0}(\tau=t)=0$ if and only
if $\mathbb{P}_{h}(\tau=t)=0$ for any $h\in\mathbb{R}^{d}$. We are
thus led to conclude that
\[
\lim_{n\to\infty}\beta_{n}(h\vert t):=\lim_{n\to\infty}\frac{\mathbb{E}_{nT,h}\left[\varphi_{n}\mathbb{I}\{\hat{\tau}=t\}\right]}{\mathbb{E}_{nT,h}\left[\mathbb{I}\{\hat{\tau}=t\}\right]}=\frac{\mathbb{E}_{h}\left[\varphi_{n}\mathbb{I}\{\hat{\tau}=t\}\right]}{\mathbb{E}_{h}\left[\mathbb{I}\{\hat{\tau}=t\}\right]}:=\beta(h\vert t)
\]
for every $h\in\mathbb{R}^{d}$, and $t\in\{1,2,\dots,T\}$ satisfying
$\mathbb{P}_{0}(\tau=t)\neq0$. This proves the desired claim.

\subsection{Proof of Proposition \ref{Prop: Non-parametric}}

Fix some arbitrary $g_{1}\in T(P_{0})$. To simplify matters, we set
$\delta=1$. The case of general $\delta$ can be handled by simply
replacing $g_{1}$ with $g_{1}/\delta$. By standard results for Hilbert
spaces, we can write $g_{1}=\sigma^{-1}\left\langle \psi,g\right\rangle (\psi/\sigma)+\tilde{g}_{1}$,
where $\tilde{g}_{1}\perp(\psi/\sigma)$ . Define $\bm{g}:=\left(\psi/\sigma,\tilde{g}_{1}/\left\Vert \tilde{g}_{1}\right\Vert \right)^{\intercal}$,
and consider sub-models of the form $P_{1/\sqrt{n},\bm{h}^{\intercal}\bm{g}}$
for $\bm{h}\in\mathbb{R}^{2}$. By (\ref{eq:SLAN nonparametric setting}),
\begin{equation}
\sum_{i=1}^{\left\lfloor nt\right\rfloor }\ln\frac{dP_{1/\sqrt{n},\bm{h}^{\intercal}\bm{g}}}{dP_{0}}(Y_{i})=\frac{\bm{h}^{\intercal}}{\sqrt{n}}\sum_{i=1}^{\left\lfloor nt\right\rfloor }\bm{g}(Y_{i})-\frac{t}{2}\bm{h}^{\intercal}\bm{h}+o_{P_{nT,0}}(1),\ \textrm{ uniformly over }t.\label{eq:SLAN:parametric sub-models}
\end{equation}
Comparing with (\ref{eq:SLAN property}), we observe that $\left\{ P_{1/\sqrt{n},\bm{h}^{\intercal}\bm{g}}:\bm{h}\in\mathbb{R}^{2}\right\} $
is equivalent to a parametric model with score $\bm{g}(\cdot)$ and
local parameter $\bm{h}$ (note that $\mathbb{E}_{P_{0}}[\bm{g}\bm{g}^{\intercal}]=I$).
Let $G_{n}(t):=n^{-1/2}\sum_{i=1}^{n}\bm{g}(Y_{i})$ denote the score
process. By the functional central limit theorem, $G_{n}(t)\xrightarrow[P_{nT,0}]{d}G(t)\equiv(x(t),\tilde{G}(t))$,
where $x(\cdot),\tilde{G}(\cdot)$ are independent one-dimensional
Brownian motions. Take $\mathcal{G}_{t}:=\sigma\{G(s):s\le t\}$,
$\mathcal{F}_{t}:=\sigma\{x(s):s\le t\}$ to be the filtrations generated
by $G(\cdot)$ and $x(\cdot)$ respectively until time $t$. Since
the first component of $G_{n}(\cdot)$ is $x_{n}(\cdot)$ and $\hat{\tau}=\tau(x_{n}(\cdot))$
by Assumption 3(ii), the extended continuous mapping theorem implies
\begin{equation}
(G_{n}(\hat{\tau}),\hat{\tau})\xrightarrow[P_{nT,0}]{d}(G(\tau),\tau),\label{eq:weak convergence: parametric sub-models}
\end{equation}
where $\tau$ is a $\mathcal{F}_{t}$-adapted stopping time, and therefore,
$\mathcal{G}_{t}$-adapted by extension.

Consider the limit experiment where one observes a $\mathcal{F}_{t}$-adapted
stopping time $\tau$ along with a diffusion process $G(t):=\bm{h}t+W(t)$,
where $W(\cdot)$ is 2-dimensional Brownian motion. Using (\ref{eq:SLAN:parametric sub-models})
and (\ref{eq:weak convergence: parametric sub-models}), we can argue
as in the proof of Theorem \ref{Thm: ART} to show that any test in
the parametric model $\left\{ P_{1/\sqrt{n},\bm{h}^{\intercal}\bm{g}}:\bm{h}\in\mathbb{R}^{2}\right\} $
can be matched (along sub-sequences) by a test that depends only on
$G(\tau),\tau$ in the limit experiment. Hence, $\beta_{n}\left(\bm{h}^{\intercal}\bm{g}\right):=\int\varphi_{n}dP_{nT,\bm{h}^{\intercal}\bm{g}}$
converges along sub-sequences to the power function, $\beta(\bm{h})$,
of some test $\varphi(\tau,G(\tau))$ in the limit experiment. Note
that by our definitions, $\left\langle \psi,\bm{h}^{\intercal}\bm{g}\right\rangle $
is simply the first component of $\bm{h}$ divided by $\sigma$. This
in turn implies, as a consequence of the definition of asymptotically
level-$\alpha$ tests, that $\varphi(\cdot)$ is level-$\alpha$ for
testing $H_{0}:(1,0)^{\intercal}\bm{h}=0$ in the limit experiment. 

Now, by a similar argument as in the proof of Proposition \ref{Prop: Linear combinations},
along with the fact $(1,0)^{\intercal}G(t)=x(t)$, the optimal level-$\alpha$
test of $H_{0}:(1,0)^{\intercal}\bm{h}=0$ vs $H_{1}:(1,0)^{\intercal}\bm{h}=\mu_{1}/\sigma$
in the limit experiment is given by
\[
\varphi_{\mu_{1}}^{*}(\tau,x(\tau)):=\mathbb{I}\left\{ \mu_{1}x(\tau)-\frac{\mu_{1}^{2}}{2\sigma}\tau\ge\gamma\right\} .
\]
For all $\bm{h}\in H_{1}\equiv\{\bm{h}:(1,0)^{\intercal}\bm{h}=\mu_{1}/\sigma\}$
satisfying the alternative hypothesis,
\[
x(t)=(1,0)^{\intercal}G(t)\sim\sigma^{-1}\mu_{1}t+\tilde{W}(t),
\]
where $\tilde{W}(\cdot)$ is 1-dimensional Brownian motion. As $\tau$
is $\mathcal{F}_{t}$-adapted, the joint distribution of $(\tau,x(\tau))$
therefore depends only on $\mu_{1}$ for $\bm{h}\in H_{1}$. Consequently,
the power, $\mathbb{E}_{\bm{h}}[\varphi_{\mu_{1}}^{*}(\tau,x(\tau))]$,
of $\varphi_{\mu_{1}}^{*}(\cdot)$ against such alternatives depends
only on $\mu_{1}$, and is denoted by $\beta^{*}\left(\mu_{1}\right)$.
Since $\varphi_{\mu_{1}}^{*}(\cdot)$ is the optimal test and $\mu_{1}=\left\langle \psi,\bm{h}^{\intercal}\bm{g}\right\rangle $,
we conclude $\beta(\bm{h})\le\beta^{*}\left(\left\langle \psi,\bm{h}^{\intercal}\bm{g}\right\rangle \right)$.
This further implies $\lim\sup_{n}\beta_{n}(\bm{h}^{\intercal}\bm{g})\le\beta^{*}\left(\left\langle \psi,\bm{h}^{\intercal}\bm{g}\right\rangle \right)$
for any $\bm{h}\in\mathbb{R}^{2}$. Setting $\bm{h}=\left(\left\langle \psi,g_{1}\right\rangle /\sigma,\left\Vert \tilde{g}_{1}\right\Vert \right)^{\intercal}$
then gives $\lim\sup_{n}\beta_{n}(g_{1})\le\beta^{*}\left(\left\langle \psi,g_{1}\right\rangle \right)$.
Since $g_{1}\in T(P_{0})$ was arbitrary, the claim follows. 

\subsection{Proof of Proposition \ref{Prop: Two-sample non-parametric}}

For some arbitrary $\bm{g}=(g_{1},g_{0})\in T(P_{0}^{(1)})\times T(P_{0}^{(0)})$.
To simplify matters, we set $\delta_{1}=\delta_{0}=1$. The case of
general $\delta$ can be handled by simply replacing $g_{a}$ with
$g_{a}/\delta_{a}$. In what follows, let $\pi_{1}=\pi$ and $\pi_{0}=1-\pi$.
The vectors ${\bf y}_{nt}^{(1)}=(Y_{1}^{(1)},\dots,Y_{n\pi_{1}t}^{(1)})$
and ${\bf y}_{nt}^{(0)}=(Y_{1}^{(0)},\dots,Y_{n\pi_{0}t}^{(0)})$
denote the collection of outcomes from treatments 1 and 0 until time
$t$, and we set ${\bf y}_{nt}=({\bf y}_{nt}^{(1)},{\bf y}_{nt}^{(0)})$.
Define $P_{nt,\bm{g}}$ as the joint probability measure over ${\bf y}_{nt}$
when each $Y_{i}^{(a)}$ is an iid draw from $P_{1/\sqrt{n},g_{a}}^{(a)}$. 

As in the proof of Proposition \ref{Prop: Non-parametric}, we can
write $g_{a}=\sigma_{a}^{-1}\left\langle \psi_{a},g_{a}\right\rangle _{a}(\psi_{a}/\sigma_{a})+\tilde{g}_{a}$,
where $\tilde{g}_{a}\perp(\psi_{a}/\sigma_{a})$. Define $\bm{g}_{a}:=\left(\psi_{a}/\sigma_{a},\tilde{g}_{a}/\left\Vert \tilde{g}_{a}\right\Vert _{a}\right)^{\intercal}$,
and consider sub-models of the form $P_{1/\sqrt{n},\bm{h}_{1}^{\intercal}\bm{g}_{1}}\times P_{1/\sqrt{n},\bm{h}_{0}^{\intercal}\bm{g}_{0}}$
for $\bm{h}_{1},\bm{h}_{0}\in\mathbb{R}^{2}$. By the SLAN property,
(\ref{eq:SLAN nonparametric setting}), and the fact that the treatments
are independent,
\begin{align}
 & \ln\frac{dP_{nt,(\text{\ensuremath{\bm{h}_{1}^{\intercal}\bm{g}_{1}}},\bm{h}_{0}^{\intercal}\bm{g}_{0})}}{dP_{nt,\bm{0}}}({\bf y}_{nt})=\frac{\bm{h}_{1}^{\intercal}}{\sqrt{n}}\sum_{i=1}^{\left\lfloor n\pi_{1}t\right\rfloor }\bm{g}_{1}(Y_{i}^{(1)})-\frac{\pi_{1}t}{2}\bm{h}_{1}^{\intercal}\bm{h}_{1}+\dots\nonumber \\
 & \qquad\dots+\frac{\bm{h}_{0}^{\intercal}}{\sqrt{n}}\sum_{i=1}^{\left\lfloor n\pi_{0}t\right\rfloor }\bm{g}_{0}(Y_{i}^{(0)})-\frac{\pi_{0}t}{2}\bm{h}_{0}^{\intercal}\bm{h}_{0}+o_{P_{nT,0}}(1),\ \textrm{ uniformly over }t.
\end{align}

Let $G_{a,n}(t):=n^{-1/2}\sum_{i=1}^{\left\lfloor n\pi_{a}t\right\rfloor }\bm{g}_{a}(Y_{i}^{(a)})$
for $a\in\{0,1\}$. By a standard functional central limit theorem,
\[
G_{a,n}(t)\xrightarrow[P_{nT,\bm{0}}]{d}G_{a}(t)\equiv(z_{a}(t),\tilde{G}_{a}(t)),
\]
where $z_{a}(\cdot)/\sqrt{\pi_{a}},\tilde{G}_{a}(\cdot)/\sqrt{\pi_{a}}$
are independent 1-dimensional Brownian motions. Furthermore, since
the treatments are independent of each other, $G_{1}(\cdot),G_{0}(\cdot)$
are independent Gaussian processes. Define $\sigma^{2}:=\left(\frac{\sigma_{1}^{2}}{\pi_{1}}+\frac{\sigma_{0}^{2}}{\pi_{0}}\right)$,
\[
x(t):=\frac{1}{\sigma}\left(\frac{\sigma_{1}}{\pi_{1}}z_{1}(t)-\frac{\sigma_{0}}{\pi_{0}}z_{0}(t)\right)
\]
and take $\mathcal{G}_{t}:=\sigma\{(G_{1}(s),G_{0}(s)):s\le t\}$,
$\mathcal{F}_{t}:=\sigma\{x(s):s\le t\}$ to be the filtrations generated
by $\bm{G}(\cdot):=(G_{1}(\cdot),G_{0}(\cdot))$ and $x(\cdot)$ respectively
until time $t$. Using Assumption 3(ii), the extended continuous mapping
theorem implies 
\begin{equation}
(G_{1,n}(\hat{\tau}),G_{0,n}(\hat{\tau}),\hat{\tau})\xrightarrow[P_{nT,0}]{d}(G_{1}(\tau),G_{0}(\tau),\tau),\label{eq:weak convergence: parametric sub-models-1}
\end{equation}
where $\tau$ is a $\mathcal{F}_{t}$-adapted stopping time, and thereby
$\mathcal{G}_{t}$-adapted, by extension.

Consider the limit experiment where one observes a $\mathcal{G}_{t}$-adapted
stopping time $\tau$ along with diffusion processes $G_{a}(t):=\pi_{a}\bm{h}_{a}t+\sqrt{\pi_{a}}W_{a}(t),\ a\in\{0,1\}$,
where $W_{1}(\cdot),W_{0}(\cdot)$ are independent 2-dimensional Brownian
motions. By Lemma \ref{Lem: Two-sample ART} in Appendix B, any test
in the parametric model $\left\{ P_{1/\sqrt{n},\bm{h}_{1}^{\intercal}\bm{g}_{1}}\times P_{1/\sqrt{n},\bm{h}_{0}^{\intercal}\bm{g}_{0}}:\bm{h}_{1},\bm{h}_{0}\in\mathbb{R}^{2}\right\} $
can be matched (along sub-sequences) by a test that depends only on
$\bm{G}(\tau),\tau$ in the limit experiment. Hence, 
\[
\beta_{n}(\text{\ensuremath{\bm{h}_{1}^{\intercal}\bm{g}_{1}}},\bm{h}_{0}^{\intercal}\bm{g}_{0}):=\int\varphi_{n}dP_{nT,(\text{\ensuremath{\bm{h}_{1}^{\intercal}\bm{g}_{1}}},\bm{h}_{0}^{\intercal}\bm{g}_{0})}
\]
converges along sub-sequences to the power function, $\beta(\bm{h}_{1},\bm{h}_{0})$,
of some test $\varphi(\tau,\bm{G}(\tau))$ in the limit experiment.
Note that by our definitions, the first component of $\bm{h}_{a}$
is $\left\langle \psi_{a},\bm{h}_{a}^{\intercal}\bm{g}_{a}\right\rangle _{a}/\sigma_{a}$.
This in turn implies, as a consequence of the definition of asymptotically
level-$\alpha$ tests, that $\varphi(\cdot)$ is level-$\alpha$ for
testing $H_{0}:(\sigma_{1},0)^{\intercal}\bm{h}_{1}-(\sigma_{0},0)^{\intercal}\bm{h}_{0}=0$
in the limit experiment. 

Now, by Lemma \ref{Lem: Two-sample optimal test} in Appendix B, the
optimal level-$\alpha$ test of $H_{0}:(\sigma_{1},0)^{\intercal}\bm{h}_{1}-(\sigma_{0},0)^{\intercal}\bm{h}_{0}=0$
vs $H_{1}:(\sigma_{1},0)^{\intercal}\bm{h}_{1}-(\sigma_{0},0)^{\intercal}\bm{h}_{0}=\mu$
in the limit experiment is
\[
\varphi_{\mu}^{*}(\tau,x(\tau)):=\mathbb{I}\left\{ \mu x(\tau)-\frac{\mu^{2}}{2\sigma}\tau\ge\gamma\right\} .
\]
For all $\bm{h}\in H_{1}\equiv\{\bm{h}:(\sigma_{1},0)^{\intercal}\bm{h}_{1}-(\sigma_{0},0)^{\intercal}\bm{h}_{0}=\mu\}$,
\begin{align*}
x(t) & \sim\sigma^{-1}\mu t+\frac{1}{\sigma}\left(\sqrt{\frac{\sigma_{1}^{2}}{\pi_{1}}}(1,0)^{\intercal}W_{1}(t)-\sqrt{\frac{\sigma_{0}^{2}}{\pi_{0}}}(1,0)^{\intercal}W_{0}(t)\right)\\
 & \sim\sigma^{-1}\mu t+\tilde{W}(t),
\end{align*}
where $\tilde{W}(\cdot)$ is standard 1-dimensional Brownian motion.
As $\tau$ is $\mathcal{F}_{t}$-adapted, it follows that the joint
distribution of $(\tau,x(\tau))$ depends only on $\mu$ for $\bm{h}\in H_{1}$.
Consequently, the power, $\mathbb{E}_{\bm{h}}[\varphi_{\mu}^{*}(\tau,x(\tau))]$,
of $\varphi_{\mu}^{*}$ against the values in the alternative hypothesis
$H_{1}$ depends only on $\mu$, and is denoted by $\beta^{*}\left(\mu\right)$.
Since $\varphi_{\mu}^{*}(\cdot)$ is the optimal test and $\mu\in\mathbb{R}$
is arbitrary, $\beta(\bm{h}_{1},\bm{h}_{0})\le\beta^{*}(\mu)$, which
further implies $\lim\sup_{n}\beta_{n}(\text{\ensuremath{\bm{h}_{1}^{\intercal}\bm{g}_{1}}},\bm{h}_{0}^{\intercal}\bm{g}_{0})\le\beta^{*}(\mu)$
for any $\mu\in\mathbb{R}$ and $\bm{h}_{1},\bm{h}_{0}\in\mathbb{R}^{2}$
such that $\left\langle \psi_{1},\bm{h}_{1}^{\intercal}\bm{g}_{1}\right\rangle _{1}-\left\langle \psi_{0},\bm{h}_{0}^{\intercal}\bm{g}_{0}\right\rangle _{0}=\mu$.
Setting $\bm{h}_{a}=\left(\sigma_{a}^{-1}\left\langle \psi_{a},g_{a}\right\rangle _{a},\left\Vert \tilde{g}_{a}\right\Vert _{a}\right)^{\intercal}$
for $a\in\{0,1\}$ then gives $\lim\sup_{n}\int\varphi_{n}dP_{nT,(g_{1},g_{0})}\le\beta^{*}(\mu).$
Since $(g_{1},g_{0})\in T(P_{0}^{(1)})\times T(P_{0}^{(0)})$ was
arbitrary, the claim follows. 

\subsection{Proof of Theorem \ref{Thm: ART-Batched}}

As noted previously, the first claim is shown in \citet{hirano2023asymptotic}.
Consequently, we only focus on proving the second claim. Let ${\bf y}_{j,nq}^{(a)}$
denote the first $nq$ observations from treatment $a$ in batch $j$.
Define
\[
\ln\frac{dP_{n,\bm{h}}}{dP_{n,0}}({\bf y}_{j,nq}^{(a)})=\sum_{i=1}^{\left\lfloor nq\right\rfloor }\ln\frac{dp_{\theta_{0}^{(a)}+h_{a}/\sqrt{n}}}{dp_{\theta_{0}}}(Y_{i,j}^{(a)}).
\]
By the SLAN property, which is a consequence of Assumption 3, 
\begin{equation}
\ln\frac{dP_{n,\bm{h}}}{dP_{n,0}}({\bf y}_{j,n\hat{\pi}_{j}^{(a)}}^{(a)})=h_{a}^{\intercal}I_{a}^{1/2}z_{j,n}^{(a)}(\hat{\pi}_{j}^{(a)})-\frac{\hat{\pi}_{j}^{(a)}}{2}h_{a}^{\intercal}I_{a}h_{a}+o_{P_{n,0}}(1).\label{eq:pf:Thm2:0}
\end{equation}
The above is true for all $j,a$. 

Denote the observed set of outcomes by $\bar{{\bf y}}=\left({\bf y}_{1,n\hat{\pi}_{1}^{(1)}}^{(1)},{\bf y}_{1,n\hat{\pi}_{1}^{(0)}}^{(0)},\dots,{\bf y}_{J,n\hat{\pi}_{J}^{(1)}}^{(1)},{\bf y}_{J,n\hat{\pi}_{J}^{(0)}}^{(0)}\right)$.
The likelihood ratio of the observations satisfies 
\begin{align}
\ln\frac{dP_{n,\bm{h}}}{dP_{n,0}}(\bar{{\bf y}}) & =\sum_{j}\sum_{a\in\{0,1\}}\ln\frac{dP_{n,\bm{h}}}{dP_{n,0}}({\bf y}_{j,nq}^{(a)})\nonumber \\
 & =\sum_{j}\sum_{a\in\{0,1\}}\left\{ h_{a}^{\intercal}I_{a}^{1/2}z_{j,n}^{(a)}(\hat{\pi}_{j}^{(a)})-\frac{\hat{\pi}_{j}^{(a)}}{2}h_{a}^{\intercal}I_{a}h_{a}\right\} ,\label{eq:pf:Thm2:01}
\end{align}
where the second equality follows from (\ref{eq:pf:Thm2:0}). Combining
the above with the first part of the theorem, we find 
\begin{equation}
\ln\frac{dP_{n,\bm{h}}}{dP_{n,0}}(\bar{{\bf y}})\xrightarrow[P_{n,0}]{d}\sum_{j}\sum_{a\in\{0,1\}}\left\{ h_{a}^{\intercal}I_{a}^{1/2}z_{j}^{(a)}(\pi_{j}^{(a)})-\frac{\pi_{j}^{(a)}}{2}h_{a}^{\intercal}I_{a}h_{a}\right\} ,\label{eq:pf:Thm2:1}
\end{equation}
where $z_{j}^{(a)}(t)$ is distributed as $d$-dimensional Brownian
motion. 

Note that $\varphi_{n}$ is required to be measurable with respect
to $\bar{{\bf y}}_{n}$. Furthermore, $\varphi_{n}$ is tight since
$\varphi_{n}\in[0,1]$. Together with (\ref{eq:pf:Thm2:1}), this
implies the joint $\left(\varphi_{n},\ln\frac{dP_{n,\bm{h}}}{dP_{n,0}}(\bar{{\bf y}})\right)$
is also tight. Hence, by Prohorov's theorem, given any sequence $\{n_{j}\}$,
there exists a further sub-sequence $\{n_{j_{m}}\}$ - represented
as $\{n\}$ without loss of generality - such that 
\begin{equation}
\left(\begin{array}{c}
\varphi_{n}\\
\ln\frac{dP_{n,\bm{h}}}{dP_{n,0}}(\bar{{\bf y}})
\end{array}\right)\xrightarrow[P_{n,0}]{d}\left(\begin{array}{c}
\bar{\varphi}\\
V
\end{array}\right);\quad V\sim\prod_{j=1,\dots,J}\prod_{a\in\{0,1\}}\exp\left\{ h_{a}^{\intercal}I_{a}^{1/2}z_{j}^{(a)}(\pi_{j}^{(a)})-\frac{\pi_{j}^{(a)}}{2}h_{a}^{\intercal}I_{a}h_{a}\right\} ,\label{eq:pf:Thm1:weak convergence-1}
\end{equation}
where $\bar{\varphi}\in[0,1]$. Define
\[
V_{j}^{(a)}:=\exp\left\{ h_{a}^{\intercal}I_{a}^{1/2}z_{j}^{(a)}(\pi_{j}^{(a)})-\frac{\pi_{j}^{(a)}}{2}h_{a}^{\intercal}I_{a}h_{a}\right\} ,
\]
so that $V=\prod_{j=1,\dots,J}\prod_{a\in\{0,1\}}V_{j}^{(a)}$. By
the definition of $z_{j}^{(a)}(\cdot)$ and $\pi_{j}^{(a)}$ in the
limit experiment, we have that the Brownian motion $z_{j}^{(a)}(\cdot)$
is independent of data from the all past batches, and consequently,
also independent of $\pi_{j}^{(a)}$. Hence, by the martingale property
of $M_{j}^{(a)}(t):=\exp\left\{ h_{a}^{\intercal}I_{a}^{1/2}z_{j}^{(a)}(t)-\frac{t}{2}h_{a}^{\intercal}I_{a}h_{a}\right\} $,
\[
E[V_{j}^{(a)}\vert z_{1}^{(1)},z_{1}^{(0)},\pi_{1}^{(1)},\pi_{1}^{(0)}\dots,z_{j-1}^{(1)},z_{j-1}^{(0)},\pi_{j-1}^{(1)},\pi_{j-1}^{(0)}]=1
\]
for all $j$ and $a\in\{0,1\}$. This implies, by an iterative argument,
that $E[V]=1$. Consequently, we can employ similar arguments as in
the proof of Theorem 1 to show that
\begin{align}
\lim_{n\to\infty}\beta_{n}(\bm{h}) & :=\lim_{n\to\infty}\mathbb{E}_{n,\bm{h}}\left[\varphi_{n}\right]\nonumber \\
 & =E\left[\bar{\varphi}\prod_{j=1,\dots,J}\prod_{a\in\{0,1\}}e^{h_{a}^{\intercal}I_{a}^{1/2}z_{j}^{(a)}(\pi_{j}^{(a)})-\frac{\pi_{j}^{(a)}}{2}h_{a}^{\intercal}I_{a}h_{a}}\right]\nonumber \\
 & =E\left[\bar{\varphi}\prod_{a\in\{0,1\}}e^{h_{a}^{\intercal}I_{a}^{1/2}x_{a}-\frac{q_{a}}{2}h_{a}^{\intercal}I_{a}h_{a}}\right],\label{eq:pf:Thm2:2}
\end{align}
where the last equality follows from the definition of $x_{a},q_{a}$.
Define 
\[
\varphi\left(q_{1},q_{0},x_{1},x_{0}\right):=E[\bar{\varphi}\vert q_{1},q_{0},x_{1},x_{0}].
\]
Then, the right hand side of (\ref{eq:pf:Thm2:2}) becomes 
\[
E\left[\varphi\left(q_{1},q_{0},x_{1},x_{0}\right)\prod_{a\in\{0,1\}}e^{h_{a}^{\intercal}I_{a}^{1/2}x_{a}-\frac{q_{a}}{2}h_{a}^{\intercal}I_{a}h_{a}}\right].
\]
But by a repeated application of the Girsanov theorem, this is just
the expectation, $\mathbb{E}_{\bm{h}}[\varphi]$, of $\varphi$ when
each $z_{j}^{(a)}(t)$ is distributed as a Gaussian process with drift
$I_{a}^{1/2}h_{a}$, i.e., when $z_{j}^{(a)}(t)\sim I_{a}^{1/2}h_{a}t+W_{j}^{(a)}(t)$,
and $\{W_{j}^{(a)}(\cdot)\}_{j,a}$ are independent Brownian motions. 

\subsection{Proof of Proposition \ref{Prop: Non-parametric batched experiments}}

Denote the observed set of outcomes by $\bar{{\bf y}}=\left({\bf y}_{1,n\hat{\pi}_{1}^{(1)}}^{(1)},{\bf y}_{1,n\hat{\pi}_{1}^{(0)}}^{(0)},\dots,{\bf y}_{J,n\hat{\pi}_{J}^{(1)}}^{(1)},{\bf y}_{J,n\hat{\pi}_{J}^{(0)}}^{(0)}\right)$.
For some arbitrary $\bm{g}=(g_{1},g_{0})\in T(P_{0}^{(1)})\times T(P_{0}^{(0)})$.
As in the proof of Proposition \ref{Prop: Two-sample non-parametric},
we can write $g_{a}=\sigma_{a}^{-1}\left\langle \psi_{a},g_{a}\right\rangle _{a}(\psi_{a}/\sigma_{a})+\tilde{g}_{a}$,
where $\tilde{g}_{a}\perp(\psi_{a}/\sigma_{a})$. Define $\bm{g}_{a}:=\left(\psi_{a}/\sigma_{a},\tilde{g}_{a}/\left\Vert \tilde{g}_{a}\right\Vert _{a}\right)^{\intercal}$,
and consider sub-models of the form $P_{1/\sqrt{n},\bm{h}_{1}^{\intercal}\bm{g}_{1}}\times P_{1/\sqrt{n},\bm{h}_{0}^{\intercal}\bm{g}_{0}}$
for $\bm{h}_{1},\bm{h}_{0}\in\mathbb{R}^{2}$. Following similar rationales
as in the proofs of Propositions \ref{Prop: Non-parametric} and \ref{Prop: Two-sample non-parametric},
we set $\delta_{1}=\delta_{0}=1$ without loss of generality. 

Let $P_{n,\bm{h}}$ and $P_{n,0}$ be defined as in Section \ref{subsec:Asymptotic-representation-theorem-batched},
and set 
\[
\bm{Z}_{j,n}^{(a)}(t):=\frac{1}{\sqrt{n}}\sum_{i=1}^{\left\lfloor nt\right\rfloor }\bm{g}_{a}(Y_{i,j}^{(a)}),\ \textrm{and }\ z_{j,n}^{(a)}(t):=\frac{1}{\sigma_{a}\sqrt{n}}\sum_{i=1}^{\left\lfloor nt\right\rfloor }\psi_{a}(Y_{i,j}^{(a)}).
\]
By similar arguments as that leading to (\ref{eq:pf:Thm2:01}), the
likelihood ratio, 
\[
\ln\frac{dP_{n,(\text{\ensuremath{\bm{h}_{1}^{\intercal}\bm{g}_{1}}},\bm{h}_{0}^{\intercal}\bm{g}_{0})}}{dP_{n,0}}(\bar{{\bf y}}),
\]
of all observations, $\bar{\bm{y}}$, under the sub-model $P_{1/\sqrt{n},\bm{h}_{1}^{\intercal}\bm{g}_{1}}\times P_{1/\sqrt{n},\bm{h}_{0}^{\intercal}\bm{g}_{0}}$
satisfies
\begin{align}
 & \ln\frac{dP_{n,(\text{\ensuremath{\bm{h}_{1}^{\intercal}\bm{g}_{1}}},\bm{h}_{0}^{\intercal}\bm{g}_{0})}}{dP_{n,0}}(\bar{{\bf y}})=\sum_{a}\sum_{j}\left\{ \frac{\bm{h}_{a}^{\intercal}}{\sqrt{n}}\bm{Z}_{j,n}^{(a)}(\hat{\pi}_{j}^{(a)})-\frac{\hat{\pi}_{j}^{(a)}}{2}\bm{h}_{a}^{\intercal}\bm{h}_{a}\right\} +o_{P_{nT,0}}(1).\label{eq:pf:non-parametric:batched-1}
\end{align}
Now, by iterative use of the functional central limit theorem and
the extended continuous mapping theorem (using Assumption 6), 
\begin{align}
\left(\begin{array}{c}
\hat{\pi}_{j}^{(a)}\\
\bm{Z}_{j,n}^{(a)}(\hat{\pi}_{j}^{(a)})
\end{array}\right) & \xrightarrow[P_{nT,\bm{0}}]{d}\left(\begin{array}{c}
\pi_{j}^{(a)}\\
\bm{Z}_{j}^{(a)}(\pi_{j}^{(a)})
\end{array}\right),\quad\bm{Z}_{j}^{(a)}(\cdot)\sim W_{a,j}(\cdot),\label{pf:non-parametric:batched:2}
\end{align}
where $\{W_{a,j}\}_{a,j}$ are independent $2$-dimensional Brownian
motions, and $\pi_{j}^{(a)}$ is measurable with respect to $\sigma\left\{ z_{l}^{(a)}(\cdot);l\le j-1\right\} $
since $\hat{\pi}_{j}^{(a)}$ is measurable with respect to $\sigma\left\{ z_{l,n}^{(a)}(\cdot);l\le j-1\right\} $.

Consider the limit experiment where one observes $q_{a}=\sum_{j}\pi_{j}^{(a)}$
and $x_{a}:=\sum_{j}z_{j}^{(a)}(\pi_{j}^{(a)})$, where
\begin{equation}
z_{j}^{(a)}(t):=\mu_{a}t+W_{j}^{(a)}(t),\label{eq:pf:non-parametric:batched:3}
\end{equation}
and $\pi_{j}$ is measurable with respect to $\sigma\left\{ z_{l}^{(a)}(\cdot);l\le j-1\right\} $.
Using (\ref{eq:pf:non-parametric:batched-1}), (\ref{pf:non-parametric:batched:2})
and employing similar arguments as in Theorem \ref{Thm: ART-Batched},
we find that any test in the parametric model $\left\{ P_{1/\sqrt{n},\bm{h}_{1}^{\intercal}\bm{g}_{1}}\times P_{1/\sqrt{n},\bm{h}_{0}^{\intercal}\bm{g}_{0}}:\bm{h}_{1},\bm{h}_{0}\in\mathbb{R}^{2}\right\} $
can be matched (along sub-sequences) by a test that depends only on
$\bm{G}_{1},\bm{G}_{0},q_{1},q_{0}$ in the limit experiment. Hence,
\[
\beta_{n}(\text{\ensuremath{\bm{h}_{1}^{\intercal}\bm{g}_{1}}},\bm{h}_{0}^{\intercal}\bm{g}_{0}):=\int\varphi_{n}dP_{nT,(\text{\ensuremath{\bm{h}_{1}^{\intercal}\bm{g}_{1}}},\bm{h}_{0}^{\intercal}\bm{g}_{0})}
\]
converges along sub-sequences to the power function, $\beta(\bm{h}_{1},\bm{h}_{0})$,
of some test $\varphi(q_{1},q_{0},\bm{G}_{1},\bm{G}_{0})$ in the
limit experiment. Note that by our definitions, the first component
of $\bm{h}_{a}$ is $\left\langle \psi_{a},\bm{h}_{a}^{\intercal}\bm{g}_{a}\right\rangle _{a}/\sigma_{a}$.
This in turn implies, as a consequence of the definition of asymptotically
level-$\alpha$ tests, that $\varphi(\cdot)$ is level-$\alpha$ for
testing 
\[
H_{0}:\left((\sigma_{1},0)^{\intercal}\bm{h}_{1},(\sigma_{0},0)^{\intercal}\bm{h}_{0}\right)=(0,0)
\]
in the limit experiment. 

Now, by Lemma \ref{Lem: Batched experiments optimal test} in Appendix
B, the optimal level-$\alpha$ test of the null $H_{0}$ vs $H_{1}:\left((\sigma_{1},0)^{\intercal}\bm{h}_{1},(\sigma_{0},0)^{\intercal}\bm{h}_{0}\right)=(\mu_{1},\mu_{0})$
in the limit experiment is
\[
\varphi_{\mu_{1},\mu_{0}}^{*}=\mathbb{I}\left\{ \sum_{a\in\{0,1\}}\left(\frac{\mu_{a}}{\sigma_{a}}x_{a}-\frac{q_{a}}{2\sigma_{a}^{2}}\mu_{a}^{2}\right)\ge\gamma_{\mu_{1},\mu_{0}}\right\} .
\]
Using (\ref{eq:pf:non-parametric:batched:3}) and the fact $\pi_{j}$
depends only on the past values of $z_{j}^{(a)}(\cdot)$, it follows
that the joint distribution of $(q_{1},q_{0},x_{1},x_{0})$ depends
only on $\mu_{1},\mu_{0}$ for $\bm{h}\in H_{1}$. Consequently, the
power, $\mathbb{E}_{\bm{h}}\left[\varphi_{\mu_{1},\mu_{0}}^{*}\right]$,
of $\varphi_{\mu_{1},\mu_{0}}^{*}$ against the values in the alternative
hypothesis $H_{1}$ depends only on $(\mu_{1},\mu_{0})$, and is denoted
by $\beta^{*}\left(\mu_{1},\mu_{0}\right)$. Since $\varphi_{\mu_{1},\mu_{0}}^{*}$
is the optimal test and $(\mu_{1},\mu_{0})\in\mathbb{R}^{2}$ is arbitrary,
$\beta(\bm{h}_{1},\bm{h}_{0})\le\beta^{*}(\mu_{1},\mu_{0})$. This
further implies $\lim\sup_{n}\beta_{n}(\text{\ensuremath{\bm{h}_{1}^{\intercal}\bm{g}_{1}}},\bm{h}_{0}^{\intercal}\bm{g}_{0})\le\beta^{*}(\mu_{1},\mu_{0})$
for any $(\mu_{1},\mu_{0})\in\mathbb{R}$ and $\bm{h}_{1},\bm{h}_{0}\in\mathbb{R}^{2}$
such that $\left\langle \psi_{a},\bm{h}_{a}^{\intercal}\bm{g}_{a}\right\rangle _{a}=\mu_{a}$.
Setting $\bm{h}_{a}=\left(\sigma_{a}^{-1}\left\langle \psi_{a},g_{a}\right\rangle _{a},\left\Vert \tilde{g}_{a}\right\Vert _{a}\right)^{\intercal}$
for $a\in\{0,1\}$ then gives $\lim\sup_{n}\int\varphi_{n}dP_{nT,(g_{1},g_{0})}\le\beta^{*}(\mu_{1},\mu_{0}).$
Since $(g_{1},g_{0})\in T(P_{0}^{(1)})\times T(P_{0}^{(0)})$ was
arbitrary, the claim follows. 

\newpage{}

\section{Additional results}

\subsection{Variance estimators\label{subsec:Variance-estimators}}

The score/efficient influence function process $x_{n}(\cdot)$ depends
on the information matrix $I$ (in the case of parametric models)
or on the variance $\sigma$ (in the case of non-parametric models).
For parametric models, if the reference parameter, $\theta_{0}$,
is known, we could simply set $I=I(\theta_{0})$. In most applications,
however, this would be unknown, and we would need to replace $I$
and $\sigma$ with consistent estimators. Here, we discuss various
proposals for variance estimation (note that $I$ can be thought of
as variance since $E_{0}[\psi\psi^{\intercal}]=I$). 

\subsubsection*{Batched experiments.}

If the experiment is conducted in batches, we can simply use the data
from the first batch to construct consistent estimators of the variances.
This of course has the drawback of not using all the data, but it
is unbiased and $\sqrt{n}$-consistent under very weak assumptions
(i.e., existence of second moments).

\subsubsection*{Running-estimator of variance.}

For an estimator that is more generally valid and uses all the data,
we recommend the running-variance estimate 
\begin{equation}
\hat{\Sigma}_{a,t}=\frac{1}{nt}\sum_{i=1}^{\left\lfloor nt\right\rfloor }\psi_{a}(Y_{i}^{(a)})\psi_{a}(Y_{i}^{(a)})^{\intercal}-\left(\frac{1}{nt}\sum_{i=1}^{\left\lfloor nt\right\rfloor }\psi_{a}(Y_{i}^{(a)})\right)\left(\frac{1}{nt}\sum_{i=1}^{\left\lfloor nt\right\rfloor }\psi_{a}(Y_{i}^{(a)})\right)^{\intercal},\label{eq:running variance estimate}
\end{equation}
for each treatment $a$. The final estimate of the variance would
then be $\hat{\Sigma}_{a,\hat{\tau}}$ for stopping-times experiments,
and $\hat{\Sigma}_{a,q_{a}}$ for batched experiments. Let $\Sigma_{a}:=E_{0,a}[\psi_{a}\psi_{a}^{\intercal}]$
and suppose that $\psi_{a}\psi_{a}^{\intercal}$ is $\lambda$-sub-Gaussian
for some $\lambda>0$. Then using standard concentration inequalities,
see e.g., \citet[Corollary 5.5]{lattimore2020bandit}, we can show
that 
\[
P_{nT,0}\left(\bigcup_{t=1}^{T}\left\{ \left|\hat{\Sigma}_{a,t}-\Sigma_{a}\right|\ge C\sqrt{\frac{\ln(1/\delta)}{nt}}\right\} \right)\le nT\delta\quad\forall\quad\delta\in[0,1],
\]
where $C$ is independent of $n,t,\delta$ (but does depend on $\lambda$).
Setting $\delta=n^{-a}$ for some $a>0$ then implies that $\hat{\Sigma}_{a,\hat{\tau}}$
and $\hat{\Sigma}_{a,q_{a}}$ are $\sqrt{n}$-consistent for $\Sigma_{a}$
(upto log factors) as long as $\hat{\tau},q_{a}>0$ almost-surely
under $P_{nT,0}$.

\subsubsection*{Bayes estimators.}

Yet a third alternative is to place a prior on $\Sigma_{a}$ and continuously
update its value using posterior means. As a default, we suggest employing
an inverse-Wishart prior and computing the posterior by treating the
outcomes as Gaussian (this is of course justified in the limit). Since
posterior consistency holds under mild assumptions, we expect this
estimator to perform similarly to (\ref{eq:running variance estimate}). 

\subsection{Supporting information for Section \ref{subsec:Sequential-linear-boundary}\label{subsec:Supporting-information-for-sequential-linear-boundary}}

In this section, we provide a proof of Lemma \ref{Lem: costly sampling}.
The proof proceeds in two steps: First, we characterize the best unbiased
test in the limit experiment described in Section \ref{subsec:Sequential-linear-boundary}.
Then, we show that the finite sample counterpart of this test attains
the power envelope for asymptotically unbiased tests.

\subsubsection*{Step 1:}

Consider the problem of testing $H_{0}:\mu=0$ vs $H_{1}:\mu\neq0$
in the limit experiment. Let $\mathbb{P}_{\mu}(\cdot)$ denote the
induced probability measure over the sample paths of $x(\cdot)$ in
the limit experiment, and $\mathbb{E}_{\mu}[\cdot]$ its corresponding
expectation. Due to the nature of the stopping time, $x(\tau)$ can
only take on two values $\gamma,-\gamma$. Let $\delta$ denote the
sign of $x(\tau)$. Then, by sufficiency, any test $\varphi$, in
the limit experiment can be written as a function only of $\tau,\delta$.
Furthermore, by Proposition \ref{Prop: Unbiased}, any unbiased test,
$\varphi(\tau,\delta)$, must satisfy $\mathbb{E}_{0}[\delta\varphi(\tau,\delta)]=0$. 

Fix some alternative $\mu\neq0$ and consider the functional optimization
problem 
\begin{align}
 & \max_{\varphi(\cdot)}\mathbb{E}_{\mu}[\varphi(\tau,\delta)]\equiv\mathbb{E}_{0}\left[\varphi(\tau,\delta)e^{\frac{1}{\sigma}\mu\delta\gamma-\frac{\tau}{2\sigma^{2}}\mu^{2}}\right]\label{eq:costly sampling: best test}\\
 & \textrm{s.t}\ \mathbb{E}_{0}[\varphi(\tau,\delta)]\le\alpha\ \textrm{and }\ \mathbb{E}_{0}[\delta\varphi(\tau,\delta)]=0.\nonumber 
\end{align}
Here, and in what follows, it should implicitly understood that the
candidate functions, $\varphi(\cdot)$, are tests, i.e., their range
is $[0,1]$. Let $\varphi^{*}$ denote the optimal solution to (\ref{eq:costly sampling: best test}).
Note that $\varphi^{*}$ is unbiased since $\varphi=\alpha$ also
satisfies the constraints in (\ref{eq:costly sampling: best test});
indeed, $\mathbb{E}_{0}[\delta]=0$ by symmetry. Consequently, if
$\varphi^{*}$ is shown to be independent of $\mu$, we can conclude
that it is the best unbiased test.

Now, by \citet{fudenberg2018speed}, $\delta$ is independent of $\tau$
given $\mu$. Furthermore, by symmetry, $\mathbb{P}_{0}(\delta=1)=\mathbb{P}_{0}(\delta=-1)=1/2$
for $\mu=0$. Based on these results, we have
\begin{align*}
(0=)\mathbb{E}_{0}[\delta\varphi(\tau,\delta)] & =\frac{1}{2}\int\left\{ \varphi(\tau,1)-\varphi(\tau,0)\right\} dF_{0}(\tau),\\
\mathbb{E}_{0}[\varphi(\tau,\delta)] & =\frac{1}{2}\int\left\{ \varphi(\tau,1)+\varphi(\tau,0)\right\} dF_{0}(\tau),\ \textrm{and}\\
\mathbb{E}_{0}\left[\varphi(\tau,\delta)e^{\frac{1}{\sigma}\mu\delta\gamma-\frac{\tau}{2\sigma^{2}}\mu^{2}}\right] & =\frac{e^{\mu\gamma/\sigma}}{2}\int\varphi(\tau,1)e^{-\frac{\tau}{2\sigma^{2}}\mu^{2}}dF_{0}(\tau)+\frac{e^{-\mu\gamma/\sigma}}{2}\int\varphi(\tau,0)e^{-\frac{\tau}{2\sigma^{2}}\mu^{2}}dF_{0}(\tau).
\end{align*}
The first two equations above imply $\mathbb{E}_{0}[\varphi(\tau,1)]=\mathbb{E}_{0}[\varphi(\tau,0)]=\mathbb{E}_{0}[\varphi(\tau,\delta)]$.
Hence, we can rewrite the optimization problem (\ref{eq:costly sampling: best test})
as 
\begin{align}
 & \max_{\varphi(\cdot)}\left\{ \frac{e^{\mu\gamma/\sigma}}{2}\int\varphi(\tau,1)e^{-\frac{\tau}{2\sigma^{2}}\mu^{2}}dF_{0}(\tau)+\frac{e^{-\mu\gamma/\sigma}}{2}\int\varphi(\tau,0)e^{-\frac{\tau}{2\sigma^{2}}\mu^{2}}dF_{0}(\tau)\right\} \label{eq:unbiased test - optimization problem 2}\\
 & \textrm{s.t.}\int\varphi(\tau,1)dF_{0}(\tau)\le\alpha,\ \int\varphi(\tau,0)dF_{0}(\tau)\le\alpha\ \textrm{and}\nonumber \\
 & \quad\int\varphi(\tau,1)dF_{0}(\tau)=\int\varphi(\tau,0)dF_{0}(\tau).\nonumber 
\end{align}
Let us momentarily disregard the last constraint in (\ref{eq:unbiased test - optimization problem 2}).
Then the optimization problem factorizes, and the optimal $\varphi(\cdot)$
can be determined by separately solving for $\varphi(\cdot,1),\varphi(\cdot,0)$
as the functions that optimize
\[
\max_{\varphi(\cdot,a)}\int\varphi(\tau,a)e^{-\frac{\tau}{2\sigma^{2}}\mu^{2}}dF_{0}(\tau)\quad\textrm{s.t. }\ \int\varphi(\tau,a)dF_{0}(\tau)\le\alpha
\]
for $a\in\{0,1\}$. Let $\varphi^{*}(\cdot,a)$ denote the optimal
solution. It is immediate from the optimization problem above that
$\varphi^{*}(\tau,1)=\varphi^{*}(\tau,0):=\varphi^{*}(\tau)$, i.e.,
the optimal $\varphi^{*}$ is independent of $\delta$. Hence, the
last constraint in (\ref{eq:unbiased test - optimization problem 2})
is satisfied. Furthermore, by the Neyman-Pearson lemma, 
\[
\varphi^{*}(\tau)=\mathbb{I}\left\{ e^{-\frac{\tau}{2\sigma^{2}}\mu^{2}}\ge\gamma\right\} \equiv\mathbb{I}\left\{ \tau\le c\right\} ,
\]
where $c=F_{0}^{-1}(\alpha)$ due to the requirement that $\int\varphi(\tau,a)dF_{0}(\tau)\le\alpha$.
Consequently, the solution, $\varphi^{*}(\cdot)$, to (\ref{eq:costly sampling: best test})
is given by $\mathbb{I}\left\{ \tau\le F_{0}^{-1}(\alpha)\right\} $.
This is obviously independent of $\mu$. We conclude that it is the
best unbiased test in the limit experiment. 

\subsubsection*{Step 2:}

The finite sample counterpart of $\varphi^{*}(\cdot)$ is given by
$\hat{\varphi}(\hat{\tau}):=\mathbb{I}\left\{ \hat{\tau}\le F_{0}^{-1}(\alpha)\right\} $,
where it may be recalled that $\hat{\tau}=\inf\{t:\vert x_{n}(t)\vert\ge\gamma\}$.
Fix some arbitrary $\bm{g}:=(g_{1},g_{0})\in T(P_{0}^{(1)})\times T(P_{0}^{(0)})$.
Let $P_{nT,\bm{g}}$ be defined as in the proof of Proposition \ref{Prop: Two-sample non-parametric}.
By similar arguments as in the proofs of \citet[Theorems 3 and 5]{adusumilli2022sample},
\[
\hat{\tau}\xrightarrow[P_{nT,\bm{g}}]{d}\tau:=\inf\{t:\vert x(t)\vert\ge\gamma\}
\]
along sub-sequences, where $x(t)\sim\sigma^{-1}\mu t+\tilde{W}(t)$
and $\mu:=\left\langle \psi_{1},g_{1}\right\rangle _{1}-\left\langle \psi_{0},g_{0}\right\rangle _{0}$.
Hence, 
\[
\lim_{n\to\infty}\hat{\beta}(g_{1},g_{0}):=\lim_{n\to\infty}P_{nT,(g_{1},g_{0})}\left(\hat{\tau}\le F_{0}^{-1}(\alpha)\right)=\mathbb{P}_{\mu}\left(\tau\le F_{0}^{-1}(\alpha)\right),
\]
where $\mathbb{P}_{\mu}(\cdot)$ is the probability measure defined
in Step 1. But $\tilde{\beta}^{*}(\mu):=P_{\mu}\left(\tau\le F_{0}^{-1}(\alpha)\right)$
is just the power function of the best unbiased test, $\varphi^{*}$,
in limit experiment. Hence, $\hat{\varphi}(\cdot)$ is an asymptotically
optimal unbiased test. 

\subsection{Supporting information for Section \ref{subsec:Group-Sequential-Experiments}\label{subsec:Supporting-information-for-GST}}

\subsubsection{Nonparametric level-$\bm{\alpha}$ and conditionally unbiased tests}

Here, we define non-parametric versions of the level-$\bm{\alpha}$
and conditionally unbiased requirements. We follow the same notation
as in Section \ref{sec:Two-sample-tests}. A test, $\varphi_{n}$,
of $H_{0}:\mu_{1}-\mu_{0}=\mu/\sqrt{n}$ is said to asymptotically
level-$\bm{\alpha}$ if 
\begin{equation}
\sup_{\left\{ \bm{h}:\left\langle \psi_{1},h_{1}\right\rangle _{1}-\left\langle \psi_{0},h_{0}\right\rangle _{0}=\mu\right\} }\limsup_{n}\int\mathbb{I}\{\hat{\tau}=k\}\varphi_{n}dP_{nT,\bm{h}}\le\alpha_{k}\ \forall\ k.\label{eq:Level -alpha - two sample definition-asymptotic}
\end{equation}
Similarly, a test, $\varphi_{n}$, of $H_{0}:\mu_{1}-\mu_{0}=\mu/\sqrt{n}$
vs $H_{1}:\mu_{1}-\mu_{0}\neq\mu/\sqrt{n}$ is asymptotically conditionally
unbiased if 
\begin{align*}
 & \sup_{\left\{ \bm{h}:\left\langle \psi_{1},h_{1}\right\rangle _{1}-\left\langle \psi_{0},h_{0}\right\rangle _{0}=\mu\right\} }\limsup_{n}\int\mathbb{I}\{\tau=k\}\varphi_{n}dP_{nT,\bm{h}}\\
 & \ge\inf_{\left\{ \bm{h}:\left\langle \psi_{1},h_{1}\right\rangle _{1}-\left\langle \psi_{0},h_{0}\right\rangle _{0}\neq\mu\right\} }\liminf_{n}\int\varphi_{n}dP_{nT,\bm{h}}.
\end{align*}

\subsubsection{Attaining the bound}

Recall the definition of $x_{n}(\cdot)$ in (\ref{eq:two-sample tests sufficient statistic}).
While $x_{n}(\cdot)$ depends on the unknown quantities $\sigma_{1},\sigma_{0}$,
we can replace them with consistent estimates $\hat{\sigma}_{1},\hat{\sigma}_{0}$
using data from the first batch without affecting the asymptotic results,
so there is no loss of generality in taking them to be known. Let
$\hat{\varphi}:=\varphi^{*}(\hat{\tau},x_{n}(\hat{\tau}))$ denote
the finite sample counterpart of $\varphi^{*}$. 

By an extension of Proposition \ref{Prop: Two-sample non-parametric}
to $\alpha$-spending tests, as in Theorem \ref{Thm: ART-conditional inference},
the conditional power function, $\beta^{*}(\mu|k)$, of $\varphi^{*}$
in the limit experiment is an upper bound on the asymptotic power
function of any test in the original experiment. We now show that
the local (conditional) power, $\hat{\beta}(g_{1},g_{0}\vert k)$,
of $\hat{\varphi}$ against sub-models $P_{1/\sqrt{n},g_{1}}\times P_{1/\sqrt{n},g_{0}}$
converges to $\beta^{*}(\mu\vert k)$. This implies that $\hat{\varphi}$
is an asymptotically optimal level-$\bm{\alpha}$ test in this experiment. 

Fix some arbitrary $\bm{g}:=(g_{1},g_{0})\in T(P_{0}^{(1)})\times T(P_{0}^{(0)})$.
Let $P_{nT,\bm{g}}$ be defined as in the proof of Proposition \ref{Prop: Two-sample non-parametric}.
By similar arguments as in the proofs of \citet[Theorems 3 and 5]{adusumilli2022sample},
\[
x_{n}(\cdot)\xrightarrow[P_{nT,\bm{g}}]{d}x(\cdot)
\]
along sub-sequences, where $x(t)\sim\sigma^{-1}\mu t+\tilde{W}(t)$
and $\mu:=\left\langle \psi_{1},g_{1}\right\rangle _{1}-\left\langle \psi_{0},g_{0}\right\rangle _{0}$.
Since $\hat{\tau}$ is a function of $x_{n}(\cdot)$, the above implies,
by an application of the extended continuous mapping theorem \citep[Theorem 1.11.1]{van1996weak},
that
\begin{align*}
\lim_{n\to\infty}\int\mathbb{I}\{\hat{\tau}=k\}\hat{\varphi}P_{nT,(g_{1},g_{0})} & =\int\mathbb{I}\{\tau=k\}\varphi^{*}d\mathbb{P}_{\mu},\ \textrm{and}\\
\lim_{n\to\infty}\int\mathbb{I}\{\hat{\tau}=k\}P_{nT,(g_{1},g_{0})} & =\int\mathbb{I}\{\tau=k\}d\mathbb{P}_{\mu}.
\end{align*}
Hence, as long as $\mathbb{P}_{0}(\tau=k)\neq0$, by the definition
of conditional power, we obtain
\[
\lim_{n\to\infty}\hat{\beta}(g_{1},g_{0}\vert k)=\frac{\int\mathbb{I}\{\tau=k\}\varphi^{*}d\mathbb{P}_{\mu}}{\mathbb{I}\{\tau=k\}d\mathbb{P}_{\mu}}:=\beta^{*}(\mu\vert k),
\]
for any $\mu\in\mathbb{R}$. This implies that $\hat{\varphi}$ is
asymptotically level-$\bm{\alpha}$ (as can be verified by setting
$\mu=0$ etc), and furthermore, its conditional power attains the
upper bound $\beta^{*}(\cdot\vert k)$. Hence, $\hat{\varphi}$ is
an asymptotically optimal level-$\bm{\alpha}$ test.

\subsection{Supporting results for the proof of Proposition \ref{Prop: Two-sample non-parametric}}

\begin{lem} \label{Lem: Two-sample ART} Consider the setup in the
proof of Proposition \ref{Prop: Two-sample non-parametric}. Let $P_{1/\sqrt{n},\bm{h}_{a}^{\intercal}\bm{g}_{a}}^{(a)}$
denote the probability sub-model for treatment $a$, and suppose that
it satisfies the SLAN property 
\[
\ln\frac{dP_{nt,\text{\ensuremath{\bm{h}_{a}^{\intercal}\bm{g}_{a}}}}}{dP_{nt,\bm{0}}}({\bf y}_{nt}^{(a)})=\frac{\bm{h}_{a}^{\intercal}}{\sqrt{n}}\sum_{i=1}^{\left\lfloor n\pi_{a}t\right\rfloor }\bm{g}_{a}(Y_{i}^{(a)})-\frac{\pi_{a}t}{2}\bm{h}_{a}^{\intercal}\bm{h}_{a}++o_{P_{nT,\bm{0}}}(1),\ \textrm{ uniformly over }t.
\]
Then, any test in the parametric model $\left\{ P_{1/\sqrt{n},\bm{h}_{1}^{\intercal}\bm{g}_{1}}\times P_{1/\sqrt{n},\bm{h}_{0}^{\intercal}\bm{g}_{0}}:\bm{h}_{1},\bm{h}_{0}\in\mathbb{R}^{2}\right\} $
can be matched (along sub-sequences) by a test that depends only on
$\bm{G}(\tau),\tau$ in the limit experiment. \end{lem} 
\begin{proof}
Recall that $G_{a,n}(t):=n^{-1/2}\sum_{i=1}^{\left\lfloor n\pi_{a}t\right\rfloor }\bm{g}_{a}(Y_{i}^{(a)})$
for $a\in\{0,1\}$. Then, by the statement of the lemma, we have
\begin{equation}
\ln\frac{dP_{n\hat{\tau},\bm{h}_{a}^{\intercal}\bm{g}_{a}}}{dP_{n\hat{\tau},0}}({\bf y}_{n\hat{\tau}}^{(a)})=\bm{h}_{a}^{\intercal}G_{a,n}(\hat{\tau})-\frac{\pi_{a}\hat{\tau}}{2}\bm{h}_{a}^{\intercal}\bm{h}_{a}+o_{P_{nT,\bm{0}}}(1),\label{eq:lem:two-sample:1}
\end{equation}
for $a\in\{0,1\}$. In the proof of Proposition \ref{Prop: Two-sample non-parametric},
we argued that
\begin{equation}
(G_{1,n}(\hat{\tau}),G_{0,n}(\hat{\tau}),\hat{\tau})\xrightarrow[P_{nT,\bm{0}}]{d}(G_{1}(\tau),G_{0}(\tau),\tau),\label{pf:lem:2-sample:2}
\end{equation}
where $G_{a}(t)\sim\sqrt{\pi_{a}}W_{a}(t)$ with $W_{1}(\cdot),W(\cdot)$
being independent $2$-dimensional Brownian motions; and $\tau$ is
a $\mathcal{G}_{t}$-adapted stopping time. Equations (\ref{eq:lem:two-sample:1})
and (\ref{pf:lem:2-sample:2}) imply
\begin{equation}
\ln\frac{dP_{nt,(\text{\ensuremath{\bm{h}_{1}^{\intercal}\bm{g}_{1}}},\bm{h}_{0}^{\intercal}\bm{g}_{0})}}{dP_{nt,\bm{0}}}({\bf y}_{nt})\xrightarrow[P_{nT,\bm{0}}]{d}\sum_{a\in\{0,1\}}\left\{ \bm{h}_{a}^{\intercal}G_{a}(\tau)-\frac{\pi_{a}\tau}{2}\bm{h}_{a}^{\intercal}\bm{h}_{a}\right\} .\label{pf:Thm:ART:eq:1-1}
\end{equation}

Now, any two-sample test, $\varphi_{n}$, is tight since $\varphi_{n}\in[0,1]$.
Then, as in the proof of Theorem \ref{Thm: ART}, we find that given
any sequence $\{n_{j}\}$, there exists a further sub-sequence $\{n_{j_{m}}\}$
- represented as $\{n\}$ without loss of generality - such that 
\begin{equation}
\left(\begin{array}{c}
\varphi_{n}\\
\frac{dP_{nt,(\text{\ensuremath{\bm{h}_{1}^{\intercal}\bm{g}_{1}}},\bm{h}_{0}^{\intercal}\bm{g}_{0})}}{dP_{nt,\bm{0}}}({\bf y}_{nt})
\end{array}\right)\xrightarrow[P_{nT,\bm{0}}]{d}\left(\begin{array}{c}
\bar{\varphi}\\
V
\end{array}\right);\quad V\sim\exp\sum_{a}\left\{ \bm{h}_{a}^{\intercal}G_{a}(\tau)-\frac{\pi_{a}\tau}{2}\bm{h}_{a}^{\intercal}\bm{h}_{a}\right\} ,\label{eq:pf:Thm1:weak convergence-3}
\end{equation}
where $\bar{\varphi}\in[0,1]$. Now, given that $G_{a}(t)\sim\sqrt{\pi_{a}}W_{a}(t)$,
\[
V\sim\exp\sum_{a}\left\{ \sqrt{\pi_{a}}\bm{h}_{a}^{\intercal}W_{a}(\tau)-\frac{\pi_{a}\tau}{2}\bm{h}_{a}^{\intercal}\bm{h}_{a}\right\} .
\]
Clearly, $V_{a}$ is the stochastic/Dol{\'e}ans-Dade exponential
of $\sum_{a}\left\{ \sqrt{\pi_{a}}\bm{h}_{a}^{\intercal}W_{a}(\tau)\right\} $.
Since $W_{1}(\cdot),W_{0}(\cdot)$ are independent, the latter quantity
is in turn distributed as $\left(\sum_{a}\pi_{a}\bm{h}_{a}^{\intercal}\bm{h}_{a}\right)^{1/2}\tilde{W}(t)$,
where $\tilde{W}(\cdot)$ is standard 1-dimensional Brownian motion.
Hence, by standard results on stochastic exponentials, 
\[
M(t):=\exp\sum_{a}\left\{ \bm{h}_{a}^{\intercal}G_{a}(t)-\frac{\pi_{a}t}{2}\bm{h}_{a}^{\intercal}\bm{h}_{a}\right\} 
\]
is a martingale with respect to the filtration $\mathcal{\mathcal{G}}_{t}$.
Since $\tau$ is an $\mathcal{G}_{t}$-adapted stopping time, $E[V]\equiv E[M(\tau)]=E[M(0)]=1$
using the optional stopping theorem.

The above then implies, as in the proof of Theorem \ref{Thm: ART},
that
\begin{equation}
\lim_{n\to\infty}\beta_{n}(\text{\ensuremath{\bm{h}_{1}^{\intercal}\bm{g}_{1}}},\bm{h}_{0}^{\intercal}\bm{g}_{0}):=\lim_{n\to\infty}\int\varphi_{n}dP_{nT,(\text{\ensuremath{\bm{h}_{1}^{\intercal}\bm{g}_{1}}},\bm{h}_{0}^{\intercal}\bm{g}_{0})}=E\left[\bar{\varphi}e^{\sum_{a}\left\{ \bm{h}_{a}^{\intercal}G_{a}(\tau)-\frac{\pi_{a}\tau}{2}\bm{h}_{a}^{\intercal}\bm{h}_{a}\right\} }\right].\label{eq:pf:Thm:ART:2-1}
\end{equation}
Define $\varphi(\tau,\bm{G}(\tau)):=E[\bar{\varphi}\vert\tau,\bm{G}(\tau)]$;
this is a test statistic since $\varphi\in[0,1]$. The right hand
side of (\ref{eq:pf:Thm:ART:2-1}) then becomes 
\[
E\left[\varphi(\tau,\bm{G}(\tau))e^{\sum_{a}\left\{ \bm{h}_{a}^{\intercal}G_{a}(\tau)-\frac{\pi_{a}\tau}{2}\bm{h}_{a}^{\intercal}\bm{h}_{a}\right\} }\right].
\]
But by the Girsanov theorem, this is just the expectation, $\mathbb{E}_{\bm{h}}[\varphi(\tau,\bm{G}(\tau))]$,
of $\varphi(\tau,\bm{G}(\tau))$ when $G_{a}(t)\sim\pi_{a}\bm{h}_{a}t+\sqrt{\pi_{a}}W_{a}(t)$
. This proves the desired claim. 
\end{proof}
\begin{lem} \label{Lem: Two-sample optimal test} Consider the limit
experiment where one observes a stopping time $\tau$ and independent
diffusion processes $G_{1}(\cdot),G_{0}(\cdot)$, where $G_{a}(t):=\pi_{a}\bm{h}_{a}t+\sqrt{\pi_{a}}W_{a}(t)$.
Let $\sigma$, $x(\cdot)$ and $\mathcal{F}_{t}$ be as defined in
the proof of Proposition \ref{Prop: Two-sample non-parametric}, and
suppose that $\tau$ is $\mathcal{F}_{t}$-adapted. Then, the optimal
level-$\alpha$ test of $H_{0}:(\sigma_{1},0)^{\intercal}\bm{h}_{1}-(\sigma_{0},0)^{\intercal}\bm{h}_{0}=0$
vs $H_{1}:(\sigma_{1},0)^{\intercal}\bm{h}_{1}-(\sigma_{0},0)^{\intercal}\bm{h}_{0}=\mu$
in the limit experiment is given by
\[
\varphi_{\mu}^{*}(\tau,x(\tau)):=\mathbb{I}\left\{ \mu x(\tau)-\frac{\mu^{2}}{2\sigma}\tau\ge\gamma\right\} .
\]
 \end{lem} 
\begin{proof}
For each $a$ we employ a change of variables $\bm{h}_{a}\to\bm{\Delta}_{a}$
as $\bm{\Delta}_{a}=\Lambda_{a}\bm{h}_{a}$, where 
\[
\Lambda_{a}:=\left[\begin{array}{cc}
\sigma_{a} & 0\\
0 & 1
\end{array}\right].
\]
Set $\bm{\Delta}:=(\bm{\Delta}_{1},\bm{\Delta}_{0})$. The null and
alternative regions are then $H_{0}\equiv\{\bm{\Delta}:(1,0)^{\intercal}\bm{\Delta}_{1}-(1,0)^{\intercal}\bm{\Delta}_{0}=0\}$
and $H_{1}\equiv\{\bm{\Delta}:(1,0)^{\intercal}\bm{\Delta}_{1}-(1,0)^{\intercal}\bm{\Delta}_{0}=\mu\}$.
Let $\mathbb{P}_{\bm{\Delta}}\equiv\mathbb{P}_{\bm{h}}$ denote the
induced probability measure over the sample paths generated by $G_{1}(\cdot),G_{0}(\cdot)$
between $t\in[0,T]$, when $G_{a}(t)\sim\pi_{a}\Lambda_{a}^{-1}\bm{\Delta}_{a}t+\sqrt{\pi_{a}}W_{a}(t)$.
Also, recall that 
\[
x(t):=\frac{1}{\sigma}\left(\frac{\sigma_{1}}{\pi_{1}}z_{1}(t)-\frac{\sigma_{0}}{\pi_{0}}z_{0}(t)\right),
\]
where $z_{1}(\cdot),z_{2}(\cdot)$ are the first components of $G_{1}(\cdot),G_{0}(\cdot)$. 

Fix some $\bar{\bm{\Delta}}:=(\bm{\bar{\Delta}}_{1},\bm{\bar{\Delta}}_{0})\in H_{1}$.
Let $\bar{\Delta}_{11}$ and $\bar{\Delta}_{01}$ denote the first
components of $\bm{\bar{\Delta}}_{1},\bm{\bar{\Delta}}_{0}$, and
define $\gamma,\eta$ so that 
\begin{equation}
(\bar{\Delta}_{11},\bar{\Delta}_{01})=\left(\gamma+\frac{\sigma_{1}^{2}\eta}{\pi_{1}},\gamma-\frac{\sigma_{0}^{2}\eta}{\pi_{0}}\right).\label{eq:pf:Lem2:1}
\end{equation}
Clearly, $\eta=\mu/\sigma^{2}$ and $\gamma=\bar{\Delta}_{11}-\sigma_{1}^{2}\eta/\pi_{1}$.
Now construct $\tilde{\bm{\Delta}}=(\tilde{\bm{\Delta}}_{1},\tilde{\bm{\Delta}}_{0})$
as follows: The second components of $\tilde{\bm{\Delta}}_{1},\tilde{\bm{\Delta}}_{0}$
are the same as that of $\bm{\bar{\Delta}}_{1},\bm{\bar{\Delta}}_{0}$.
As for the first components, $\tilde{\Delta}_{11},\tilde{\Delta}_{01}$
of $\tilde{\bm{\Delta}}_{1},\tilde{\bm{\Delta}}_{0}$ , take them
to be 
\begin{equation}
(\tilde{\Delta}_{11},\tilde{\Delta}_{01})=\left(\gamma,\gamma\right).\label{eq:pf:lem2-2}
\end{equation}
By construction, $(\tilde{\bm{\Delta}}_{1},\tilde{\bm{\Delta}}_{0})\in H_{0}$. 

Consider testing $H_{0}^{\prime}:\bm{\Delta}=\tilde{\bm{\Delta}}$
vs $H_{1}^{\prime}:\bm{\Delta}=\bar{\bm{\Delta}}$. Let $\ln\frac{d\mathbb{P}_{\bar{\bm{\Delta}}}}{d\mathbb{P}_{\tilde{\bm{\Delta}}}}(\mathcal{G}_{t})$
denote the likelihood ratio between the probabilities induced by the
parameters $\tilde{\bm{h}},\bar{\bm{h}}$ over the filtration $\mathcal{G}_{t}$.
Since $G_{1}(\cdot),G_{0}(\cdot)$ are independent, the Girsanov theorem
gives
\begin{align*}
\ln\frac{d\mathbb{P}_{\bar{\bm{\Delta}}}}{d\mathbb{P}_{\tilde{\bm{\Delta}}}}(\mathcal{G}_{t}) & =\left(\bm{\bar{\Delta}}_{1}^{\intercal}\Lambda_{1}^{-1}G_{1}(\tau)-\frac{\pi_{1}\tau}{2}\bar{\bm{\Delta}}_{1}^{\intercal}\Lambda_{1}^{-2}\bm{\bar{\Delta}}_{1}\right)-\left(\tilde{\bm{\Delta}}_{1}^{\intercal}\Lambda_{1}^{-1}G_{1}(\tau)-\frac{\pi_{1}\tau}{2}\bm{\tilde{\Delta}}_{1}^{\intercal}\Lambda_{1}^{-2}\bm{\tilde{\Delta}}_{1}\right)\\
 & \quad+\left(\bm{\bar{\Delta}}_{0}^{\intercal}\Lambda_{0}^{-1}G_{0}(\tau)-\frac{\pi_{0}\tau}{2}\bar{\bm{\Delta}}_{0}^{\intercal}\Lambda_{0}^{-2}\bm{\bar{\Delta}}_{0}\right)-\left(\bm{\tilde{\Delta}}_{0}^{\intercal}\Lambda_{0}^{-1}G_{0}(\tau)-\frac{\pi_{0}\tau}{2}\bm{\tilde{\Delta}}_{0}^{\intercal}\Lambda_{0}^{-2}\bm{\tilde{\Delta}}_{0}\right)\\
 & =\sigma\eta x(\tau)-\frac{\eta^{2}\sigma^{2}}{2}\tau,
\end{align*}
where the last step follows from some algebra after making use of
(\ref{eq:pf:Lem2:1}) and (\ref{eq:pf:lem2-2}). Based on the above,
an application of the Neyman-Pearson lemma shows that the UMP test
of $H_{0}^{\prime}:\bm{\Delta}=\tilde{\bm{\Delta}}$ vs $H_{1}^{\prime}:\bm{\Delta}=\bar{\bm{\Delta}}$
is given by 
\[
\varphi_{\mu}^{*}=\mathbb{I}\left\{ \sigma\eta x(\tau)-\frac{\eta^{2}\sigma^{2}}{2}\tau\ge\tilde{\gamma}\right\} =\mathbb{I}\left\{ \mu x(\tau)-\frac{\mu^{2}}{2\sigma}\tau\ge\gamma\right\} .
\]
Here, $\gamma$ is to be determined by the size requirement. Now,
for any $\bm{\Delta}\in H_{0}$, 
\[
x(t)\equiv\frac{1}{\sigma}\left(\sqrt{\frac{\sigma_{1}^{2}}{\pi_{1}}}(1,0)^{\intercal}W_{1}(t)-\sqrt{\frac{\sigma_{0}^{2}}{\pi_{0}}}(1,0)^{\intercal}W_{0}(t)\right)\sim\tilde{W}(t),
\]
where $\tilde{W}(\cdot)$ is standard 1-dimensional Brownian motion.
Hence, the distribution of the sample paths of $x(\cdot)$ is independent
of $\bm{\Delta}$ under the null. Combined with the assumption that
$\tau$ is $\mathcal{F}_{t}$-adapted, this implies $\varphi_{\mu}^{*}$
does not depend on $\tilde{\bm{\Delta}}$ and, by extension, $\bar{\bm{\Delta}}$,
except through $\mu$. Since $\bar{\bm{\Delta}}\in H_{1}$ was arbitrary,
we are led to conclude $\varphi_{\mu}^{*}$ is UMP more generally
for testing $H_{0}\equiv\{\bm{\Delta}:(1,0)^{\intercal}\bm{\Delta}_{1}-(1,0)^{\intercal}\bm{\Delta}_{0}=0\}$
vs $H_{1}\equiv\{\bm{\Delta}:(1,0)^{\intercal}\bm{\Delta}_{1}-(1,0)^{\intercal}\bm{\Delta}_{0}=\mu\}$.
\end{proof}

\subsection{Supporting results for the proof of Proposition \ref{Prop: Non-parametric batched experiments}}

\begin{lem} \label{Lem: Batched experiments optimal test} Consider
the limit experiment where one observes $q_{a}=\sum_{j}\pi_{j}^{(a)}$
and $x_{a}:=(1,0)^{\intercal}\sum_{j}\bm{Z}_{j}^{(a)}(\pi_{j}^{(a)})$,
where
\[
\bm{Z}_{j}^{(a)}(t):=\bm{h}_{a}t+W_{j}^{(a)}(t),
\]
and $\pi_{j}$ is measurable with respect to 
\[
\mathcal{F}_{j-1}\equiv\sigma\left\{ (1,0)^{\intercal}\bm{Z}_{l}^{(a)}(\cdot);l\le j-1,a\in\{0,1\}\right\} .
\]
Then, the optimal level-$\alpha$ test of $H_{0}:\left((1,0)^{\intercal}\bm{h}_{1},(1,0)^{\intercal}\bm{h}_{0}\right)=(0,0)$
vs $H_{1}:\left((1,0)^{\intercal}\bm{h}_{1},(1,0)^{\intercal}\bm{h}_{0}\right)=(\mu_{1},\mu_{0})$
in the limit experiment is
\[
\varphi_{\mu_{1},\mu_{0}}^{*}=\mathbb{I}\left\{ \sum_{a\in\{0,1\}}\left(\mu_{a}x_{a}-\frac{q_{a}}{2}\mu_{a}^{2}\right)\ge\gamma_{\mu_{1},\mu_{0}}\right\} .
\]
 \end{lem} 
\begin{proof}
Denote 
\begin{align*}
H_{0} & \equiv\left\{ \bm{h}:\left((1,0)^{\intercal}\bm{h}_{1},(1,0)^{\intercal}\bm{h}_{0}\right)=(0,0)\right\} ,\ \textrm{ and}\\
H_{1} & \equiv\left\{ \bm{h}:\left((1,0)^{\intercal}\bm{h}_{1},(1,0)^{\intercal}\bm{h}_{0}\right)=(\mu_{1},\mu_{0})\right\} .
\end{align*}
Let $\mathbb{P}_{h}$ denote the induced probability measure over
the sample paths generated by $\{z_{j}^{(a)}(t):t\le\pi_{j}^{(a)}\}_{j,a}$. 

Given any $(\bm{h}_{1},\bm{h}_{0})\in H_{1}$, define $\tilde{\bm{h}}_{a}=\bm{h}_{a}-(1,0)^{\intercal}\bm{h}_{a}(1,0)$
for $a\in\{0,1\}$. Note that $(\tilde{\bm{h}}_{1},\tilde{\bm{h}}_{0})\in H_{0}$
and $(1,0)^{\intercal}\bm{h}_{a}=\mu_{a}$. Let 
\[
\ln\frac{d\mathbb{P}_{(\tilde{\bm{h}}_{1},\tilde{\bm{h}}_{0})}}{d\mathbb{P}_{(\bm{h}_{1},\bm{h}_{0})}}(\mathcal{G})
\]
denote the likelihood ratio between the probabilities induced by the
parameters $(\tilde{\bm{h}}_{1},\tilde{\bm{h}}_{0}),(\bm{h}_{1},\bm{h}_{0})$
over the filtration 
\[
\mathcal{\mathcal{G}}\equiv\sigma\left\{ \bm{Z}_{j}^{(a)}(t):t\le\pi_{j}^{(a)};j=1,\dots,J;a\in\{0,1\}\right\} .
\]
By the Girsanov theorem, noting that $\{z_{j}^{(a)}(t):t\le\pi_{j}^{(a)}\}_{j}$
are independent across $a$ and defining $G_{a}:=\sum_{j}\bm{Z}_{j}^{(a)}(\pi_{j}^{(a)})$,
we obtain after some straightforward algebra that 
\begin{align*}
\ln\frac{d\mathbb{P}_{(\tilde{\bm{h}}_{1},\tilde{\bm{h}}_{0})}}{d\mathbb{P}_{(\bm{h}_{1},\bm{h}_{0})}}(\mathcal{F}) & =\sum_{a}\left\{ \left(\tilde{\bm{h}}_{a}^{\intercal}G_{a}-\frac{q_{a}}{2}\tilde{\bm{h}}_{a}^{\intercal}\tilde{\bm{h}}_{a}\right)-\left(\bm{h}_{a}^{\intercal}G_{a}-\frac{q_{a}}{2}\bm{h}_{a}^{\intercal}\bm{h}_{a}\right)\right\} \\
 & =\sum_{a}\left(\mu_{a}x_{a}(\tau)-\frac{\mu_{a}^{2}}{2}q_{a}\right),
\end{align*}
where $x_{a}$ is the first component of $G_{a}$. Hence, an application
of the Neyman-Pearson lemma shows that the UMP test of $H_{0}^{\prime}:\bm{h}=(\tilde{\bm{h}}_{1},\tilde{\bm{h}}_{0})$
vs $H_{1}^{\prime}:\bm{h}=(\bm{h}_{1},\bm{h}_{0})$ is given by 
\[
\varphi_{\mu_{1},\mu_{0}}^{*}=\mathbb{I}\left\{ \sum_{a}\left(\mu_{a}x_{a}(\tau)-\frac{\mu_{a}^{2}}{2}q_{a}\right)\ge\gamma\right\} ,
\]
where $\gamma$ is determined by the size requirement. 

Now, for any $\bm{h}\in H_{0}$, both $x_{a}$ and $q_{a}$ measurable
with respect to $\mathcal{F}$ by assumption. Since $(1,0)^{\intercal}\bm{Z}_{j}^{(a)}(\cdot)$
is independent of $\bm{h}_{a}$ given $\mu_{a}$ for all $j,a$, it
follows that the distribuion of $x_{a},q_{a}$ is independent of $\bm{h}\in H_{0}$
under the null. This implies that $\varphi_{\mu_{1},\mu_{0}}^{*}$
does not depend on $(\tilde{\bm{h}}_{1},\tilde{\bm{h}}_{0})$ and,
by extension, $(\bm{h}_{1},\bm{h}_{0})$, except through $(\mu_{1}\mu_{0})$.
Since $(\bm{h}_{1},\bm{h}_{0})\in H_{1}$ was arbitrary, we are led
to conclude $\varphi_{\mu_{1},\mu_{0}}^{*}$ is UMP more generally
for testing the composite hypotheses $H_{0}$ vs $H_{1}$. 
\end{proof}

\end{document}